\documentclass[superscriptaddress, preprintnumbers,twocolumn, nofootinbib,amsmath,amssymb, aps, prd]{revtex4-2}
\usepackage{mathrsfs}
\usepackage{graphicx}
\usepackage{dcolumn}
\usepackage{bm}
\renewcommand{\labelenumi}{(\roman{enumi})}
\usepackage{comment}
\usepackage{braket}
\usepackage{amsthm}
\usepackage{tikz}
\usetikzlibrary{positioning}
\usepackage{color,soul}
\usepackage[hyperfootnotes=false]{hyperref}
\hypersetup{
    colorlinks=true,
    linkcolor=blue,
    filecolor=magenta,      
    urlcolor=cyan,
}
\newtheorem{theorem}{Theorem}
\newtheorem{lemma}[theorem]{Lemma}
\newtheorem{claim}{Claim}

\begin{document}

\preprint{CHIBA-EP-248, 2021.05.16}

\title{Reconstructing propagators of confined particles in the presence of complex singularities
}

\author{Yui Hayashi}
\email{yhayashi@chiba-u.jp}
\affiliation{
Department of Physics, Graduate School of Science and Engineering, Chiba University, Chiba 263-8522, Japan
}

\author{Kei-Ichi Kondo}
\email{kondok@faculty.chiba-u.jp}
\affiliation{
Department of Physics, Graduate School of Science and Engineering, Chiba University, Chiba 263-8522, Japan
}
\affiliation{
Department of Physics, Graduate School of Science, Chiba University, Chiba 263-8522, Japan
}

\begin{abstract}
Propagators of confined particles, especially the Landau-gauge gluon propagator, may have complex singularities as suggested by recent numerical works as well as several theoretical models, e.g., motivated by the Gribov problem.
In this paper, we study formal aspects of propagators with complex singularities in reconstructing Minkowski propagators starting from Euclidean propagators by the analytic continuation.
We derive the following properties rigorously for propagators with arbitrary complex singularities satisfying some boundedness condition.
The two-point Schwinger function with complex singularities violates the reflection positivity.
In the presence of complex singularities, while the holomorphy in the usual tube is maintained, the reconstructed Wightman function on the Minkowski spacetime becomes a non-tempered distribution and violates the positivity condition.
On the other hand, the Lorentz symmetry and locality are kept intact under this reconstruction.
Finally, we argue that complex singularities can be realized in a state space with an indefinite metric and correspond to confined states.
We also discuss consequences of complex singularities in the BRST formalism.
Our results could open up a new way of understanding a confinement mechanism, mainly in the Landau-gauge Yang-Mills theory.
\end{abstract}

\maketitle

\section{INTRODUCTION}

One of the most fundamental properties of strong interactions is color confinement, the absence of colored degrees of freedom from the physical spectrum. Understanding this property in the framework of relativistic quantum field theory (QFT) is a long-standing problem and of crucial importance for particle and nuclear physics. Analytic structures of the correlation functions enable us to extract valuable information on the state-space structure through, e.g., the K\"all\'en-Lehmann spectral representation \cite{spectral_repr_UKKL}, which will be useful toward understanding a confinement mechanism.
Therefore, investigating analytic structures of confined propagators, e.g., the gluon propagator, and considering their implications are of great interest.

In the last decades,  the gluon, ghost, and quark propagators in the Landau gauge have been extensively studied by both Lattice numerical simulations and semi-analytical methods (for example, Dyson-Schwinger equation and Functional renormalization group), for reviews see \cite{Alkofer:2000wg, Huber:2018ned, Maas13}, and also by models motivated by the massive-like gluon propagator of these results \cite{TW10,TW11,PTW14}.
Based on these advances, in recent years, there has been an increasing interest in the analytic structures of the gluon, ghost, and quark propagators \cite{Alkofer:2003jj, SFK12, HFP14, Siringo16a, Siringo16b, DOS17, Lowdon17, Lowdon18, Lowdon:2018mbn, CPRW18, HK2018, DORS19, KWHMS19, BT2019, LLOS20, HK2020, Fischer-Huber, Falcao:2020vyr,Horak:2021pfr}.
In particular, unusual singularities invalidating the K\"all\'en-Lehmann spectral representation, which we call \textit{complex singularities}, receive much attention.
A pair of complex conjugate poles of the gluon propagator, which is a typical example of such singularities, were predicted in old literature \cite{Gribov78, Zwanziger89, Stingl85, HKRSW90, Stingl96, Dudal:2008sp} , e.g., by improving the gauge fixing procedure.
The most remarkable point of the recent studies without assuming the K\"all\'en-Lehmann representation is that the independent approaches consisting of reconstructing from Euclidean data \cite{BT2019, Falcao:2020vyr}, modeling by the massive-like gluons \cite{Siringo16a, Siringo16b, HK2018, HK2020}, and ray-technique of the Dyson-Schwinger equation \cite{SFK12,Fischer-Huber} consistently suggest the existence of complex singularities of the gluon propagator. Moreover, some results support complex poles of the quark propagator \cite{HK2020}.

There are also studies of complex singularities on other models \cite{Maris:1991cb, Maris:1994ux, Maris:1995ns}. A relation between complex poles of a fermion propagator and confinement in the three-dimensional QED was suggested in \cite{Maris:1995ns}.

Since complex singularities cannot appear in propagators for observable particles, we expect that the complex singularities are related to color confinement. 
However, while the analytic structures have been investigated in many works, implications of complex singularities for the QFT have been much less studied.
Theoretical consequences of complex singularities are of crucial importance since such considerations on complex singularities could play a pivotal role in obtaining a clear description of a confinement mechanism.
Thus, we will study theoretical aspects of complex singularities in this paper.

For this purpose, the reconstruction of the two-point Wightman function, or the vacuum expectation value of the product of field operators, from the two-point Schwinger function, or the Euclidean propagator, has to be carefully investigated.
Thus, we will reconstruct the Wightman function based on the holomorphy of the Wightman function in ``the tube'' \cite{Streater:1989vi} following the Osterwalder-Schrader (OS) reconstruction \cite{OS73, OS75}.
This is the standard method to relate Euclidean field theories to QFTs in axiomatic quantum field theory.

Some argue that the appearance of complex singularities might indicate non-locality, e.g. \cite{Stingl85,HKRSW90, Stingl96}. Nevertheless, this argument relying on the naive inverse Wick-rotation is not fully convincing.
Actually, as we will briefly remark in this paper, the naive inverse Wick-rotation differs from the reconstruction based on the holomorphy of the Wightman function in the presence of complex singularities.
Since the relation between complex singularities and locality is thus in a confusing situation, we will also address this topic carefully.

 \begin{figure}[t]
  \begin{center}
  \includegraphics[width= \linewidth]{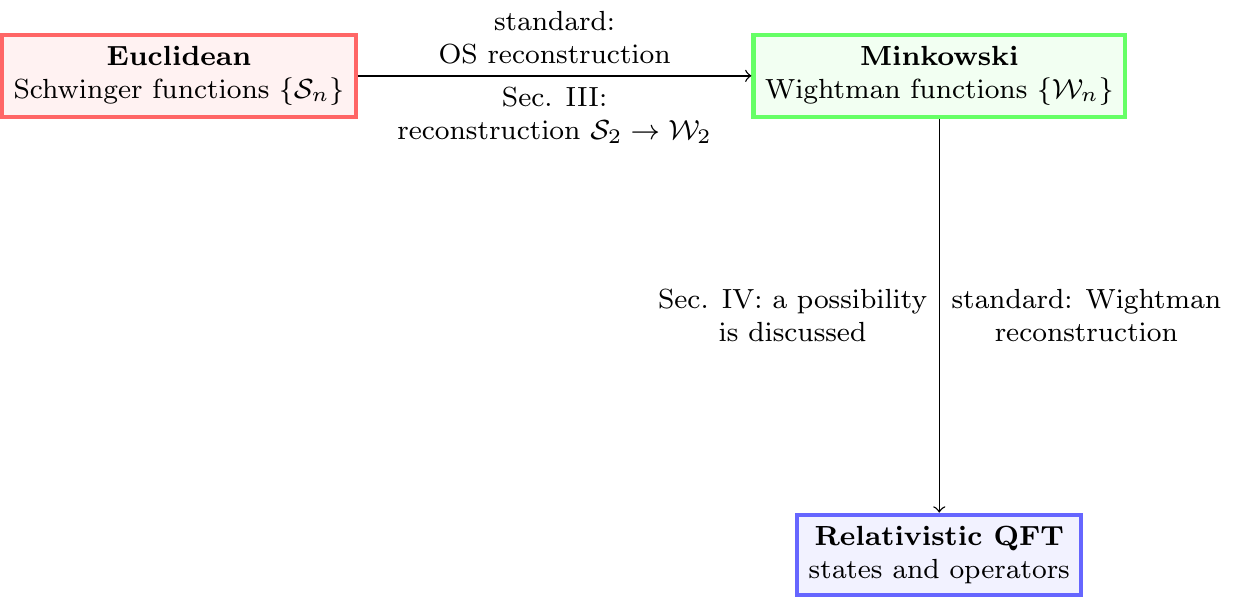}
  \end{center}
   \caption{The reconstruction procedure and contents of this paper. In the standard reconstruction procedure, we start from a family of Schwinger functions satisfying OS axioms and finally reconstruct a QFT by the OS theorem \cite{OS73, OS75} and Wightman's reconstruction theorem \cite[Theorem 3-7]{Streater:1989vi}. We re-examine this reconstruction procedure when a propagator has complex singularities.
   In Sec.~\ref{sec:prelim_disc}, it is pointed out that we shall begin with a Schwinger function with complex singularities. In Sec.~\ref{sec:complex-singularities}, we reconstruct a Wightman function from the Schwinger function in the same way as the OS reconstruction based on the holomorphy in the tube. In Sec.~\ref{sec:interpretation}, we discuss a possibility in the reconstruction procedure from the Wightman functions to a QFT.}
    \label{fig:introduction}
\end{figure}

In this paper, we study formal aspects of complex singularities, namely analytic properties of the reconstructed two-point Wightman function and implications of complex singularities for the state-space structure.
The standard reconstruction procedure and contents of this paper are illustrated in Fig.~\ref{fig:introduction}.
Because of the somewhat confusing situation on this subject as mentioned above, it is essential to clarify consequences of complex singularities that can be stated unambiguously.
Thus, we will derive these analytic properties with rigorous proofs.
Moreover, since it is very important to investigate states related to the confined particles for understanding a confinement mechanism, we will consider state-space structures yielding complex singularities.

The main results of this paper are listed as follows, as announced in \cite{HK2021}.
Suppose that the Euclidean propagator, or the two-point Schwinger function, has complex singularities in the complex squared momentum plane, as defined in \ref{sec:complex-def}. Then, the following claims are derived.
\begin{enumerate}
\renewcommand{\labelenumi}{(\Alph{enumi})}
    \item The \textit{reflection positivity} is violated for the Schwinger function (Theorem \ref{thm:violation_ref_pos}).
    \item The \textit{holomorphy} of the Wightman function $W(\xi - i \eta)$ in the tube (Theorem \ref{thm:holomorphy_general}) and the \textit{existence of the boundary value} as a distribution (Theorem \ref{thm:complex_boundary_value}) are still valid. Thus, we can reconstruct the Wightman function from the Schwinger function.
    \item The \textit{temperedness} (Theorem \ref{thm:nontempered}) and the \textit{positivity condition} in $\mathscr{D}(\mathbb{R}^4)$ (Theorem \ref{thm:violation_W-pos}) are violated for the reconstructed Wightman function. The \textit{spectral condition} is never satisfied since it requires the temperedness as a prerequisite.
    \item The \textit{Lorentz symmetry} (Theorem \ref{thm:Lorentz} and Theorem \ref{thm:complex_Lorentz}) and \textit{spacelike commutativity} (Theorem \ref{thm:spacelike_commutativity}) are kept intact.
    \item A quantum mechanical observation (Claim \ref{claim:3}) suggests, together with an example of QFT (Sec.~\ref{sec:Lee-Wick}), that complex singularities correspond to pairs of zero-norm eigenstates of complex eigenvalues.
\end{enumerate}

This paper is organized as follows.
In Sec.~\ref{sec:prelim_disc}, we emphasize difference between complex singularities of Euclidean propagator and (real-)time-ordered one in the momentum space and take a glimpse of some properties to be generally derived in Sec.~\ref{sec:complex-properties}.
In Sec.~\ref{sec:complex-singularities}, we give a definition of complex singularities (Sec.~\ref{sec:complex-def}) and derive the properties (Sec.~\ref{sec:complex-properties}) listed above with a mathematical rigor except for the last one (E).
In Sec,~\ref{sec:interpretation}, based on the results of Sec.~\ref{sec:complex-singularities}, we consider quantum-theoretical aspects, namely what complex singularities imply on the state-space structure. We also discuss implications of complex singularities in the BRST formalism.
A summary is given in Sec.~\ref{sec:summary}, and Sec.~\ref{sec:discussion} is devoted to discussion on related topics and future prospects.
The mathematical notations and standard axioms are summarized in Appendix \ref{sec:notations_axioms}.
Appendix \ref{sec:Appendix_ref_pos} contains a detailed proof of the violation of the reflection positivity (Theorem \ref{thm:violation_ref_pos}).
Appendix \ref{sec:Appendix_which_axioms} summarizes violated axioms of the OS axioms for Schwinger functions and the Wightman axioms for Wightman functions.

\section{Preliminary discussion} \label{sec:prelim_disc}

In this section, we sketch out main properties of complex singularities and emphasize the difference between complex singularities of Euclidean propagator and (real\nobreakdash-) time-ordered one in the momentum space. For simplicity, we consider $(0+1)$-dimensional field theories in this section. This non-rigorous discussion helps us to determine a point of departure toward the rigorous discussion in Sec.~\ref{sec:complex-singularities}.

\subsection{Difference between complex singularities of Euclidean propagator and (real-)time-ordered one}

We consider complex singularities of Euclidean and real-time propagators on the complex squared momentum plane.
We point out that the conventional Wick rotation in the squared momentum plane $p^2 \rightarrow - p^2_E$ is \textit{not applicable} in the presence of complex singularities.
Thus, we emphasize that complex singularities in the propagators that appear in many works should be regarded as Euclidean ones and that the reconstruction procedure must be carefully considered.

We define the ``Wightman functions'' $D^>(t)$ and $D^<(t)$ and the real-time propagator $D(t)$ by
\begin{align}
&D^>(t) := \braket{0|\phi(t) \phi(0)|0}, \notag \\
&D^<(t) := \braket{0|\phi(0) \phi(t)|0}, \notag \\
&D(t) := \theta (t) D^>(t) + \theta (-t) D^<(t).
\end{align}
Usually, we can analytically continue $D^>(t)$ and $D^<(t)$ to the lower and upper half planes of the complex $t$ plane, respectively. In particular, $D^>(-i \tau)$ can be defined for $\tau > 0$, and $D^<(-i \tau)$ for $\tau < 0$.

Thus, we introduce the Euclidean propagator $\Delta(\tau)$, which is identified with the ``two-point Schwinger function,'' as
\begin{align}
    \Delta^{>}(\tau) &:= D^>( - i \tau) ~~~ (\mathrm{for~}\tau > 0), \notag \\
    \Delta^{<}(\tau) &:= D^<( - i \tau) ~~~ (\mathrm{for~}\tau < 0), \notag \\
\Delta(\tau) &:= \theta(\tau) \Delta^>( \tau) + \theta(-\tau) \Delta^< (\tau) 
    .\label{eq:connection_Euc_QFT}
\end{align}
This connection between the Wightman and Schwinger functions is consistent with the standard reconstruction method given in (\ref{eq:W_S_connection_A}) and (\ref{eq:W_S_connection_B}), where the Schwinger function is regarded as the ``values'' of the Wightman function at pure imaginary times.
We denote the Fourier transforms of $D(t)$ and $\Delta(\tau)$ by $ \tilde{D}(p_0)$ and $\tilde{\Delta}(p_E)$, respectively.

We emphasize that the connection between Euclidean correlation functions and vacuum expectation values of the product of field operators should be implemented in the complex-time plane rather than in the complex squared momentum plane.
Here, with the connection ({\ref{eq:connection_Euc_QFT}}), we demonstrate that the reconstructed propagator $D(t)$ cannot have a well-defined Fourier transform if $\tilde{\Delta}(p_E)$ has complex poles. This indicates that a real-time propagator with complex poles (where $ \tilde{D}(p_0)$ has complex poles) is not the reconstructed propagator from a Euclidean propagator with complex poles (where $ \tilde{\Delta}(p_E)$ has complex poles).

\subsubsection{Physical case}

First, we observe the physical case for a comparison. Let us assume as a definition of the ``physical case'',
\begin{enumerate}
    \item completeness: $1 = \sum_n  \ket{n} \bra{n}$, where $\ket{n}$ is an eigenstate of the Hamiltonian $H$ with an eigenvalue $E_n$: $H \ket{n} = E_n \ket{n}$, 
    \item translational covariance: $\phi(t) = e^{iHt} \phi(0) e^{-iHt}$,
    \item spectral condition: positivity of $H$, namely $E_n \geq 0$.
\end{enumerate}
Then, one can relate Euclidean and real-time propagators $\tilde{\Delta}(p_E)$ and $ \tilde{D}(p_0)$ by the conventional Wick rotation $p_0^2 \rightarrow - p_E^2$.
Indeed, these three conditions yield the spectral representations for the Wightman functions and the real-time propagator,
\begin{align}
    D^{>}(t) &= \int_0^\infty d \sigma ~ e^{-i \sigma t} \rho(\sigma), \notag \\
     D^{<}(t) &= \int_0^\infty d \sigma ~ e^{i \sigma t} \rho(\sigma), \notag \\
    \tilde{D}(p_0) &= i \int d \sigma~ \frac{2 \sigma \rho(\sigma)}{ p_0^2 - \sigma^2 + i \epsilon},
\end{align}
where we have defined the spectral function $\rho(\sigma)$ by
\begin{align}
\rho(\sigma) := \sum_n \delta(\sigma - E_n) |\braket{n|\phi(0)|0}|^2.
\end{align}
Consequently, from (\ref{eq:connection_Euc_QFT}), the Euclidean propagator has the spectral representation given by
\begin{align}
    \Delta^{>}(\tau) &= D^>( - i \tau) =  \int_0^\infty d\sigma~ e^{- \sigma \tau} \rho(\sigma), \notag \\
     \Delta^{<}(\tau) &= D^<( - i \tau) = \int_0^\infty d\sigma~ e^{ \sigma \tau} \rho(\sigma), \notag \\
    \tilde{\Delta}(p_E) &= \int d \sigma~ \frac{2 \sigma \rho(\sigma)}{p_E^2 + \sigma^2}.
\end{align}
Therefore, in the physical case, the Euclidean propagator $\tilde{\Delta}(p_E)$ and the real-time propagator $\tilde{D}(p_0)$ are related by the analytic continuation on the complex squared momentum plane: $ p^2_0 \rightarrow - p_E^2$. The spectral representation guarantees this consequence, which does not hold in the presence of complex singularities as will be shown below.

\subsubsection{With complex poles}

For example, let us take the Gribov-type propagator with complex poles:
\begin{align}
    \tilde{\Delta}(p_E) := \frac{p_E^2}{p_E^4 + \gamma^4}.
\end{align}
This gives the following Euclidean propagator in the Euclidean time:
\begin{align}
    \Delta(\tau) = \frac{1}{2 \gamma} e^{-\frac{\gamma |\tau|}{\sqrt{2}}} \sin \left( -\frac{\gamma |\tau|}{\sqrt{2}} + \frac{\pi}{4} \right),
\end{align}

Although a complete reconstruction method from Euclidean to Minkowski in the presence of complex singularities has not been established, we here assume the connection introduced in (\ref{eq:connection_Euc_QFT}) which is consistent with the standard reconstruction method even in the presence of complex singularities.
With this connection, we have the Wightman functions:
\begin{align}
    D^>(t) = D^<(-t) = \frac{i}{2 \gamma} e^{i \frac{\gamma t}{\sqrt{2}}} \sinh \left(  \frac{\gamma t}{\sqrt{2}} - \frac{i \pi}{4} \right),
\end{align}
Then, both $D^>(t)$ and $D^<(t)$ increase exponentially as $t \rightarrow \pm \infty$.

Therefore, starting with the Gribov-type Euclidean propagator, we have the Wightman functions $D^>(t),~
D^<(t)$ of exponential growth. \textit{Such Wightman functions $D^>(t),~
D^<(t)$ cannot be regarded as tempered distributions}, and therefore they do not have well-defined Fourier transforms. Thus, the Minkowski propagator \textit{cannot} be reconstructed from the Euclidean propagator with complex poles by using the simple ``inverse Wick rotation'' $p_E^2 \rightarrow -p_0^2$ in the complex squared momentum plane, since the ``reconstructed'' real-time propagator has no Fourier transform. In other words, \textit{a Euclidean propagator with complex poles (where $\tilde{\Delta} (p_E)$ has complex poles) is different from a real-time propagator with complex poles (where $\tilde{D} (p_0)$ has complex poles)}.
In particular, one has to take care of the definition of complex singularities.

Again, one should reconstruct the propagator
not by the simple inverse Wick rotation on the complex squared momentum plane: $ p^2_E \rightarrow - p_0^2$ but by the standard method explained in (\ref{eq:W_S_connection_A}) and (\ref{eq:W_S_connection_B}) .
The former reconstruction is often discussed in some literature, e.g. in \cite{Stingl85,Stingl96,HKRSW90}, which is \textit{different} from the latter one.
As more discussed in Sec.~V~A, we argue that the latter one should be adopted because of the fundamental relation (\ref{eq:W_S_connection_A}) and some advantages.

\subsection{Properties} \label{sec:2B-properties}

Let us briefly summarize properties of complex poles. Here we suppose that the Euclidean propagator $\tilde{\Delta} (p_E)$ has complex poles. 

\begin{enumerate}
\renewcommand{\labelenumi}{(\alph{enumi})}
    \item The Wightman functions $D^>(t)$ and $D^<(t)$ reconstructed from the Euclidean propagator $\Delta (\tau)$ cannot be regarded as tempered distributions because they grow exponentially as $t \rightarrow \pm \infty$.
    \item A Euclidean propagator with only complex poles violates the reflection positivity (\ref{eq:two_pt_ref_pos}) because $\Delta (\tau)$ violates the necessary condition for the reflection positivity (\ref{eq:two_pt_ref_pos_concise}): $\Delta (\tau) \geq 0 ~ \mathrm{for~all}~\tau >0$.
    \item The positivity in the sector $\{ \phi(t) \ket{0} \}_{t \in \mathbb{R}}$ is violated due to the non-temperedness.
    Indeed, suppose that the sector $\{ \phi(t) \ket{0} \}_{t \in \mathbb{R}}$ had a positive metric. From the translational invariance of the two-point function, the time-translation operator defined on this sector: $U(s) \phi(t) \ket{0} := \phi(t+s) \ket{0}$ is unitary, i.e., $\braket{0|\phi(t) U(s)^\dagger U(s) \phi(t')|0} = \braket{0|\phi(t) \phi(t')|0}$. Since the modulus of a matrix element of a unitary operator is not more than one in a space with a positive metric, we would have an upper bound $|\braket{0|\phi(0) U(s) \phi(0)|0} | \leq \braket{0|\phi(0) \phi(0)|0}$, or $|D^<(s)| \leq |D^<(0)|$, which contradicts the non-temperedness.
\end{enumerate}

In the next section, we see that these properties always hold rigorously if $\tilde{\Delta} (p_E)$ has complex singularities (Theorem \ref{thm:nontempered}, \ref{thm:violation_ref_pos}, and \ref{thm:violation_W-pos}).

\section{Complex singularities: Definition and properties} \label{sec:complex-singularities}

In this section, we give a definition of complex singularities and rigorous proofs of some properties for propagators.
These ``complex singularities'' should be regarded as complex singularities on the complex squared momentum plane of an analytically continued Euclidean propagator. Indeed, in many studies, propagators with complex poles are compared with numerical results on Euclidean ones.
Therefore, we start with a two-point Schwinger function. For details of mathematical notations, see Appendix \ref{sec:notations_axioms}.

For simplicity, we work in four-dimensional Euclidean space $D = 4$. However, our main results can be easily generalized to arbitrary dimensions $D \geq 2$ except for Theorem \ref{thm:complex_Lorentz}, where the Bargmann–Hall–Wightman theorem must be used for the proof.

\subsection{Definition} \label{sec:complex-def}

\subsubsection{Preliminary assumptions}

For simplicity, we consider a two-point function for a scalar field. Throughout this paper, we assume the following conditions for a two-point Schwinger function $\mathcal{S}_2  (x_1,x_2)$ which follow from the OS axioms \cite{OS73, OS75} (see Appendix \ref{sec:notations_axioms}).
\begin{enumerate}
    \item ~[OS0] Temperedness $\mathcal{S}_2 (x_1,x_2) \in  ~^0 \mathscr{S}' (\mathbb{R}^{4 \cdot 2})$: $\mathcal{S}_2 (x_1,x_2)$ is a tempered distribution defined on the space of test functions vanishing at coincident points $x_1 = x_2$.
    \item ~[OS1] Euclidean (translational and rotational) invariance:  
     $\mathcal{S}_2 (R x_1 + a, R x_2 + a) = \mathcal{S}_2(x_1 , x_2)$, for all $a \in \mathbb{R}^4,~ R \in SO(4)$.
\end{enumerate}
From (i) temperedness and (ii) translational invariance, there exists a distribution $S_{1}(\xi) \in \mathscr{S}' (\mathbb{R}^{4}_+)$ such that $\mathcal{S}_2 (x_1, x_2) = S_1 (x_2 - x_1)$ for $x_1^4 < x_2^4$. We can also regard $S_{1}(\xi) \in \mathscr{S}' (\mathbb{R}^{4}_{\neq 0})$, where $ \mathscr{S}' (\mathbb{R}^{4}_{\neq 0})$ is the dual space of
\begin{align}
    \mathscr{S} (\mathbb{R}^{4}_{\neq 0}) := \Bigl\{ f(\xi) \in \mathscr{S} (\mathbb{R}^{4}) ~;~
    \begin{gathered}
        \left. D^\alpha f(\xi) \right|_{\xi = 0} = 0~~ \\
        ~~\mathrm{for~any}~ \alpha \in \mathbb{Z}_{\geq 0}^{4} 
    \end{gathered}
    \Bigr\}.
\end{align} Moreover, (ii) Euclidean rotational invariance implies $S_{1}(R \xi) = S_{1}(\xi)$ for all $R \in SO(4)$.

Let us comment on the other conditions of the standard OS axioms \cite{OS73,OS75} (see Appendix A). They are [OS2] reflection positivity, [OS3] permutation symmetry, [OS4] cluster property, and [OS0'] Laplace transform condition.
Intuitively, [OS2] reflection positivity corresponds to the positivity of the metric of the state space. If we consider gauge theories in Lorentz covariant gauges including confined degrees of freedom, we must allow violation of the reflection positivity. Thus, we do not require the reflection positivity, which is actually broken in the presence of complex singularities (Theorem~\ref{thm:violation_ref_pos}).
For a two-point function, [OS3] permutation symmetry is a consequence of [OS1] Euclidean rotational invariance.
For generality, we do not impose [OS4] cluster property, which corresponds to the uniqueness of the vacuum and could be violated by a severe infrared singularity of a propagator.
In the view of the reconstruction from Euclidean field theories, [OS0'] Laplace transform condition is introduced for the purpose of controlling higher point functions. Since we focus on the two-point function in this paper, we will not take a further look into this condition.
Incidentally, the Laplace transform condition itself is violated if the two-point function has complex singularities due to the non-temperedness of the Wightman functions (Theorem \ref{thm:nontempered}).

In addition to the assumptions taken from the standard OS axiom, we further require that the two-point Schwinger function $S_{1}(\xi) \in \mathscr{S}' (\mathbb{R}^{4}_{\neq 0})$ has a well-defined Fourier transform $S_{1} (k)$. Simply, this can be realized by the following assumption:
\begin{enumerate}
\setcounter{enumi}{2} 
    \item The Schwinger function $S_{1}(\xi) $ can be regarded as an element of $\mathscr{S}' (\mathbb{R}^{4})$: $S_{1}(\xi) \in \mathscr{S}' (\mathbb{R}^{4})$.
\end{enumerate}
This assumption allows the well-defined Fourier transform,
\begin{align}
    S_{1} (k)  = \int d^4 \xi ~ e^{-ik \xi} S_{1}(\xi).
\end{align}
From the rotational invariance, we can write \footnote{Note the difference of conventions with our previous papers \cite{HK2018, KWHMS19, HK2020}, where we took $S_{1} (k) = D(-k^2)$. In particular, the timelike axis is the negative real axis in this paper unlike the previous ones. See Fig.~\ref{fig:complex.pdf}}.
\begin{align}
    S_{1} (k) = D(k^2).
\end{align}
A few remarks are in order.
\begin{enumerate}
\renewcommand{\labelenumi}{(\alph{enumi})}

\item While the condition $S_{1}(\xi) \in \mathscr{S}' (\mathbb{R}^{4}_{\neq 0})$ allows any singularity at $\xi = 0$, the new condition (iii) $S_{1}(\xi) \in \mathscr{S}' (\mathbb{R}^{4})$ imposes that such singularity is of at most derivatives of a delta function $D^\alpha \delta(\xi)$. We do not expect appearance of singularities beyond usual distributions at least in an ultraviolet asymptotic free theory.

\item For real-valued fields, namely real-valued $S_1 (\xi)$, $S_{1} (k) = D(k^2)$ is a real distribution from the rotational symmetry (or the permutation symmetry) $S_1(-\xi) = S_1 (\xi)$.

\item There is a constraint on the massless singularities. For example, this formulation excludes the ``dipole ghost pole'': $D(k^2) \sim 1/k^4$ without a regularization since $ D(k^2) = S_{1} (k) \in \mathscr{S}' (\mathbb{R}^{4})$. This constraint depends on the spacetime dimension. The massless pole (without a regularization) is prohibited in the two-dimensional space. 
\end{enumerate}

\subsubsection{Definition of complex singularities}

Now, let us define complex singularities of a two-point Schwinger function.
We call the positive real axis of the complex $k^2$ plane the \textit{Euclidean (spacelike) axis} and call the negative real axis of the complex $k^2$ plane the \textit{timelike axis}.
In addition to (i)--(iii), we assume for $D(k^2)$,
\begin{enumerate}
\setcounter{enumi}{3} 
    \item $D(k^2) = S_{1} (k)$ is holomorphic except singularities on the timelike axis $\{ k^2; k^2 < 0 \}$ and a finite number of poles and branch cuts of finite length satisfying:
    \begin{enumerate}
    \renewcommand{\labelenumii}{(iv~\alph{enumii})}
        \item The singularities on the timelike axis can be represented as a tempered distribution on $[-\infty,0]$, namely,
\begin{align}
     &D(-\sigma^2 - i \epsilon) - D(-\sigma^2 + i \epsilon) \notag \\
     &~~\xrightarrow{\epsilon \rightarrow +0} \operatorname{Disc} D(-\sigma^2) \in \mathscr{S}'([0,\infty]),
\end{align}
        where $\mathscr{S}'([0,\infty])$ is the dual space of $\mathscr{S}([0,\infty]) := \bigl\{ f(\lambda) = g(- (1+\lambda)^{-1}) ~;~g\mathrm{~is~a}~C^\infty \mathrm{~function~on~}[-1,0] \bigr\}$. For details, see Appendix \ref{sec:notations_axioms} or \cite[Sec. A.3.]{Bogolyubov:1990kw}.

        \item $D(k^2) = S_{1} (k)$ is holomorphic at least in a neighborhood of the Euclidean axis $\{ k^2; k^2 > 0 \}$ in the sense that there is no singularity on the Euclidean axis.
        \item The complex branch cuts are not located across the real axis.
    \end{enumerate} 
    \item The analytically continued $D(z)$ on the complex plane $z = k^2$ tends to vanish as $|z| \rightarrow \infty$.
\end{enumerate}
With these assumptions (i) -- (v), we call singularities except on the negative real axis \textit{complex singularities}.

The first assumption (iv a) is imposed for a practical purpose. Without this condition, the spectral function would generally be a hyperfunction, which makes an analytical treatment difficult.
Due to this condition, the ``spectral'' integral: $\int_0 ^\infty \frac{\operatorname{Disc} D(-\sigma^2) }{k^2 + \sigma^2}$ (see Theorem \ref{thm:spec_repr_complex}) is well-defined.
The second assumption (iv b) excludes ``tachyonic singularities'', which could make $S_1(\xi)$ ill-defined.
The third one (iv c) claims that, except for the timelike singularities, there are no singularities in the vicinity of the real axis.
This is a technical assumption for defining the spectral function and also for simplifying the proof of Theorem \ref{thm:nontempered}.

Although the assumption (v) is a technical one\footnote{Note that discussion similar to the following one can be done for $D(z)$ of polynomial growth in $z$ as $|z| \rightarrow \infty$ by applying the Cauchy theorem to $D(z)/z^n$ in Theorem \ref{thm:spec_repr_complex}}, we expect that the gluon, ghost, and quark propagators satisfy this property due to the ultraviolet asymptotic freedom.
Indeed, in the Landau gauge, the QCD propagators have the asymptotic form of $D(k^2) \sim \frac{1}{k^2 (\ln |k^2|)^{\gamma_0/\beta_0}}$, where $\gamma_0$ and $\beta_0$ are respectively the first coefficients of the anomalous dimension and the beta function \cite{OZ80}.

The finiteness of branch cuts is required for the reconstruction of the Wightman function.
One could allow infinitely long branch cuts whose discontinuities are suppressed faster than any exponential decay as $|z| \rightarrow \infty$ and those which approach asymptotically to the negative real axis sufficiently fast.
We shall make a further comment on this point below.
For simplicity, we restrict ourselves to the case without branch cuts of infinite length in this paper.

Although we have restricted ourselves to poles and cuts at the assumption (iv), we note that one can easily generalize theorems in Sec.~{\ref{sec:complex-properties}}, i.e., Theorem 2 -- 11, to arbitrary complex singularities if the following conditions are satisfied: boundedness of locations in $|k^2|$, (iv a) regularity of the timelike singularities, (iv b, iv c) holomorphy in a neighborhood of the real axis except for the timelike axis, and (v) $D(k^2) \rightarrow 0$ as $k^2 \rightarrow \infty$. With these conditions, contributions from complex singularities can be represented as integrals along contours surrounding these singularities according to the Cauchy integral theorem. Then, we can use the same proofs described in Sec.~{\ref{sec:complex-properties}} for this generalization.

\subsubsection{Generalized spectral representation}

 \begin{figure}[t]
  \begin{center}
   \includegraphics[width=0.9\linewidth]{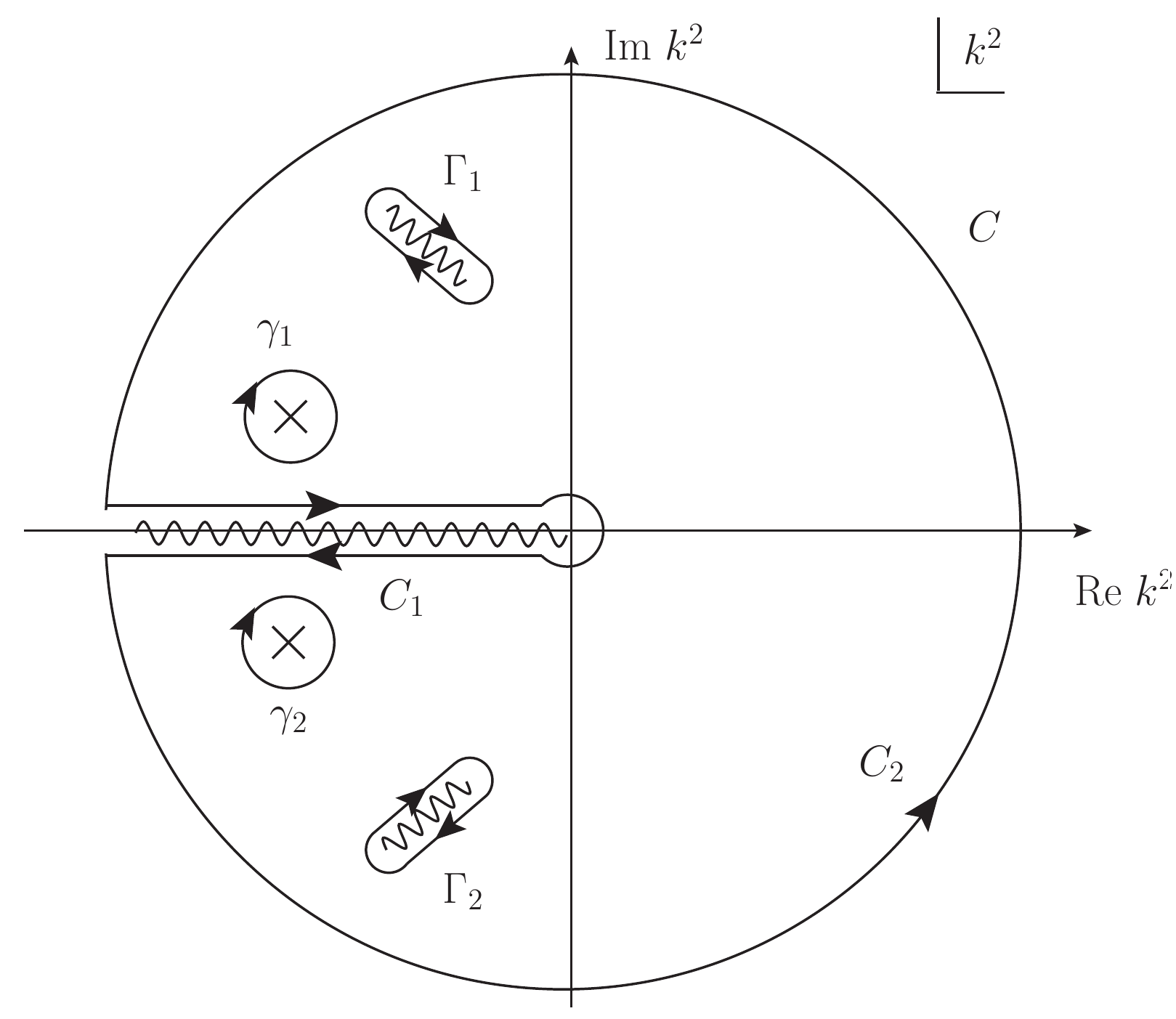}
  \end{center}
   \caption{
   The contours $\gamma_\ell$ and $\Gamma_k$ surround the pole $z_\ell$ and the branch cut $\mathcal{C}_k$ clockwise, respectively. 
   The contour $C$ consists of the path $C_1$ winding the negative real axis and the large circle $C_2$: $C = C_1 \cup C_2$. The orientation of the contour $C$ is taken counterclockwise.
   The propagator $D(k^2)$ is holomorphic in the region bounded by $C \cup \{ \gamma_\ell \}_{\ell = 1} ^{N_p} \cup \{ \Gamma_k \}_{k = 1} ^{N_c}$.
   }
    \label{fig:complex.pdf}
\end{figure}

As an immediate consequence following from the complex singularities, we derive the generalized spectral representation for $D(k^2)$.

Here, we consider the set-up illustrated in Fig.~\ref{fig:complex.pdf}, which is characterized by:
\begin{enumerate}
\renewcommand{\labelenumi}{(\arabic{enumi})}
    \item $\{ z_\ell \}_{\ell = 1} ^{N_p}$: positions of the complex poles
    \item $\{ n_\ell \}_{\ell = 1} ^{N_p}$: their orders
    \item $\gamma_\ell$: a small contour surrounding $z_\ell$ clockwise
    \item $\{ \mathcal{C}_k \}_{k = 1} ^{N_c}$: the complex branch cuts
    \item $\Gamma_k$: a contour wrapping around $\mathcal{C}_k$ clockwise
    \item $\mathcal{C}_0$: the negative real axis
    \item $C = C_1 \cup C_2$: the contour consisting of the path $C_1$ encompassing $\mathcal{C}_0$ and the large circle $C_2$ counterclockwise.
\end{enumerate}

The discontinuity of $D(\zeta)$ for a cut $\mathcal{C}_k$ ($k = 0,1,\cdots, N_c$) is denoted by $\operatorname{Disc}_{\mathcal{C}_k} D(\zeta)$. On a cut with an orientation, $\operatorname{Disc}_{\mathcal{C}_k} D(\zeta) :=  D(\zeta + i d \zeta) - D(\zeta - i d \zeta)$, where $d\zeta$ is an infinitesimal along the given orientation of $\mathcal{C}_k$.
For example, for the negative real axis $\mathcal{C}_0$ with the orientation from 0 to $-\infty$, $\operatorname{Disc}_{\mathcal{C}_0} D(-\sigma^2) =  D(-\sigma^2 - i \epsilon) - D(-\sigma^2 + i \epsilon) ~~ (\epsilon \rightarrow +0)$.

\begin{theorem}
\label{thm:spec_repr_complex}
Let $D(k^2) = S_1 (k)$ be a propagator satisfying (i) -- (v). In the above notation, the generalized spectral representation follows for $k^2$ which is not on singularities of $D(k^2)$,
\begin{align}
 D(k^2) &= \int_0 ^\infty d \sigma^2 \frac{\rho(\sigma^2)}{\sigma^2 + k^2} + \sum_{\ell=1}^{N_p} \sum_{m = 1} ^{n_\ell} \frac{Z_\ell^{(m)}}{(k^2 - z_\ell)^{m}} \notag \\
 &~~~ +  \sum_{k=1}^{N_c} \int_{\mathcal{C}_k} d \zeta \frac{\rho_k(\zeta)}{k^2 - \zeta}, \label{eq:spec_repr_complex}
 \end{align}
 where
 \begin{align}
  \rho (\sigma^2) &:= \frac{1}{2 \pi i} \operatorname{Disc}_{\mathcal{C}_0} D(-\sigma^2), \\
Z_\ell^{(m)} &:= - \frac{1}{2 \pi i} \oint_{\gamma_\ell} d k^2 D(k^2) (k^2 - z_\ell)^{m-1} \notag \\
&~~~(\ell = 1,\cdots,N_p;~m = 1, \cdots, n_\ell), \\
 \rho_k(\zeta) &:= \frac{1}{2 \pi i} \operatorname{Disc}_{\mathcal{C}_k} D(\zeta) ~~\mathrm{for}~\zeta \in \mathcal{C}_k \notag \\
 &~~~(k = 1,\cdots,N_c).
\end{align}
We have taken the orientation of $\mathcal{C}_k$ $(k = 1,\cdots,N_c)$ in the discontinuities $\operatorname{Disc}_{\mathcal{C}_k} D(\zeta)$ to coincide with the orientation of the integral in (\ref{eq:spec_repr_complex}) and the orientation of $\mathcal{C}_0$ in $\operatorname{Disc}_{\mathcal{C}_0} D(\zeta)$ to be from the origin to negative infinity.
\end{theorem}

Before proceeding to the proof, let us add several remarks.
\begin{enumerate}
\renewcommand{\labelenumi}{(\alph{enumi})}
    \item If there exists no complex singularity $(N_p = N_c = 0)$, this theorem provides the K\"all\'en-Lehmann spectral representation 
    \begin{align}
 D(k^2) &= \int_0 ^\infty d \sigma^2 \frac{\rho(\sigma^2)}{\sigma^2 + k^2}
 \end{align}
    except for the positivity $\rho(\sigma^2) > 0$. In this sense, (\ref{eq:spec_repr_complex}) is a generalization of the K\"all\'en-Lehmann spectral representation.
    \item For real-valued fields, $D(k^2)$ is real for $k^2 > 0$ as noted above. Then, from the Schwarz reflection principle $D(z^*) = [D(z)]^*$, the spectral function can be written in the form
    \begin{align}
         \rho(\sigma^2) &= \frac{1}{\pi} \operatorname{Im} D(-\sigma^2-i\epsilon)~~ (\epsilon \rightarrow +0),
    \end{align}
     which is the usual dispersion relation.
    \item Similarly, for real-valued fields, the Schwarz reflection principle $D(z^*) = [D(z)]^*$ implies that the complex singularities must appear as complex conjugate pairs (up to arbitrariness of the branch cuts).
    \item $\operatorname{Disc}_{\mathcal{C}_k} D(\zeta)$ is in general a hyperfunction, which is not very convenient for careful analyses.
    Thus, although Theorem \ref{thm:spec_repr_complex} is itself important, we utilize an equation (\ref{eq:theorem_1_Cauchy_int}) appearing in the proof given below rather than (\ref{eq:spec_repr_complex}) in order to prove subsequent theorems. Only for the timelike part, namely the first term of (\ref{eq:spec_repr_complex}), we use the spectral representation in the following subsections, since the assumption (iv~a) makes $\rho(\sigma^2)$ somewhat easy to treat.
    \item Note that the domains of the integrals only represent that $ \rho (\sigma^2) \in \mathscr{S}' ( [0,\infty])$ and that $\operatorname{supp} \rho_k$ lies in the closure of the cut $\mathcal{C}_k$. In particular, we allow a massless pole, namely a pole at the origin $k^2 = 0$, as long as the assumption (iii) is maintained.
\end{enumerate}

\begin{proof}

For any point $k^2$ not on the singularities, the Cauchy integral formula yields
\begin{align}
 D(k^2) &= \oint_C \frac{d \zeta}{2 \pi i}  \frac{D(\zeta)}{\zeta - k^2} + \sum_{\ell=1}^{N_p} \oint_{\gamma_\ell} \frac{d \zeta}{2 \pi i}  \frac{D(\zeta)}{\zeta - k^2} \notag \\
 &~~~ + \sum_{k=1}^{N_c} \oint_{\Gamma_k} \frac{d \zeta}{2 \pi i}  \frac{D(\zeta)}{\zeta - k^2}, \label{eq:theorem_1_Cauchy_int}
\end{align}
where we have chosen the contours $(C_1,~\gamma_\ell,~\Gamma_k)$ sufficiently close to the singularities.

The assumption (v) guarantees that the integration along the large circle $C_2$ vanishes. Thus, the first term reads
\begin{align}
\oint_C \frac{d \zeta}{2 \pi i}  \frac{D(\zeta)}{\zeta - k^2} &= \int_{C_1} \frac{d \zeta}{2 \pi i}  \frac{D(\zeta)}{\zeta - k^2},
\end{align}
where $C_1$ surrounds the negative real axis.

For the second term, a calculation yields
\begin{align}
\sum_{\ell=1}^{N_p} \oint_{\gamma_\ell} \frac{d \zeta}{2 \pi i}  \frac{D(\zeta)}{\zeta - k^2} = \sum_{\ell=1}^{N_p} \sum_{m = 1} ^{n_\ell} \frac{Z_\ell^{(m)}}{(k^2 - z_\ell)^{m}}.
\end{align}

Therefore, we have
\begin{align}
 D(k^2) &= \int_{C_1^{-1}} \frac{d \zeta}{2 \pi i}  \frac{D(\zeta)}{k^2 -\zeta} + \sum_{\ell=1}^{N_p} \sum_{m = 1} ^{n_\ell} \frac{Z_\ell^{(m)}}{(k^2 - z_\ell)^{m}} \notag \\
 &~~~ + \sum_{k=1}^{N_c} \oint_{\Gamma_k^{-1}} \frac{d \zeta}{2 \pi i}  \frac{D(\zeta)}{k^2-\zeta}, \label{eq:theorem_1_last}
\end{align}
where $C_1^{-1}$ and $\Gamma_k^{-1}$ denote $C_1$ and $\Gamma_k$ with inverse directions, respectively. Note that $C_1^{-1}$ and $\Gamma_k^{-1}$ are roughly ``contours surrounding the cuts counterclockwise''. By taking a limit shrinking these contours $(C_1,~\Gamma_k)$, the right hand side of (\ref{eq:theorem_1_last}) is represented as (\ref{eq:spec_repr_complex}).
\end{proof}

\subsection{Some properties of complex singularities} \label{sec:complex-properties}

Here, we derive analytic properties of propagators with complex singularities.
As a first step, we consider (Sec.~\ref{sec:CCP}) an example of one pair of complex conjugate simple poles.
After that, we prove the properties of general complex singularities: (Sec.~\ref{sec:holomorphy}) Holomorphy in the tube, (Sec.~\ref{sec:temperedness}) Violation of temperedness of the reconstructed Wightman function, (Sec.~\ref{sec:ref_pos}) Violation of reflection positivity,
(Sec.~\ref{sec:W-positivity}) Violation of (Wightman) positivity,
(Sec.~\ref{sec:Lorentz}) Lorentz symmetry, and (Sec.~\ref{sec:locality}) Locality.
The organization of this section is illustrated in Fig.~\ref{fig:summary-2-b}.

 \begin{figure*}[t]
  \begin{center}
\includegraphics[width=0.9 \linewidth]{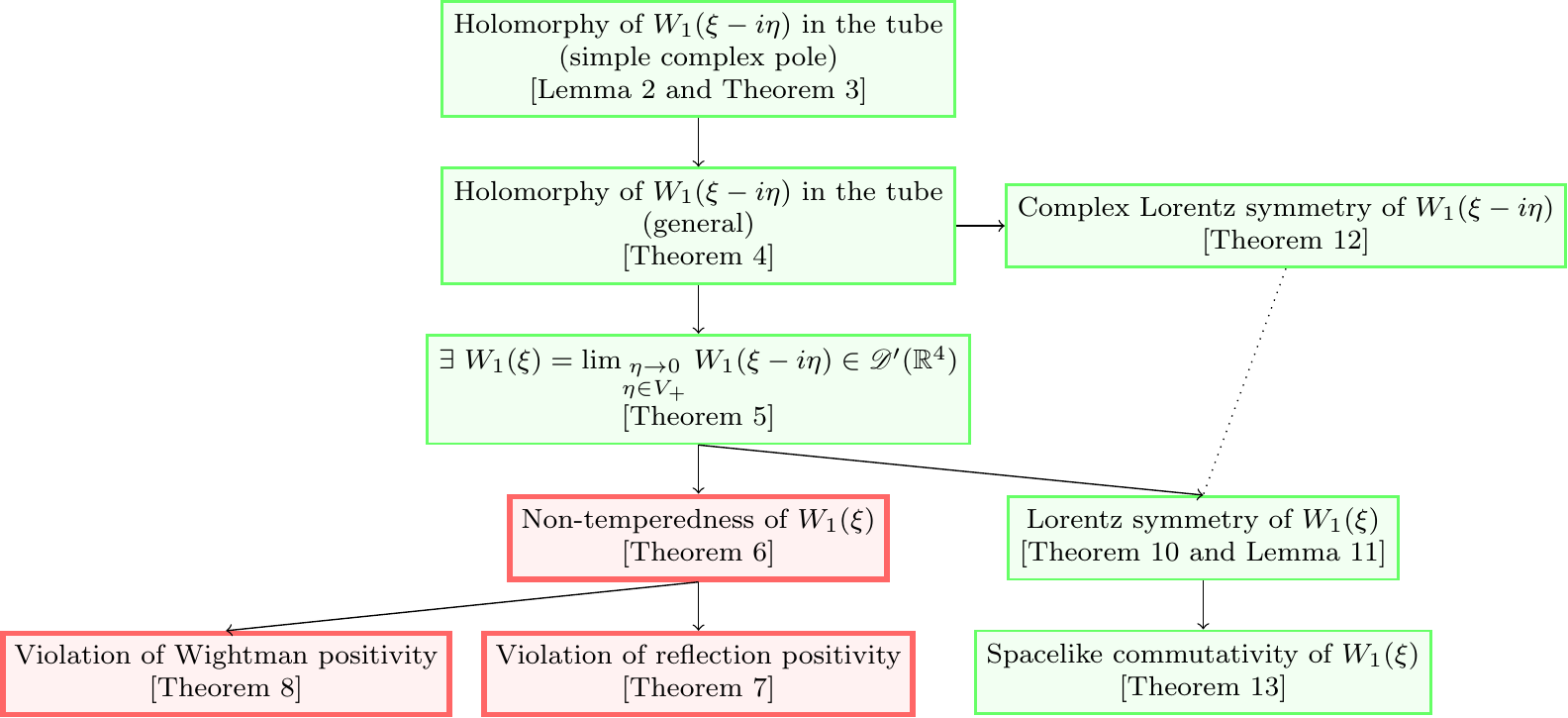}
  \end{center}
   \caption{Flow chart summarizing Sec.~\ref{sec:complex-properties}. In our proofs, a theorem at a destination of an arrow requires theorems in its upstream. 
   Fig.~\ref{fig:complex_lorentz.pdf} shows the detailed relation on the dotted line between Theorem \ref{thm:complex_Lorentz} and \ref{thm:Lorentz}.
   The green blocks are consistent with the usual QFT, while the red blocks with thick boxes contradict that.}
    \label{fig:summary-2-b}
\end{figure*}

\subsubsection{Example: one pair of complex conjugate simple poles} \label{sec:CCP}

Let us first consider the propagator $D(k^2)$ with one pair of complex conjugate simple poles, which is decomposed into the ``timelike part'' $D_{tl}(k^2) $ and ``complex-pole part'' $D_{cp}(k^2)$,
\begin{align}
    D(k^2) &= D_{tl}(k^2)  + D_{cp}(k^2) , \notag \\
    D_{tl}(k^2) &= \int_0 ^\infty d \sigma^2 \frac{\rho(\sigma^2)}{\sigma^2 + k^2} , \notag \\
    D_{cp}(k^2) &= \frac{Z}{M^2 + k^2}  + \frac{Z^*}{(M^*)^2 + k^2}, \label{eq:propagator_complex_poles}
\end{align}
Without loss of generality, we can assume $\operatorname{Im} M^2 > 0$.
Accordingly, the Schwinger function is decomposed as
\begin{align}
    S_1(\xi) &= S_{tl}(\xi)  + S_{cp}(\xi), \notag \\
    S_{tl}(\xi) &= \int \frac{d^4 k}{(2 \pi)^4} e^{ik \xi} D_{tl}(k^2), \notag \\
    S_{cp}(\xi) &= \int \frac{d^4 k}{(2 \pi)^4} e^{ik \xi} D_{cp}(k^2).
\end{align}
Our aim here is to demonstrate the reconstruction procedure $S_1 (\vec{\xi}, \xi_4) \rightarrow W_1 (\xi^0, \vec{\xi})$ according to the definition of the reconstruction (\ref{eq:W_S_connection_A}) and (\ref{eq:W_S_connection_B}).
We can reconstruct each part of the Wightman function separately, as $S_{tl} \rightarrow W_{tl}$ and $S_{cp} \rightarrow W_{cp}$,
\begin{align}
    W_1 (\xi) &= W_{tl}(\xi)  + W_{cp}(\xi).
\end{align}

We first consider the timelike part $S_{tl} \rightarrow W_{tl}$. Since the timelike part is not a main subject of this paper, let us describe the reconstruction procedure of this part only briefly. This reconstruction procedure consists of the following steps.

\begin{enumerate}
\renewcommand{\labelenumi}{Step \arabic{enumi}.}
    \item regarding $S_{tl} (\xi)$ as an ordinary function $\hat{S}_{tl} (\xi)$ on $\{ (\vec{\xi} , \xi_4)~;~ \xi_4 > 0 \}$,
    \item performing analytic continuation from $W_{tl} (-i \xi_4, \vec{\xi}) = \hat{S}_{tl} (\vec{\xi}, \xi_4)$ to $W_{tl} (\xi - i \eta)$ defined on the tube $\mathbb{R}^4 - i V_+$,
    \item taking the boundary value as a tempered distribution $W_{tl} (\xi) = \lim_{\substack{\eta \rightarrow 0 \\ \eta \in V_+}} W_{tl} (\xi-i \eta) \in \mathscr{S}'(\mathbb{R}^4)$,
\end{enumerate}
where $V_+$ denotes the (open) forward light cone
\begin{align}
    V_+ := \{ (\eta^0,\vec{\eta}) \in \mathbb{R}^4 ~;~ \eta^0 > |\vec{\eta}|  \}.
\end{align}

Let us take a closer look into each step. Main properties of the spectral function that we shall use in these steps are $ \rho (\sigma^2) \in \mathscr{S}' ( [0,\infty])$ and its regularization $\frac{1}{2 \pi i} ( D(-\sigma^2 - i \epsilon) - D(-\sigma^2 + i \epsilon))$ ($\epsilon \rightarrow + 0$).

\textit{Step 1}.
This step claims that there exists a function $\hat{S}_{tl} (\xi)$ such that\footnote{Recall that the Fourier transform of a tempered distribution is defined by the Fourier transform of its test function.}, for any test function $f(\xi) \in \mathscr{S} (\mathbb{R}^4_+)$,
\begin{align}
    \int \frac{d^4 k}{(2 \pi)^4} D_{tl}(k^2) \tilde{f}(k) = \int d^4 \xi~ f(\xi) \hat{S}_{tl} (\xi),
\end{align}
 where $\tilde{f}(k) := \int d^4 \xi ~f(\xi)  e^{i k \xi}$.
 Noting the properties of $\rho(\sigma^2)$, we have the desired function $\hat{S}_{tl} (\xi)$:
 \begin{align}
    \hat{S}_{tl} (\xi) &:= \int_0^\infty d\sigma^2 ~ \rho(\sigma^2) \hat{S}_{\sigma^2}  (\xi), \notag \\
    \hat{S}_{\sigma^2}  (\xi) &:= \int \frac{d^3 \vec{k}}{(2 \pi)^3} e^{i\vec{k} \cdot \vec{\xi}}  \frac{e^{- \sqrt{\sigma^2 + \vec{k}^2} |\xi_4|}}{2 \sqrt{\vec{k}^2 + \sigma^2}}. \label{eq:timelike-Schwinger}
\end{align}

\textit{Step 2}.
We can confirm that the Cauchy-Riemann equation holds in the tube $\xi - i \eta \in \mathbb{R}^4 - i V_+$ for the following function $W_{tl} (\xi - i \eta)$:
\begin{align}
    W_{tl} (\xi - i \eta) &:= \int_0^\infty d\sigma^2 ~ \rho(\sigma^2) W_{\sigma^2}  (\xi - i \eta),  \notag \\
    W_{\sigma^2}  (\xi - i \eta) &:= \int \frac{d^3 \vec{k}}{(2 \pi)^3} e^{i\vec{k} \cdot (\vec{\xi} - i \vec{\eta})}  \frac{e^{- i \sqrt{\sigma^2 + \vec{k}^2} (\xi^0 - i \eta^0) }}{2 \sqrt{\vec{k}^2 + \sigma^2}}, \label{eq:timelike-complex-Wightman}
\end{align}
which satisfies $W_{tl} (-i \xi_4, \vec{\xi}) = \hat{S}_{tl} (\vec{\xi}, \xi_4)$. Thus, $W_{tl} (\xi - i \eta)$ is the desired analytic continuation.

\textit{Step 3}.
We can take the limit $\eta \rightarrow 0 ~ (\eta \in V_+)$ of $W_{tl} (\xi - i \eta) $ as a functional of $\mathscr{S}(\mathbb{R}^4)$. For each $f \in \mathscr{S}(\mathbb{R}^4)$, we define
\begin{align}
    W_{tl} (f) &:= \lim_{\substack{\eta \rightarrow 0 \\ \eta \in V_+}} \int d^4 \xi~ f(\xi) W_{tl} (\xi-i \eta) \notag \\
    &= \int_0^\infty d\sigma^2 ~ \rho(\sigma^2) i \Delta^+ (f, \sigma^2), \label{eq:timelike-Wightman}
\end{align}
where $i \Delta^+ (f, \sigma^2)$ is the free Wightman function of mass $\sigma^2$:
\begin{align}
    &i \Delta^+ (f, \sigma^2) := \int d^4 \xi~f(\xi) i \Delta^+ (\xi, \sigma^2) \notag \\
    &~~:= \int \frac{d^3 k}{(2 \pi)^3}~\frac{1}{2 \sqrt{\vec{k}^2 + \sigma^2}} \left[  \int d^4 \xi~ f(\xi)  e^{- i \sqrt{\sigma^2 + \vec{k}^2} \xi^0 + i\vec{k} \cdot \vec{\xi}}   \right], \notag \\
    &i \Delta^+ (\xi, \sigma^2) = (2 \pi) \int \frac{d^4 k}{(2 \pi)^4}~ e^{-i k \xi} \theta (k_0) \delta (k^2 - \sigma^2),
\end{align}
with the Loretzian vectors $\xi = (\xi^0,\vec{\xi}), ~k = (k^0,\vec{k})$.
We can check that this linear functional $W_{tl} (f)$ is continuous in $f \in \mathscr{S}(\mathbb{R}^4)$. Hence, we obtain the timelike part of the reconstructed Wightman function which is a tempered distribution.

Let us next reconstruct the complex-pole part $S_{cp} \rightarrow W_{cp}$ in a similar way.
The complex-pole part $S_{cp}(\xi)$ can be expressed as
\begin{align}
    S_{cp}(\vec{\xi}, \xi_4) = \int \frac{d^3 \vec{k}}{(2 \pi)^3} e^{i\vec{k} \cdot \vec{\xi}} \left[ \frac{Z}{2 E_{\vec{k}} } e^{- E_{\vec{k}} |\xi_4|} +  \frac{Z^*}{2 E_{\vec{k}}^*} e^{- E_{\vec{k}}^* |\xi_4|}
    \right], \label{eq:simple_complex_poles_Schwinger}
\end{align}
where $E_{\vec{k}} = \sqrt{\vec{k}^2 +M^2}$ is a branch of $\operatorname{Re} E_{\vec{k}} > 0$.
Since we chose $\operatorname{Im} M^2 >0$, so that $\operatorname{Im} E_{\vec{k}} > 0$ holds. Note that $S_{cp}(\vec{\xi}, \xi_4)$ can be regarded as a function for $\xi_4 > 0$.

For a later purpose, we state this derivation as a lemma.

\begin{widetext}

\begin{lemma} \label{lem:S_zeta}
The following equation holds for $\zeta \in \mathbb{C}-(-\infty,0]$.
\begin{align}
S_{\zeta} (\xi) &:= \int \frac{d^4 k}{(2 \pi)^4} e^{i k \xi} \frac{1}{k^2 + \zeta} = \int \frac{d^{3} \vec{k}}{(2 \pi)^{3}} e^{i\vec{k} \cdot \vec{\xi}} \left[ \frac{e^{- \sqrt{\vec{k}^2 + \zeta} |\xi_4|}}{2 \sqrt{\vec{k}^2 + \zeta} } \right], \label{eq:lemma2}
\end{align}
where we have chosen $\operatorname{Re} \sqrt{\vec{k}^2 +\zeta^2} > 0$, and these Fourier transforms are understood in $\mathscr{S}' (\mathbb{R}^4)$ and $\mathscr{S}' (\mathbb{R}^3)$ respectively.
Moreover, the right-hand side is an ordinary function for $\xi_4 > 0$:
\begin{align}
S_{\zeta} (\xi) &= \int \frac{d^{3} \vec{k}}{(2 \pi)^{3}}  \frac{e^{i\vec{k} \cdot \vec{\xi} - \sqrt{\vec{k}^2 + \zeta} \xi_4 }}{2 \sqrt{\vec{k}^2 + \zeta} }~~~(\mathrm{in}~\mathscr{S}' (\mathbb{R}^4_+)), \label{eq:S_zeta_fct}
\end{align}
where this integral over $\vec{k}$ is the ordinary integral (namely, not necessary understood as the Fourier transform of a tempered distribution).
\end{lemma}

\begin{proof}
For the former assertion (\ref{eq:lemma2}), it is sufficient to prove that, for any test function $f \in \mathscr{S}(\mathbb{R}^4)$,
\begin{align}
     \int \frac{d^4 k}{(2 \pi)^4} \left( \frac{1}{ k^2 + \zeta} \right) \tilde{f}(k) = \int d \xi_4 \int \frac{d^{3} \vec{k}}{(2 \pi)^{3}} ~\left(  \frac{e^{- \sqrt{\vec{k}^2 + \zeta} |\xi_4|}}{2 \sqrt{\vec{k}^2 + \zeta} } 
     \right) \int d^{3} \vec{\xi} ~e^{i\vec{k} \cdot \vec{\xi}}   f(\xi), \label{eq:lem2_former}
\end{align}
where $\tilde{f}(k) := \int d^4 \xi ~f(\xi)  e^{i k \xi}$. Since both $\tilde{f}(k)$ and $f(\xi)$ are of rapid decrease, Fubini's theorem (for $\int d^4 k \rightarrow \int d^{3} \vec{k} \int d k_4$ and $\int d k_4 \int d^4 \xi \rightarrow \int d^4 \xi \int d k_4$) yields
\begin{align}
     \int \frac{d^4 k}{(2 \pi)^4} \left( \frac{1}{ k^2 + \zeta} \right) \tilde{f}(k) &= \int \frac{d^{3} \vec{k}}{(2 \pi)^{3}} \int \frac{d k_4}{(2 \pi)} \frac{1}{ k^2 + \zeta} \int d^4 \xi ~f(\xi)  e^{i k \xi} \notag \\
     &= \int \frac{d^{3} \vec{k}}{(2 \pi)^{3}} \int d^4 \xi \int \frac{d k_4}{(2 \pi)}  ~f(\xi)  \frac{e^{i k \xi}}{ k^2 + \zeta}.
\end{align}
Therefore, a simple residue calculation gives
\begin{align}
\int \frac{d^4 k}{(2 \pi)^4} \left( \frac{1}{ k^2 + \zeta} \right) \tilde{f}(k) &= \int \frac{d^{3} \vec{k}}{(2 \pi)^{3}} \int d^4 \xi ~e^{i\vec{k} \cdot \vec{\xi}} \left(  \frac{e^{- \sqrt{\vec{k}^2 + \zeta} |\xi_4|}}{2 \sqrt{\vec{k}^2 + \zeta} } 
     \right) f(\xi).
\end{align}
Since both $f(\xi)$ and $\int d^{3} \vec{\xi} ~e^{i\vec{k} \cdot \vec{\xi}}   f(\xi)$ are of rapid decrease, we can change the order of the integrals to obtain the right-hand side of (\ref{eq:lem2_former}).
This establishes the former assertion (\ref{eq:lem2_former}).

For the latter assertion (\ref{eq:S_zeta_fct}), it is enough to prove that, for any test function $f(\xi) \in \mathscr{S} (\mathbb{R}^4_+)$,
\begin{align}
     \int d \xi_4 \int \frac{d^{3} \vec{k}}{(2 \pi)^{3}} ~\left(  \frac{e^{- \sqrt{\vec{k}^2 + \zeta} |\xi_4|}}{2 \sqrt{\vec{k}^2 + \zeta} } 
     \right) \int d^{3} \vec{\xi} ~e^{i\vec{k} \cdot \vec{\xi}}   f(\xi) = \int d^4 \xi ~\left[ \int \frac{d^{3} \vec{k}}{(2 \pi)^{3}}  \frac{e^{i\vec{k} \cdot \vec{\xi} - \sqrt{\vec{k}^2 + \zeta} \xi_4 }}{2 \sqrt{\vec{k}^2 + \zeta} } \right] f(\xi).
\end{align}
This follows from Fubini's theorem and integrability\footnote{The integrability can be verified by the following estimation: for $f \in \mathscr{S} (\mathbb{R}^4_+)$,
\begin{align*}
    \left| \frac{e^{- \sqrt{\vec{k}^2 + \zeta} \xi_4}}{2 \sqrt{\vec{k}^2 + \zeta} }  f(\xi) \right| &\leq  \frac{|f(\xi)| \max_{X \geq 0} \left| e^{- X \xi_4} X^3  \right|}{2 \left| \sqrt{\vec{k}^2 + \zeta} \right| \left( \operatorname{Re} \sqrt{\vec{k}^2 + \zeta} \right)^3}  \notag \\
    &\leq \left( 2 \left| \sqrt{\vec{k}^2 + \zeta} \right| \left( \operatorname{Re} \sqrt{\vec{k}^2 + \zeta} \right)^3 (1 + (\xi)^2)^3 \right)^{-1} \notag \\
    &\times \sup_{\eta \in \mathbb{R}^4} \left( |f(\eta) (1 + (\eta)^2)^3| \max\left(1, e^{-3} \left( \frac{3}{\eta^4} \right)^3 \right) \right),
\end{align*}
which is integrable in $\vec{k}$ and $\xi$. Note that the supremum is finite due to $\left. \partial_{\xi_4}^n f(\xi) \right|_{\xi_4 = 0} = 0$ for any $n \in \mathbb{Z}_{\geq 0}$.
} of $\left| \frac{e^{- \sqrt{\vec{k}^2 + \zeta} \xi_4}}{2 \sqrt{\vec{k}^2 + \zeta} }  f(\xi) \right|$ for $f \in \mathscr{S} (\mathbb{R}^4_+)$.
\end{proof}
\end{widetext}
Note that $E_{\vec{k}} = |\vec{k}| + O(1/|\vec{k}|)$ strongly suggests that $\operatorname{Im} M^2$ does not affect the convergence.
Then, the convergence and holomorphy of the analytically-continued Schwinger function is valid in the usual tube $( - i \xi_4 , \vec{\xi}) \in  \mathbb{R}^4 - iV_+$.
This holomorphy is an important step. We shall prove this claim carefully.

\begin{theorem}
\label{thm:holomorphy_simple_complex_poles}
The complex-pole part of the Wightman function:
\begin{align}
    W_{cp}(\xi - i\eta) &= \int \frac{d^3 \vec{k}}{(2 \pi)^3} e^{i\vec{k} \cdot (\vec{\xi} - i \vec{\eta})} \notag \\
    & ~~~\times \Biggl[ \frac{Z}{2 E_{\vec{k}} } e^{- i E_{\vec{k}} (\xi^0 - i \eta^0)} +  \frac{Z^*}{2 E_{\vec{k}}^*} e^{- i E_{\vec{k}}^* (\xi^0 - i \eta^0)} 
    \Biggr] \label{eq:simple_complex_poles_hol_Wightman}
\end{align}
is holomorphic in the tube $\xi - i \eta = (\xi^0 - i \eta^0, \vec{\xi} - i \vec{\eta}) \in \mathbb{R}^4 - iV_+$.
\end{theorem}

\begin{proof}
The first and second terms of the integrand in (\ref{eq:simple_complex_poles_hol_Wightman}) decreases rapidly as $|\vec{k}| \rightarrow \infty$. Indeed, we find
\begin{align}
     & \left| \frac{Z}{2 E_{\vec{k}} } e^{- i E_{\vec{k}} (\xi^0 - i \eta^0) + i\vec{k} \cdot (\vec{\xi} - i \vec{\eta})} \right| \notag \\
     &= \frac{|Z|}{2 | E_{\vec{k}} | } e^{- \eta^0 \operatorname{Re} E_{\vec{k}} + \xi^0  \operatorname{Im} E_{\vec{k}} + \vec{k} \cdot \vec{\eta} } \notag \\
     &= \frac{|Z|}{2 | E_{\vec{k}} | } e^{\xi^0  \operatorname{Im} E_{\vec{k}} } e^{- \eta^0 (\operatorname{Re} E_{\vec{k}} - |\vec{k}| )}
     e^{-\eta^0 |\vec{k}| + \vec{k} \cdot \vec{\eta} } \notag \\
    &\leq \frac{|Z|}{2 | E_{\vec{k}} | } e^{\xi^0  \operatorname{Im} E_{\vec{k}} } e^{- \eta^0 (\operatorname{Re} E_{\vec{k}} - |\vec{k}| )} 
     e^{-(\eta^0 - |\vec{\eta}|) |\vec{k}|  }.
\end{align}
Thus, for $\eta \in V_+$, we have, as $|\vec{k}| \rightarrow \infty$,
\begin{enumerate}
\renewcommand{\labelenumi}{(\alph{enumi})}
    \item $\operatorname{Im} E_{\vec{k}} \rightarrow 0$ and $(\operatorname{Re} E_{\vec{k}} - |\vec{k}|) \rightarrow 0$ from $E_{\vec{k}} = |\vec{k}| + O(1/|\vec{k}|)$ ,
    \item exponential decreasing of $e^{-(\eta^0 - |\vec{\eta}|) |\vec{k}|}$ in $|\vec{k}|$,
\end{enumerate}
from which the first term decreases rapidly: $\frac{Z}{2 E_{\vec{k}} } e^{- i E_{\vec{k}} (\xi^0 - i \eta^0) + i\vec{k} \cdot (\vec{\xi} - i \vec{\eta})} \in \mathscr{S} (\mathbb{R}^3)$ for fixed $\xi \in \mathbb{R}^4$ and $\eta \in V_+$.
Similarly for the second term, we have $\frac{Z^*}{2 E_{\vec{k}}^* } e^{- i E_{\vec{k}}^* (\xi^0 - i \eta^0) + i\vec{k} \cdot (\vec{\xi} - i \vec{\eta})} \in \mathscr{S} (\mathbb{R}^3)$ for fixed $\xi \in \mathbb{R}^4$ and $\eta \in V_+$.

Since the integrand in (\ref{eq:simple_complex_poles_hol_Wightman}) decreases rapidly as $|\vec{k}| \rightarrow \infty$, we can change the order of the integration and differentiantions with respect to $\xi$ and $\eta$.
Therefore, the Cauchy-Riemann equations with respect to (several complex variables) $\xi - i \eta$ hold in the tube $\xi - i \eta \in \mathbb{R}^4 - iV_+$, which guarantees the holomorphy of $W_{cp}(\xi - i\eta)$ in the tube.
\end{proof}
Note that, usually, it is the spectral condition that guarantees the holomorphy of the Wightman function in the tube. Without the spectral condition, it is in general difficult to establish the analytic arguments based on the holomorphy of the Wightman functions. However, Theorem \ref{thm:holomorphy_simple_complex_poles} (and more generally Theorem \ref{thm:holomorphy_general}) suggests that such analytic arguments are still valid even in the presence of complex singularities, while complex singularities violate a prerequisite of the spectral condition, namely the temperedness (see the discussion below or Theorem \ref{thm:nontempered}).

Let us regard the Fourier transform in (\ref{eq:simple_complex_poles_hol_Wightman}) as a tempered distribution in $\vec{\xi}$ with a smooth parameter $\xi^0$. Then we can take the limit $\eta \rightarrow 0$ with $\eta \in V_+$ to obtain the reconstructed Wightman function (\ref{eq:W_S_connection_B}):
\begin{align}
    W_{cp}(\xi^0, \vec{\xi}) = \int \frac{d^3 \vec{k}}{(2 \pi)^3} e^{i\vec{k} \cdot \vec{\xi}} \left[ \frac{Z}{2 E_{\vec{k}} } e^{- i E_{\vec{k}} \xi^0} +  \frac{Z^*}{2 E_{\vec{k}}^*} e^{- i E_{\vec{k}}^* \xi^0}
    \right].
     \label{eq:simple_complex_poles_Wightman}
\end{align}
The first term in the bracket exponentially increases as $\xi^0 \rightarrow +\infty$, and so does the second one as $\xi^0 \rightarrow - \infty$, with the choice $\operatorname{Im}M^2 > 0$. Therefore, \textit{complex poles invalidate temperedness of the Wightman function}\footnote{Indeed, suppose that $W_{cp}(\xi^0, \vec{\xi})$ were a tempered distribution. Then, the Fourier transform of $W_{cp}(\xi^0, \vec{\xi})$ in $\vec{\xi}$:
$\frac{Z}{2 E_{\vec{k}} } e^{- i E_{\vec{k}} \xi^0} +  \frac{Z^*}{2 E_{\vec{k}}^*} e^{- i E_{\vec{k}}^* \xi^0}$
would be in $\mathscr{S}'(\mathbb{R}^4)$ (by the Schwartz nuclear theorem). This contradicts with the exponential growth in $\xi^0$.}. The non-temperedness is proved more generally in Sec.~\ref{sec:temperedness}.

\subsubsection{Holomorphy in the tube and boundary value} \label{sec:holomorphy}

We have seen the holomorphy of the Wightman function in the usual tube in the presence of the simple complex poles (Theorem \ref{thm:holomorphy_simple_complex_poles}).
Here we shall generalize this theorem to the cases with arbitrary complex singularities.

\begin{theorem}
\label{thm:holomorphy_general}
Let $S_1 (p) = D(p^2)$ be a two-point Schwinger function with complex singularities satisfying (i) -- (v).
Then, $W_1 ( - i \xi_4 , \vec{\xi}) = S_1 (\vec{\xi}, \xi_4) ~~(\xi_4 > 0)$ has an analytic continuation $W_1(\xi-i \eta)$ to the tube $\mathbb{R}^4 - iV_+$.
\end{theorem}

\begin{proof}
We first recall that
\begin{align}
    S_1 (\xi) = \int \frac{d^4 k}{(2 \pi)^4} e^{i k \xi} D(k^2).
\end{align}
and $D(k^2)$ can be represented as Theorem \ref{thm:spec_repr_complex}.
We know that the timelike part can be analytically continued to the tube. Therefore, we shall prove the holomorphy for the part coming from complex singularities. 

From (\ref{eq:theorem_1_Cauchy_int}) in the proof of Theorem \ref{thm:spec_repr_complex}, the contributions of complex singularities can be expressed as\footnote{For this proof, it is enough to take $\gamma_\ell$ and $\Gamma_k$ so close to their singularities that they do not intersect with the positive real axis.}
\begin{align}
    S_{complex} (\xi) &= \int \frac{d^4 k}{(2 \pi)^4} e^{i k \xi} \Biggl\{  \sum_{\ell=1}^{N_p} \oint_{\gamma_\ell} \frac{d \zeta}{2 \pi i}  \frac{-D(\zeta)}{ k^2 + (- \zeta)} \notag \\ 
    &~~~~~~ + \sum_{k=1}^{N_c} \oint_{\Gamma_k} \frac{d \zeta}{2 \pi i}  \frac{-D(\zeta)}{ k^2 + (- \zeta)}\Biggr\}.
\end{align}
Thus, it is sufficient to prove that 
\begin{align}
     \int \frac{d^4 k}{(2 \pi)^4} e^{i k \xi} \int_C \frac{d \zeta}{2 \pi i}  \frac{D(\zeta)}{ k^2 + \zeta} \label{eq:thm2_any_contour}
\end{align}
can be analytically continued to the tube for any smooth path $C$ of finite length and any smooth function $D(\zeta)$ on $C$.

To this end, let us proceed with the following steps:
\begin{enumerate}
\renewcommand{\labelenumi}{Step \arabic{enumi}.}
    \item interpreting (\ref{eq:thm2_any_contour}) as an ordinary function on $(\vec{\xi}, \xi_4) \in \mathbb{R}^3 \times (0, \infty)$, that is to say, proving that there exists an analytic function $S_C (\xi)$ on $\mathbb{R}^3 \times (0, \infty)$ such that for any test function $f(\xi) \in \mathscr{S}(\mathbb{R}_+^4)$,
\begin{align}
     \int \frac{d^4 k}{(2 \pi)^4}& \left( \int_C \frac{d \zeta}{2 \pi i}  \frac{D(\zeta)}{ k^2 + \zeta} \right) \left( \int d^4 \xi ~f(\xi)  e^{i k \xi}  \right) \notag \\
     &= \int d^4 \xi~ S_C(\xi) f(\xi), \label{eq:thm3_step1}
\end{align}
    \item constructing a holomorphic function $W_C (\xi -i \eta)$ in the tube $\mathbb{R}^4 - i V_+$ satisfying $W_C (- i \eta^0,\vec{\xi}) = S_C (\vec{\xi}, \eta^0)$ for $\eta^0 > 0$.
\end{enumerate}

\textit{Step 1: interpreting  (\ref{eq:thm2_any_contour}) as a function.}
We shall prove that 
\begin{align}
     S_C (\xi)  := \int_C \frac{d \zeta}{2 \pi i} D(\zeta) S_{\zeta} (\xi) \label{eq:thm2_any_contour2}
\end{align}
has the desired properties of Step 1, where $S_{\zeta} (\xi)$ is a function defined by (\ref{eq:S_zeta_fct}) for $\xi_4 > 0$.

\begin{enumerate}
\renewcommand{\labelenumi}{(\alph{enumi})}
    \item $S_C (\xi) $ is an analytic function in $\mathbb{R}^3 \times (0,\infty)$. Indeed, as shown in Theorem \ref{thm:holomorphy_simple_complex_poles}, $S_{\zeta} (\xi)$ is an analytic function for $\xi_4 > 0$. Since $C$ is a finite smooth path and $D(\zeta)$ is a smooth function on $C$, $S_C (\xi) $ defined by (\ref{eq:thm2_any_contour2}) is also analytic for $\xi_4 > 0$.
    \item Let us verify that (\ref{eq:thm2_any_contour2}) satisfies (\ref{eq:thm3_step1}).
For any test function $f(\xi) \in \mathscr{S}(\mathbb{R}_+^4)$,
\begin{align}
     \int \frac{d^4 k}{(2 \pi)^4}& \left( \int_C \frac{d \zeta}{2 \pi i}  \frac{D(\zeta)}{ k^2 + \zeta} \right) \left( \int d^4 \xi ~f(\xi)  e^{i k \xi}  \right) \notag \\
     &= \int_C \frac{d \zeta}{2 \pi i} D(\zeta) \int \frac{d^4 k}{(2 \pi)^4} \frac{1}{k^2 + \zeta} \int d^4 \xi ~f(\xi)  e^{i k \xi}  \notag \\
     &=  \int_C \frac{d \zeta}{2 \pi i} D(\zeta) \int d^4 \xi ~f(\xi) S_{\zeta} (\xi),
\end{align}
where we have used Lemma \ref{lem:S_zeta} in the last equality.
Since $f(\xi) S_{\zeta}(\xi)$ is integrable in $(\xi, \zeta) \in \mathbb{R}^4 \times C$, we can change the order of the integrals to obtain (\ref{eq:thm3_step1}).
\end{enumerate}
Hence, $S_C (\xi) $ given in (\ref{eq:thm2_any_contour2}) is the analytic function on $\mathbb{R}^3 \times (0,\infty)$ satisfying (\ref{eq:thm3_step1}). This completes the step 1.

\textit{Step 2: analytic continuation of $S_C(\xi)$.}
We shall prove that
\begin{align}
     W_C (\xi - i \eta ) &:= \int_C \frac{d \zeta}{2 \pi i} D(\zeta) W_{\zeta} (\xi - i \eta), \\
W_{\zeta} (\xi - i \eta) &:= \int \frac{d^{3} \vec{k}}{(2 \pi)^{3}} e^{i\vec{k} \cdot (\vec{\xi} - i \vec{\eta})} \left[ \frac{1}{2 E_{\vec{k}} } e^{- i E_{\vec{k}} (\xi^0 - i \eta^0)} \right]. \label{eq:W_zeta_holomorphic}
\end{align}
is the desired function. Indeed, $W_C (\xi - i \eta )$ satisfies the following properties.
\begin{enumerate}
\renewcommand{\labelenumi}{(\alph{enumi})}
    \item holomorphy of $W_C (\xi - i \eta ) $:
    From Theorem \ref{thm:holomorphy_simple_complex_poles}, $W_C (\xi - i \eta ) $ is holomorphic in the tube $\mathbb{R}^4 - i V_+$ due to the finiteness of $C$ and smoothness of $D(\zeta)$.
    \item $W_C (- i \eta^0, \vec{\xi} ) = S_C (\vec{\xi}, \eta^0)$ for $\eta^0 > 0$.
    Indeed, we find
\begin{align}
     W_C (- i \eta^0, \vec{\xi} )  &= \int_C \frac{d \zeta}{2 \pi i} D(\zeta) W_{\zeta} (- i \eta^0, \vec{\xi} ) \notag \\
     &= \int_C \frac{d \zeta}{2 \pi i} D(\zeta) S_{\zeta} (\vec{\xi}, \eta^0) \notag \\
     &= S_C (\vec{\xi}, \eta^0)
\end{align}
\end{enumerate}

Therefore, $W_C (\xi - i \eta )$ provides the analytic continuation of (\ref{eq:thm2_any_contour}) to the tube. This completes the proof of Theorem \ref{thm:holomorphy_general}.
\end{proof}

Note that the finiteness of branch cuts is essential in this proof.
If there existed a branch cut of infinite length with an asymptotic line $\{ r e^{i \theta} ~; ~r>0 \}$, the holomorphic Wightman function would be 
\begin{align}
     W_C (\xi - i \eta ) &= \int_C \frac{d \zeta}{2 \pi i} \left[ \frac{1}{2 \pi} \operatorname{Disc} D(\zeta) \right] W_{\zeta} (\xi - i \eta),
\end{align}
and an estimate for large $|\zeta|$ contribution would be 
\begin{align}
     W_C (\xi - i \eta ) &\sim \int dr e^{-i \sqrt{r} e^{i \theta /2} (\xi^0 - i \eta^0)} \notag \\
     &\sim \int dr e^{ \sqrt{r} (\xi^0 \sin \theta/2 - \eta^0 \cos \theta /2)}
\end{align}
Unless $\operatorname{Disc} D(\zeta)$ is strongly suppressed faster than any exponential decay as $|\zeta| \rightarrow \infty$ or the asymptotic line is the positive real axis ($\theta = 0$), the holomorphy would not be guaranteed at least by this integral representation.
Therefore, the finiteness in (iv) plays an important role to reconstruct the Wightman function.

With the finiteness of complex singularities, we can take safely the limit $\eta \rightarrow 0~ (\eta \in V_+)$ as a distribution in $\mathscr{D}'(\mathbb{R}^4)$, which is the dual space of $\mathscr{D}(\mathbb{R}^4) = \{ f(\xi) ~;~  f(\xi)\mathrm{~is~a~} C^\infty\mathrm{~function~with~a~compact~support}  \}$.

\begin{theorem}
\label{thm:complex_boundary_value}
Let $S_1 (p) = D(p^2)$ be a two-point Schwinger function with complex singularities satisfying (i) -- (v).
By Theorem \ref{thm:holomorphy_general}, $W_1 ( - i \xi_4 , \vec{\xi}) = S_1 (\vec{\xi}, \xi_4) ~~(\xi_4 > 0)$ has the analytic continuation $W_1(\xi-i \eta)$ to the tube $\mathbb{R}^4 - iV_+$.
Then, there exists the limit $\lim_{\substack{\eta \rightarrow 0 \\ \eta \in V_+}} W_1(\xi-i \eta) \in \mathscr{D}'(\mathbb{R}^4)$.
Moreover, while the part reconstructed from timelike singularities is a tempered distribution in $\mathscr{S}'(\mathbb{R}^4)$, the part from complex singularities is a tempered distribution in $\vec{\xi}$ with a smooth parameter $\xi^0$.
\end{theorem}

\begin{proof}
By Theorem \ref{thm:holomorphy_general}, $W_1 ( - i \xi_4 , \vec{\xi}) = S_1 (\vec{\xi}, \xi_4) ~~(\xi_4 > 0)$ has an analytic continuation $W_1(\xi-i \eta)$ to the tube $\mathbb{R}^4 - iV_+$.

From the proof of Theorem \ref{thm:holomorphy_general}, we can write $W_1(\xi-i \eta)$ corresponding to the representation of Theorem \ref{thm:spec_repr_complex} as
\begin{align}
    W_1(\xi-i \eta) &= W_{tl} (\xi-i \eta) + W_{complex} (\xi-i \eta) \notag \\
    W_{tl} (\xi-i \eta) &= \int_0^\infty d\sigma^2 ~ \rho(\sigma^2) W_{\sigma^2}  (\xi - i \eta)  \notag \\
    W_{complex} (\xi - i \eta) &= - \sum_{\ell=1}^{N_p} \oint_{\gamma_\ell} \frac{d \zeta}{2 \pi i} W_\zeta (\xi -i \eta) D(\zeta) \notag \\
    &~~~- \sum_{k=1}^{N_c} \oint_{\Gamma_k} \frac{d \zeta}{2 \pi i}  W_\zeta (\xi -i \eta) D(\zeta), \label{eq:hol_Wightman_general_repr}
\end{align}
where $W_{\sigma^2} (\xi -i \eta)$ and $W_\zeta (\xi -i \eta)$ are given by (\ref{eq:timelike-complex-Wightman}) and (\ref{eq:W_zeta_holomorphic}), respectively.

As seen in Sec.~\ref{sec:CCP}, the boundary value of the timelike part is a tempered distribution, represented as (\ref{eq:timelike-Wightman}): $W_{tl} (\xi) = \lim_{\substack{\eta \rightarrow 0 \\ \eta \in V_+}} W_{tl} (\xi-i \eta) \in \mathscr{S}'(\mathbb{R}^4) \subset \mathscr{D}'(\mathbb{R}^4)$.

Next, we consider the complex part $W_{complex} (\xi - i \eta)$.
As discussed in (\ref{eq:simple_complex_poles_Wightman}), $W_\zeta (\xi -i \eta)$ has a boundary value that is a tempered distribution in $\vec{\xi}$ with a smooth parameter $\xi^0$.
Indeed, by smearing it with any test function $f(\vec{\xi}) \in \mathscr{S} (\mathbb{R}^3)$,
\begin{align}
&\int d^3 \vec{\xi}~ f(\vec{\xi}) W_{\zeta} (\xi - i \eta) \notag \\
&= \int \frac{d^3 \vec{k}}{(2 \pi)^{3}} e^{\vec{k} \cdot \vec{\eta}} \left[ \frac{1}{2 E_{\vec{k}} } e^{- i E_{\vec{k}} (\xi^0 - i \eta^0)} \right] \left( \int d^{3} \vec{\xi} ~ f(\vec{\xi}) e^{i\vec{k} \cdot \vec{\xi} } \right)
\end{align}
converges to, as $\eta \rightarrow 0~ (\eta \in V_+)$,
\begin{align}
&\int d^{3} \vec{\xi}~ f(\vec{\xi}) W_{\zeta} (\xi - i \eta) \notag \\
&\rightarrow \int \frac{d^{3} \vec{k}}{(2 \pi)^{3}} \left[ \frac{1}{2 E_{\vec{k}} } e^{- i E_{\vec{k}} \xi^0} \right] \left( \int d^{3} \vec{\xi} ~ f(\vec{\xi}) e^{i\vec{k} \cdot \vec{\xi} } \right),
\end{align}
which is a $C^\infty$ function of $\xi^0$.

Let us show that the boundary value of $W_{complex} (\xi - i \eta)$ is also a tempered distribution in $\vec{\xi}$ with a smooth parameter $\xi^0$.
It suffices to prove that, for any test function $f(\vec{\xi}) \in \mathscr{S} (\mathbb{R}^3)$ and any finite smooth path $C$,
\begin{align}
    \int d^{3} \vec{\xi}~ f(\vec{\xi}) \left[ \int_C \frac{d \zeta}{2 \pi i} D(\zeta) W_{\zeta} (\xi - i \eta) \right] \label{eq:thm4_smeared_complex_Wightman}
\end{align}
has a limit that is a $C^\infty$ function of $\xi^0$ as $\eta \rightarrow 0~ (\eta \in V_+)$.

This can be proved as follows.
Due to the finiteness of $C$ and the rapid decrease of $f(\vec{\xi})$, we have
\begin{align}
    &\int d^{3} \vec{\xi}~ f(\vec{\xi}) \left[ \int_C \frac{d \zeta}{2 \pi i} D(\zeta) W_{\zeta} (\xi - i \eta) \right] \notag \\
    &= \int_C \frac{d \zeta}{2 \pi i} D(\zeta) \left[ \int d^{3} \vec{\xi}~ f(\vec{\xi}) W_{\zeta} (\xi - i \eta) \right].
\end{align}
We have already shown that $ \int d^{3} \vec{\xi}~ f(\vec{\xi}) W_{\zeta} (\xi - i \eta) $ has a limit that is a $C^\infty$ function of $\xi^0$ as $\eta \rightarrow 0~ (\eta \in V_+)$.
From the finiteness of $C$, (\ref{eq:thm4_smeared_complex_Wightman}) also has such a desired limit.

Therefore, $W_{complex}(\xi-i \eta)$ has the limit $\lim_{\substack{\eta \rightarrow 0 \\ \eta \in V_+}} W_{complex}(\xi-i \eta)$ that is a tempered distribution in $\vec{\xi}$ with a smooth parameter $\xi^0$.
Since any smooth function can be regarded as a distribution in $\mathscr{D}'(\mathbb{R}^4)$, we have $\lim_{\substack{\eta \rightarrow 0 \\ \eta \in V_+}} W_1(\xi-i \eta) \in \mathscr{D}'(\mathbb{R}^4)$.
This completes the proof of Theorem \ref{thm:complex_boundary_value}.
\end{proof}

So far, we have seen that, even in the presence of complex singularities, we can analytically continue a Schwinger function to the tube and define its Wightman function $W_1(\xi)$ on the real space as a distribution. However, the existence of complex singularities always violates the temperedness of a Wightman function as a boundary value, which is proved in the next section.

\subsubsection{Violation of temperedness of Wightman functions and ill-defined asymptotic states} \label{sec:temperedness}

\begin{theorem}
\label{thm:nontempered}
Let $S_1 (p) = D(p^2)$ be a two-point Schwinger function with complex singularities satisfying (i) -- (v).
By Theorem \ref{thm:holomorphy_general} and \ref{thm:complex_boundary_value}, $W_1 ( - i \xi_4 , \vec{\xi}) = S_1 (\vec{\xi}, \xi_4) ~~(\xi_4 > 0)$ has the analytic continuation $W_1(\xi-i \eta)$ to the tube $\mathbb{R}^4 - iV_+$ and there exists the boundary value as a distribution $W_1(\xi) := \lim_{\substack{\eta \rightarrow 0 \\ \eta \in V_+}} W_1(\xi-i \eta) \in \mathscr{D}'(\mathbb{R}^4)$.
Then, the boundary value cannot be regarded as a tempered distribution $W_1(\xi) \notin \mathscr{S}'(\mathbb{R}^4)$.
\end{theorem}

Note that this theorem can be intuitively understood as follows. Readers who can accept the following reasoning can skip the (somewhat technical) proof.
\begin{enumerate}
\renewcommand{\labelenumi}{(\alph{enumi})}
    \item For simple complex poles, the non-temperedness follows from (\ref{eq:simple_complex_poles_Wightman}).
    \item The higher-order poles $\frac{1}{(k^2 - z_\ell)^{m}}$ can be formally represented as the $(m - 1)$-th order derivative of the simple pole $\frac{1}{k^2 - z_\ell}$ with respect to $z_\ell$. Since the derivative with respect to $z_\ell$ cannot suppress the exponential growth of $W_{cp}(\xi^0,\vec{\xi})$ given in (\ref{eq:simple_complex_poles_Wightman}), higher-order complex poles also break temperedness.
    \item The contribution of a complex branch cut $\int_{\mathcal{C}_k} d \zeta \frac{\rho_k(\zeta)}{k^2 - \zeta}$ is a superposition of $ W_{-\zeta}(\xi^0, \vec{\xi}) $ with the weight $\rho_k(\zeta)$. Therefore, the exponential growth of the Wightman function in $\xi^0$ would be unchanged.
    \item Finally, let us comment on a possibility of cancellation between contributions from different complex singularities. For such cancellations to occur, they must have the same exponentially growing factor $e^{ \xi^0 \operatorname{Im} E_{\vec{k}}}$ and oscillating factor $e^{-i \xi^0 \operatorname{Re} E_{\vec{k}}}$. This indicates that this possibility occurs only if singularities are located in the same position in complex $k_4$-plane. Therefore, we would exclude this possibility.
\end{enumerate}

We prove this theorem rigorously as follows.
This proof is based on an intuition that the holomorphy in the tube would essentially imply the spectral condition for the Wightman function in momentum representation, which leads to the usual spectral representation against complex singularities as in Sec.~\ref{sec:prelim_disc}, if the Wightman function were a tempered distribution.
\begin{proof}

As a preparation, we define a holomorphic function $F_h(\xi^0 - i \eta^0)$ as
\begin{align}
    F_h (\xi^0 - i \eta^0) := \int d^3 \vec{\xi} ~ W_1 (\xi^0 - i \eta^0, \vec{\xi}) h (\vec{\xi}),
\end{align}
where $h(\vec{\xi})$ is a test function on the spatial directions $h(\vec{\xi}) \in \mathscr{S}(\mathbb{R}^3)$. We require that its Fourier transform has a compact support:
\begin{align}
    \tilde{h}(\vec{k}) := \int d^3 \vec{\xi} ~e^{i \vec{k} \cdot \vec{\xi}} h(\xi) \in \mathscr{D}(\mathbb{R}^3).
\end{align}

This function $F_h(\xi^0 - i \eta^0)$ satisfies the following properties.
\begin{enumerate}
\renewcommand{\labelenumi}{(\alph{enumi})}
    \item $F_h (\xi^0 - i \eta^0)$ is holomorphic in the lower-half plane $\eta^0 > 0$.
    \item In all directions of the limit $|\xi^0 - i \eta^0| \rightarrow \infty$ in the lower-half plane ($\eta^0 > 0$), $F_h (\xi^0 - i \eta^0)$ grows at most exponentially as can be seen from the representation (\ref{eq:hol_Wightman_general_repr}).
    \item For $\xi_4 \neq 0$, $F_h (-i |\xi_4|)$ coincides with the Schwinger function smeared by $h(\vec{\xi})$, 
\begin{align}
    S_h (\xi_4) := \int d^{3} \vec{\xi}~ S_1(\vec{\xi},\xi_4) h(\vec{\xi}).
\end{align}
    \item We define, for $\epsilon > 0$,
\begin{align}
    \tilde{S}_h ^{(\epsilon)}(k_4^2) &:= \int d\xi_4~ S_h (|\xi_4| + \epsilon) e^{- i k_4 \xi_4} \notag \\
    &= \int d\xi_4~ F_h (-i(|\xi_4| + \epsilon)) e^{- i k_4 \xi_4}.
\end{align}
The representation (\ref{eq:spec_repr_complex}) and (\ref{eq:theorem_1_Cauchy_int}), together with (\ref{eq:timelike-Schwinger}) and (\ref{eq:thm2_any_contour2}), yields\footnote{Note that the limit $\epsilon \rightarrow 0$ gives the smeared Schwinger function $\tilde{S}_h (k_4) := \int  d\xi_4~ S_h (\xi_4) e^{- i k_4 \xi_4}$. In other words, the representation (\ref{eq:spec_repr_complex}) enables us to ``complete'' the point $\xi_4 = 0$ from $F_h (-i |\xi_4|)$ defined on $\xi_4 \neq 0$.}
\begin{align}
\tilde{S}_h ^{(\epsilon)}(k_4^2) &=  \int d \xi_4~ S_h (|\xi_4| + \epsilon) e^{- i k_4 \xi_4} \notag \\
&=  \int \frac{d^{3} \vec{k}}{(2 \pi)^{3}} \tilde{h} (\vec{k}) \Biggl[ \int_0 ^\infty d \sigma^2 \frac{\rho(\sigma^2)}{\sigma^2 + \vec{k}^2 + k_4^2} e^{- \epsilon \sqrt{\sigma^2 + \vec{k}^2}}  \notag \\
&~~~~ - \sum_{\ell=1}^{N_p} \oint_{\gamma_\ell} \frac{d \zeta}{2 \pi i}  \frac{D(\zeta)}{(-\zeta) + \vec{k}^2 + k_4^2} e^{- \epsilon \sqrt{\vec{k}^2 - \zeta}} \notag \\
&~~~ - \sum_{k=1}^{N_c} \oint_{\Gamma_k} \frac{d \zeta}{2 \pi i}  \frac{D(\zeta)}{(-\zeta) + \vec{k}^2 + k_4^2} e^{- \epsilon \sqrt{\vec{k}^2 - \zeta}} \Biggr], \label{eq:nontemp_deformed_repr}
\end{align}
from which $\tilde{S}_h ^{(\epsilon)}(k_4^2)$ has some singularities in $\mathbb{C}-(-\infty,0]$ for some $\epsilon > 0$ and some $\tilde{h}(\vec{k}) \in \mathscr{D}(\mathbb{R}^3)$.

Indeed, otherwise, $\tilde{S}_h ^{(\epsilon)}(k_4^2)$ would be holomorphic in $\mathbb{C}-(-\infty,0]$ for all $\epsilon > 0$ and $\tilde{h}(\vec{k}) \in \mathscr{D}(\mathbb{R}^3)$. This implies the last line (except for the first term) of (\ref{eq:nontemp_deformed_repr}) would vanish for all $\epsilon > 0$\footnote{Since the last line of (\ref{eq:nontemp_deformed_repr}) is holomorphic at least on the negative real axis (where we have used the third assumption of (iv)), it would be an entire function.
Furthermore, it tends to vanish as $|k_4^2| \rightarrow \infty$ and therefore would vanish.}. 
Then, $\lim_{\epsilon \downarrow 0} \tilde{S}_h ^{(\epsilon)}(k_4^2) = \int d^3 \vec{k}~ \tilde{h} (\vec{k}) D(k_4^2 + \vec{k}^2)$ would be also holomorphic in $\mathbb{C}-(-\infty,0]$ for any $\tilde{h} (\vec{k}) \in \mathscr{D}(\mathbb{R}^3)$.
By taking the limit of the mollifiers, ''approximations'' to the delta function, $\tilde{h} (\vec{k}) \rightarrow \delta (|\vec{k}| - x_0)~(x_0 > 0)$, this leads to holomorphy in $\mathbb{C}-(-\infty,0]$ of $D(k^2)$\footnote{Indeed,
let $h_\varepsilon (\vec{k})$ denote such a mollifier: $h_\varepsilon (\vec{k}) \rightarrow \delta (|\vec{k}| - x_0),~~\varepsilon \rightarrow + 0$. Then, $\lim_{\varepsilon \downarrow 0} \int d^3 \vec{k}~ \tilde{h} (\vec{k}) D(k_4^2 + \vec{k}^2) = C' D(k_4^2 + x_0^2)$ for some $C' > 0$.
The complex parts of the left-hand side would be identically zero for all $\varepsilon > 0$ due to the same argument as the previous footnote.
This leads to the holomorphy in $\mathbb{C}-(-\infty,0]$ of $D(k_4^2 + x_0^2)$ for $x_0 > 0$.
}.
This contradicts with the existence of complex singularities.

\end{enumerate}
The above properties follow from the prerequisites of the theorem (i) -- (v).
We prove the theorem by contradiction. Suppose that the boundary value of the Wightman function were a tempered distribution: $\lim_{\substack{\eta \rightarrow 0 \\ \eta \in V_+}} W_1(\xi-i \eta) \in \mathscr{S}'(\mathbb{R}^4)$.
\begin{enumerate}
    \renewcommand{\labelenumi}{(\alph{enumi})}
    \setcounter{enumi}{4} 
    \item Then, the boundary value of $F_h (\xi^0 - i \eta^0)$ would be a tempered distribution $F_h (\xi^0) := \lim_{\eta^0 \downarrow 0} F_h (\xi^0 - i \eta^0) \in \mathscr{S} '(\mathbb{R})$.
\end{enumerate}
Let us find a contradiction under the circumstance characterized by (a) -- (e).

We firstly decompose $F_h (\xi^0)$ as
\begin{align}
    F_h (\xi^0) &= F_+ (\xi^0) + F_- (\xi^0) \notag \\
    F_\pm (\xi^0) &= \int \frac{d \omega}{2 \pi} e^{-i \omega \xi^0} \tilde{F}_\pm (\omega), \notag \\
    \operatorname{supp} \tilde{F}_+ \subset [0, \infty) &, ~~~ \operatorname{supp} \tilde{F}_- \subset (-\infty, 0]. \label{eq:freq_decomp_temp}
\end{align}
Since $F_h (\xi^0)$ is not a function but a tempered distribution, there is a delicate point here. We can prove this decomposition with the following manipulation.
We recall (see Appendix A) that $\mathscr{S}(\bar{\mathbb{R}}_+) := \mathscr{S}(\mathbb{R}) / \mathscr{S}_-(\mathbb{R})$ and its dual space
$\mathscr{S}' (\bar{\mathbb{R}}_+) \simeq \{ F \in \mathscr{S}'(\mathbb{R}) ~;~ \operatorname{supp} F \subset [0,\infty) \}$. We similarly define $\mathscr{S}(\bar{\mathbb{R}}_-) := \mathscr{S}(\mathbb{R}) / \mathscr{S}_+(\mathbb{R})$.
We also define
$\mathscr{X}:= \{ ([f]_+, [f]_-) \in \mathscr{S} (\bar{\mathbb{R}}_+) \oplus \mathscr{S} (\bar{\mathbb{R}}_-)~;~ f \in \mathscr{S}(\mathbb{R})\}$ and its dual $\mathscr{X}'$.
Note the homeomorphism $\mathscr{X} \simeq \mathscr{S}(\mathbb{R})$.
By the Hahn-Banach theorem, an element of $\mathscr{X}'$ can be extended to the dual space of $\mathscr{S} (\bar{\mathbb{R}}_+) \oplus \mathscr{S} (\bar{\mathbb{R}}_-)$, which is isomorphic to $\mathscr{S}' (\bar{\mathbb{R}}_+) \oplus \mathscr{S}' (\bar{\mathbb{R}}_-) \simeq \{ F \in \mathscr{S}'(\mathbb{R}) ~;~ \operatorname{supp} F \subset [0,\infty) \} \oplus \{ F \in \mathscr{S}'(\mathbb{R}) ~;~ \operatorname{supp} F \subset (- \infty,0] \}$.
Therefore, for any $\tilde{F} \in \mathscr{S}' (\mathbb{R})$, there exist $\tilde{F}_+,~\tilde{F}_- \in \mathscr{S}'(\mathbb{R})$ such that $\tilde{F} = \tilde{F}_+ + \tilde{F}_-$ with $\operatorname{supp} \tilde{F}_+ \subset [0, \infty)$ and $\operatorname{supp} \tilde{F}_- \subset (-\infty, 0]$. This justifies (\ref{eq:freq_decomp_temp}).
For a more general description on this decomposition, see Proposition A.3 of \cite{Bogolyubov:1990kw}.

Next, we list several properties of $F_- (\xi^0)$ as follows.

\begin{enumerate}
\renewcommand{\labelenumi}{(\alph{enumi}')}
    \item $F_- (\xi^0)$ can be analytically continued to the whole complex plane.
    To show this, we consider the holomorphy in the (1) lower and (2) upper half planes separately and (3) glue them.
    \begin{enumerate}
\renewcommand{\labelenumii}{(\arabic{enumii})}
        \item For the lower-half plane, we define $F_- (\xi^0 - i \eta^0) := F_h (\xi^0 - i \eta^0) - F_+ (\xi^0 - i \eta^0) $, where $F_+ (\xi^0 - i \eta^0)$ is the Laplace transform of $\tilde{F}_+(\omega)$. This is the desired holomorphic function. Indeed, because of the support property $\operatorname{supp} \tilde{F}_+ \subset [0, \infty)$, $F_+ (\xi^0 - i \eta^0)$ is holomorphic in the lower-half plane ($\eta^0 > 0$).
    The holomorphy of $F_+ (\xi^0 - i \eta^0)$ and $F_h (\xi^0 - i \eta^0)$ from (a) yields that $F_- (\xi^0 - i \eta^0)$ defined above is holomorphic in the lower-half plane. The boundary values are: $F_h (\xi^0 - i \eta^0) \rightarrow F_h (\xi^0)$ from (e) and $F_+ (\xi^0 - i \eta^0) \rightarrow F_+ (\xi^0)$ as is well-known\footnote{For example, see Theorem 2-9 in \cite{Streater:1989vi}}, from which $F_- (\xi^0 - i \eta^0)$ has the boundary value $F_- (\xi^0)$. Therefore, $F_- (\xi^0 - i \eta^0) = F_h (\xi^0 - i \eta^0) - F_+ (\xi^0 - i \eta^0)$ provides the analytic continuation to the lower-half plane.
        \item For the upper-half plane, the Laplace transform of $\tilde{F}_-(\omega)$ provides the analytic continuation due to $\operatorname{supp} \tilde{F}_- \subset (-\infty, 0]$.
        \item We have two analytic continuations in the upper and lower half planes that have the coincident boundary value on the real axis. By the one-variable version of the edge of the wedge theorem, one can find an entire function which is the analytic continuation from both half-planes.
    \end{enumerate}

\renewcommand{\labelenumi}{(\alph{enumi}1')}
    \item In all directions of the limit $|\xi^0 - i \eta^0| \rightarrow \infty$ in the lower-half plane ($\eta^0 > 0$), $F_- (\xi^0 - i \eta^0)$ grows at most exponentially.
    Indeed, both $F_+ (\xi^0 - i \eta^0)$ and $F_h (\xi^0 - i \eta^0)$ satisfy this condition due to (b) and $\operatorname{supp} \tilde{F}_+ \subset [0, \infty)$.
\setcounter{enumi}{1}
\renewcommand{\labelenumi}{(\alph{enumi}2')}
    \item In all directions of the limit $|\xi^0 - i \eta^0| \rightarrow \infty$ in the upper-half plane ($\eta^0 < 0$), $F_- (\xi^0 - i \eta^0)$ grows at most polynomially because of $\operatorname{supp} \tilde{F}_- \subset (-\infty, 0]$.
\renewcommand{\labelenumi}{(\alph{enumi}')}
    \item $F_- ( - i \xi_4)$ is of at most polynomial growth in $\xi_4 > 0$ due to (c) and $\operatorname{supp} \tilde{F}_+ \subset [0, \infty)$.
\end{enumerate}
From (a'), (b1'), (c'), and the temperedness of $F_- (\xi^0)$, a variant of the Paley-Wiener-Schwartz theorem for one-sided support (see, e.g., Theorem A of \cite{Carlsson2017}) implies that $F_- (\xi^0 - i \eta^0)$ in the lower-half plane can be written as the Laplace transformation of a tempered distribution $\tilde{F}'_- (\omega)$ of $\operatorname{supp} \tilde{F}'_- \subset [0, \infty)$ (which actually coincides with $\tilde{F}_- (\omega)$). Thus, in all directions of the limit $|\xi^0 - i \eta^0| \rightarrow \infty$ in the lower-half plane, $F_- (\xi^0 - i \eta^0)$ grows at most polynomially.
Together with (b2'), we conclude that the entire function $F_- (\xi^0 - i \eta^0)$ is a polynomial, whose Fourier transform is a point-supported distribution.

Because of the support properties $\operatorname{supp} \tilde{F}_- = \{  0 \}$ and $\operatorname{supp} \tilde{F}_+ \subset [0,\infty)$, $\tilde{F}_+(\omega)$ can absorb $\tilde{F}_-(\omega)$ in the decomposition (\ref{eq:freq_decomp_temp}). From here on, we assume $\tilde{F}_- = 0$ without loss of generality.

Finally, let us construct $\tilde{S}_h ^{(\epsilon)}(k_4^2)$ defined in (d) from $F_h (\xi^0) = F_+ (\xi^0)$.
Due to $\operatorname{supp} \tilde{F}_+ \subset [0, \infty)$, the analytic continuation of $F_h (\xi^0)$ to the lower-half plane is given by the Laplace transform of $\tilde{F}_+$,
\begin{align}
    F_h (\xi^0 - i \eta^0) = \int \frac{d \omega}{2 \pi} e^{-i \omega \xi^0} e^{- \omega \eta^0} \tilde{F}_+ (\omega),
\end{align}
which is a holomorphic function for $\eta^0 > 0$.

Therefore, using (c) and (d), we have 
\begin{align}
    \tilde{S}_h ^{(\epsilon)}(k_4^2) &=  \int d \xi_4~ F_h  ( - i |\xi_4|  - i  \epsilon) e^{- i k_4 \xi_4} \notag \\
    &= \int d \xi_4 ~e^{- i k_4 \xi_4} \int \frac{d \omega}{2 \pi}  \tilde{F}_+ (\omega) e^{- \epsilon \omega} e^{- \omega |\xi_4|}.
\end{align}
Since a tempered distribution is a sum of derivatives of continuous functions (of at most polynomial growth): $\tilde{F}_+ (\omega) = \sum_{n=1}^{M}  \left( - \frac{\partial}{\partial \omega} \right)^{\alpha_n} \tilde{f}_n (\omega)$, we can rewrite
\begin{align}
    \tilde{S}_h ^{(\epsilon)}(k_4^2) &= \sum_{n=1}^{M}  \int d \xi_4 ~e^{- i k_4 \xi_4}  \int_0 ^\infty \frac{d \omega}{2 \pi}  \tilde{f}_n (\omega) \frac{\partial^{\alpha_n}}{\partial \omega^{\alpha_n}} e^{- \epsilon \omega} e^{- \omega |\xi_4|} \notag \\
    &= \sum_{n=1}^{M} \int_0 ^\infty \frac{d \omega}{2 \pi}  \tilde{f}_n (\omega) \frac{\partial^{\alpha_n}}{\partial \omega^{\alpha_n}} \frac{2 \omega e^{- \epsilon \omega} }{k_4^2 + \omega^2}, \label{eq:nontemp_smeared_Schwinger_final}
\end{align}
where $\alpha_n$ is a non-negative integer and $\tilde{f}_n(\omega)$ is a continuous function of $\operatorname{supp} \tilde{f}_n \subset [0,\infty)$.

The representation (\ref{eq:nontemp_smeared_Schwinger_final}) shows the holomorphy of $\tilde{S}_h ^{(\epsilon)}(k_4^2)$ on $\mathbb{C}- (-\infty,0]$ for all $\epsilon > 0$ and $\tilde{h} (\vec{k}) \in \mathscr{D}(\mathbb{R}^3)$, which contradicts with the existence of singularity explained in (d). This completes the proof of Theorem \ref{thm:nontempered}. 
\end{proof}

Let us comment on some implications of the non-temperedness.
As seen from (\ref{eq:simple_complex_poles_Wightman}), a typical non-tempered behavior is the exponential growth in $\xi^0$.
The exponential growth of the Wightman function largely affects asymptotic states, which correspond to ``$\xi^0 \rightarrow \pm \infty$ limit''. This indicates that asymptotic states of the field are ill-defined without some artificial manipulations\footnote{For Lee-Wick theory, which is the simplest model providing complex poles considered below, some manipulations on S-matrix were discussed in old literature, e.g., see \cite[Sec. 16]{Nakanishi72} for a review. However, these manipulations can cause Lorentz non-invariance and acausality.
We insist that such states corresponding to complex singularities should be eliminated from the physical state space before taking the asymptotic limit (rather than causing Lorentz non-invariance).}.
Since such states in the ``full'' state space are far from being identified with asymptotic particle states and should be eliminated from the physical state space, the complex singularities could be considered as a signal of confinement.

Finally, let us comment on the spectral condition. The spectral condition for the two-point Wightman function states $\operatorname{supp} \tilde{W}_1(q) \subset V_+$, where $\tilde{W}_1(q) = \int d^4 \xi ~e^{iq\xi} W_1(\xi) $ with Lorentzian vectors $\xi,~q$.
Since the existence of $\tilde{W}_1(q)$ is assumed in the spectral condition, this condition requires the temperedness as a prerequisite.
Therefore, Theorem \ref{thm:nontempered} implies that \textit{the spectral condition is never satisfied in the presence of complex singularities}.

\subsubsection{Violation of reflection positivity} \label{sec:ref_pos}

As a consequence of the non-temperedness, we can prove that the reflection positivity [OS2] is always violated in the presence of complex singularities.
Since complex singularities invalidate the K\"all\'en-Lehmann spectral representation, some conditions of the standard axiom should be violated. Therefore, the violation of the reflection positivity is in some sense trivial.
However, for this paper to be self-contained and because of importance of this claim, we will describe the proof in detail in Appendix \ref{sec:Appendix_ref_pos}. Moreover, to the best of our knowledge, an explicit proof on this claim is new.

\begin{theorem}
\label{thm:violation_ref_pos}
If $S_1 (p) = D(p^2)$ is a two-point Schwinger function with complex singularities satisfying (i) -- (v), then the reflection positivity [OS2] is violated.
\end{theorem}

\begin{proof}
The reflection positivity for the two-point function (\ref{eq:two_pt_ref_pos}) is a necessary condition of the reflection positivity [OS2].

In Appendix \ref{sec:Appendix_ref_pos}, it is proved that the reflection positivity for the two-point function (\ref{eq:two_pt_ref_pos}) yields temperedness of a reconstructed two-point Wightman function (Theorem \ref{thm:ref_pos_yields_temperedness}). Therefore, the non-temperedness (Theorem \ref{thm:nontempered}) implies the violation of the reflection positivity.
\end{proof}

The reflection positivity, especially (\ref{eq:two_pt_ref_pos}) for the two-point function, is often checked by a necessary condition: the positivity of $\underline{S}_1 (\vec{k}, \xi_4) := \int d^3 \vec{\xi} ~e^{i \vec{k} \cdot \vec{\xi}} S_1(\vec{\xi}, \xi_4) $ (\ref{eq:two_pt_ref_pos_concise}), e.g., \cite{KWHMS19}. Using this check, one can easily show that a propagator with only simple complex conjugate poles violates the reflection positivity. Indeed, from (\ref{eq:simple_complex_poles_Schwinger}), we have, for $\xi_4 > 0$,
\begin{align}
    \underline{S}_1 (\vec{k}, \xi_4) &= \frac{Z}{2 E_{\vec{k}} } e^{- E_{\vec{k}} \xi_4} +  \frac{Z^*}{2 E_{\vec{k}}^*} e^{- E_{\vec{k}}^* \xi_4} \notag \\
    &= \frac{|Z|}{|E_{\vec{k}}|} e^{- \xi_4 \operatorname{Re} E_{\vec{k}}} \cos \left( \xi_4 \operatorname{Im} E_{\vec{k}} - \arg \left( \frac{Z}{E_{\vec{k}}} \right) \right),
\end{align}
which is negative for some $\xi_4 > 0$. However, this check is not useful to prove the violation of the reflection positivity for general propagators with complex singularities.
For example, in the case seen in Sec.~\ref{sec:CCP}, we have, by assuming some regularity of the spectral function $\rho(\sigma^2)$,
\begin{align}
    \underline{S}_1 (\vec{k}, \xi_4) &= \int_{\sqrt{\vec{k}^2}} ^\infty d \sigma ~e^{- \sigma \xi_4} \rho (\sigma^2 - \vec{k}^2) \notag \\
    &~~+  \frac{|Z|}{|E_{\vec{k}}|} e^{- \xi_4 \operatorname{Re} E_{\vec{k}}} \cos \left( \xi_4 \operatorname{Im} E_{\vec{k}} - \arg \left( \frac{Z}{E_{\vec{k}}} \right) \right),
\end{align}
which could be positive if the spectral function $\rho(\sigma^2)$ is positive and large.
Theorem \ref{thm:violation_ref_pos} indicates that the existence of complex singularities always invalidates the reflection positivity irrespective of the timelike singularities. It is redundant to check the positivity of (\ref{eq:two_pt_ref_pos_concise}) numerically for a propagator with complex singularities.

\subsubsection{Violation of (Wightman) positivity} \label{sec:W-positivity}

Let us consider the positivity condition of the Wightman function.
First of all, the standard positivity condition:
\begin{align}
    \int d^4 x d^4y~ W_1 (y-x) f^*(x) f(y) \geq 0 ~~~ \mathrm{for~any~} f \in \mathscr{S}(\mathbb{R}^4)
\end{align}
makes no sense for a non-tempered distribution $W_1 (y-x)$. It is natural to examine a positivity condition in a weak sense using $\mathscr{D}(\mathbb{R}^4)$, instead of $\mathscr{S}(\mathbb{R}^4)$, which we call \textit{Wightman positivity in $\mathscr{D}(\mathbb{R}^4)$ (for the two-point function)}:
\begin{align}
    \int d^4 x d^4y~ W_1 (y-x) f^*(x) f(y) \geq 0 ~~~ \mathrm{for~any~} f \in \mathscr{D}(\mathbb{R}^4).
    \label{eq:W-positivity}
\end{align}

Here, we examine this positivity condition.
As can be inferred from the violation of the reflection positivity, this condition is also violated in the presence of complex singularities.
We prove the following theorem in a way similar to the previous section.

\begin{theorem}
\label{thm:violation_W-pos}
Let $S_1 (p) = D(p^2)$ be a two-point Schwinger function with complex singularities satisfying (i) -- (v).
By Theorem \ref{thm:holomorphy_general} and \ref{thm:complex_boundary_value}, $W_1 ( - i \xi_4 , \vec{\xi}) = S_1 (\vec{\xi}, \xi_4) ~~(\xi_4 > 0)$ has the analytic continuation $W_1(\xi-i \eta)$ to the tube $\mathbb{R}^4 - iV_+$ and there exists the boundary value as a distribution $W_1(\xi) := \lim_{\substack{\eta \rightarrow 0 \\ \eta \in V_+}} W_1(\xi-i \eta) \in \mathscr{D}'(\mathbb{R}^4)$.
Then, the Wightman positivity in $\mathscr{D}(\mathbb{R}^4)$ for $W_1(\xi)$ is violated.
\end{theorem}

\begin{proof}
In the next lemma (Lemma \ref{lem:W-pos_yields_temperedness}), we prove that the Wightman positivity implies the temperedness of $W_1$. Therefore, the Wightman positivity is violated due to the non-temperedness (Theorem \ref{thm:nontempered}).

\end{proof}

\begin{lemma}
\label{lem:W-pos_yields_temperedness}
Let $W_1(\xi) \in \mathscr{D}'(\mathbb{R}^4)$ be a distribution satisfying the Wightman positivity in $\mathscr{D}(\mathbb{R}^4)$. Then, $W_1(\xi)$ can be regarded as a tempered distribution: $W_1(\xi)\in \mathscr{S}'(\mathbb{R}^4)$.
\end{lemma}

The following proof of Lemma \ref{lem:W-pos_yields_temperedness} is based on an intuition that $W_1(\xi)$ is roughly a matrix element of a unitary operator and is therefore bounded above in a positive-definite state space as shown in Sec.~\ref{sec:2B-properties}.

\begin{proof}
We define a sesquilinear form on $\mathscr{D}(\mathbb{R}^4)$: for $f,g \in \mathscr{D}(\mathbb{R}^4)$,
\begin{align}
    (f,g)_W := \int d^4 x d^4y~ W_1 (y-x) f^*(x) g(y), 
\end{align}
which is positive semidefinite due to the Wightman positivity (\ref{eq:W-positivity}).
For $a \in \mathbb{R}^4$, $\hat{U}(a)$ denotes an operator on $\mathscr{D}(\mathbb{R}^4)$ defined by
\begin{align}
    (\hat{U}(a) f) (x) := f(x-a),
\end{align}
which satisfies $(\hat{U}(a) f, \hat{U}(a) f)_W = (f,f)_W$.

Since $(\cdot,\cdot)_W$ is positive semidefinite, the Cauchy-Schwarz inequality yields
\begin{align}
    |(f, \hat{U}(a) g)_W | \leq \sqrt{ (f,f)_W  (g,g)_W  }.
\end{align}

Thus, for all $f,g \in \mathscr{D}(\mathbb{R}^4)$,
\begin{align}
    ( f * (g * W_1))(a) = (f^*, \hat{U}(a) \hat{g} )_W
\end{align}
is bounded in $a \in \mathbb{R}^4$, where $\hat{g} (x) := g(-x)$ and $(f*g)(x):= \int d^4 \xi~f(x-\xi) g (\xi)$.

Note that there exists a convenient necessary and sufficient condition for a distribution $T \in \mathscr{D}'(\mathbb{R}^4)$ to be a tempered distribution \cite[Theorem 6, Chapter 7]{Schwartz}:
\begin{align}
 &T \in \mathscr{S}'(\mathbb{R}^4) \Leftrightarrow \notag \\
&~~~
 \begin{aligned}
    &\alpha * T \mathrm{~is~a~smooth~function~of~at~most~} \\
    &\mathrm{polynomial~growth~for~any~} \alpha \in \mathscr{D}(\mathbb{R}^4).
 \end{aligned}
 \label{eq:Schwartz-chap7-thm6}
\end{align}

Now, let us fix an arbitrary $g \in \mathscr{D}(\mathbb{R}^4)$. Then, $( f * (g * W_1))(a)$ is a smooth function bounded above for all $f \in \mathscr{D}(\mathbb{R}^4)$.
The condition for temperedness (\ref{eq:Schwartz-chap7-thm6}) implies that we can regard $(g * W_1) \in \mathscr{S}'(\mathbb{R}^4)$, from which $(g * W_1)(x)$ is a smooth function of at most polynomial growth.

Therefore, from arbitrariness of $g \in \mathscr{D}(\mathbb{R}^4)$ and (\ref{eq:Schwartz-chap7-thm6}), we obtain $W_1 \in \mathscr{S}'(\mathbb{R}^4)$. This completes the proof of Lemma~\ref{lem:W-pos_yields_temperedness}.
\end{proof}

\subsubsection{Lorentz symmetry} \label{sec:Lorentz}

Since the Lorentz invariance is itself an important nature and also an essential step to the locality, let us carefully prove the Lorentz invariance of the reconstructed Wightman function.

\begin{theorem}
\label{thm:Lorentz}
Let $S_1 (p) = D(p^2)$ be a two-point Schwinger function with complex singularities satisfying (i) -- (v).
By Theorem \ref{thm:holomorphy_general} and \ref{thm:complex_boundary_value}, $W_1 ( - i \xi_4 , \vec{\xi}) = S_1 (\vec{\xi}, \xi_4) ~~(\xi_4 > 0)$ has the analytic continuation $W_1(\xi-i \eta)$ to the tube $\mathbb{R}^4 - iV_+$ and there exists the boundary value as a distribution $\lim_{\substack{\eta \rightarrow 0 \\ \eta \in V_+}} W_1(\xi-i \eta) \in \mathscr{D}'(\mathbb{R}^4)$.
Then, both the holomorphic Wightman function and its boundary value are (restricted) Lorentz invariant. More precisely,
for all proper orthochronous Lorentz transformations $\Lambda \in SO(3,1)^+$,
\begin{align}
    W_1(\Lambda(\xi-i \eta)) = W_1(\xi-i \eta),~~ \mathrm{for}~\xi-i \eta \in \mathbb{R}^4 - iV_+ \label{eq:Lorentz_holomorphic_Wightman}
\end{align}
and for any $f \in \mathscr{D} (\mathbb{R}^4)$,
\begin{align}
    W_1(f) = W_1(f_\Lambda),~~\mathrm{with}~f_\Lambda(\xi) := f(\Lambda^{-1} \xi). \label{eq:Lorentz_boundary_Wightman}
\end{align}
\end{theorem}

\begin{proof}
Let us first consider the holomorphic Wightman function (\ref{eq:Lorentz_holomorphic_Wightman}). This can be decomposed as (\ref{eq:hol_Wightman_general_repr}): $W_1(\xi-i \eta) = W_{tl} (\xi-i \eta) + W_{complex} (\xi-i \eta)$.
Therefore, the Lorentz invariance of  $W_1(\xi-i \eta)$ follows from that of the respective part.

The timelike part $W_{tl} (\xi-i \eta)$ is expressed as ({\ref{eq:timelike-complex-Wightman}}).
Since the free Wightman function $W_{\sigma^2} (\xi-i \eta)$ is a Lorentz invariant function as is well-known, $W_{tl} (\xi-i \eta)$ is also Lorentz invariant.

For the Lorentz invariance of the complex part $W_{complex} (\xi - i \eta)$, similarly from the representation (\ref{eq:hol_Wightman_general_repr}), it is sufficient to prove that 
$W_\zeta (\Lambda(\xi-i \eta)) = W_\zeta (\xi-i \eta)$ in $\xi-i \eta \in \mathbb{R}^4 - iV_+$ for all $\Lambda \in SO(3,1)^+$. We prove this claim in Lemma \ref{lem:Lorentz} to be given below. This established the invariance (\ref{eq:Lorentz_holomorphic_Wightman}).

The latter assertion (\ref{eq:Lorentz_boundary_Wightman}) immediately follows from the former one (\ref{eq:Lorentz_holomorphic_Wightman}).
\end{proof}

\begin{lemma}
\label{lem:Lorentz}
The Wightman function $W_\zeta (\xi-i \eta)$, (\ref{eq:W_zeta_holomorphic}), of a simple complex pole defined on $\xi-i \eta \in \mathbb{R}^4 - iV_+$ satisfies,
for all $\Lambda \in SO(3,1)^+$,
\begin{align}
        W_\zeta (\Lambda(\xi-i \eta)) = W_\zeta (\xi-i \eta).
\end{align}
\end{lemma}

 \begin{figure}[t]
  \begin{center}
   \includegraphics[width=0.9\linewidth]{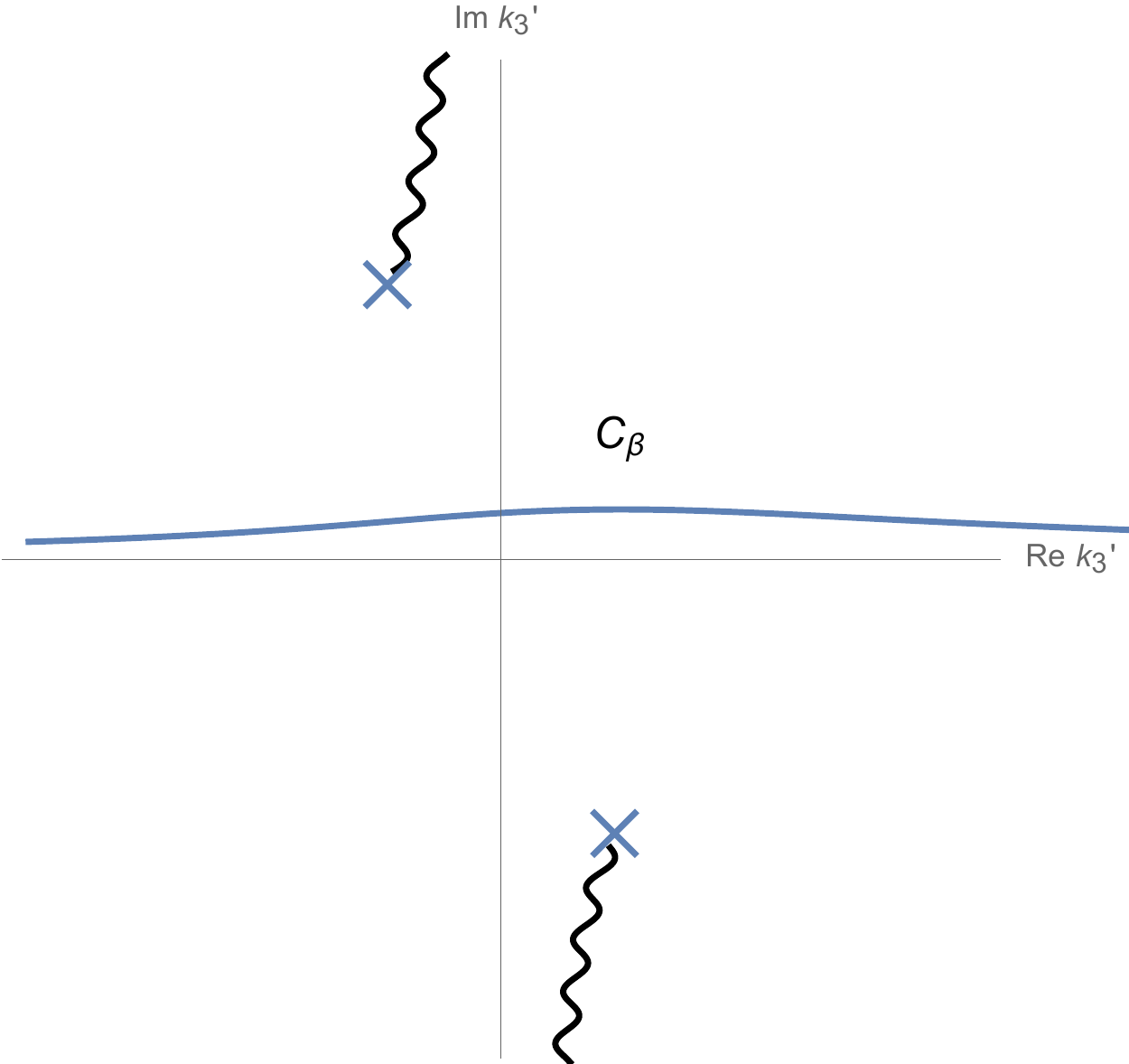}
  \end{center}
   \caption{Schematic picture of the contour $C_\beta$ in the $k_3{}'$ complex plane. The cross symbols represent the points at $E_{\vec{k}} = 0$. The integrand $\frac{1}{2 E_{\vec{k}} } e^{i\vec{k}' \cdot (\vec{\xi} - i \vec{\eta}) - i E'_{\vec{k}} (\xi^0{} - i \eta^0{}) }$ has singularities: branch points at these points and cuts represented as wavy lines. This integrand is holomorphic in the region bounded by the real axis $C_{\beta=0}$ and $C_\beta$.}
    \label{fig:Delta_Function_Contour.pdf}
\end{figure}

\begin{proof}

The spatial rotational symmetry is manifest by the expression (\ref{eq:W_zeta_holomorphic}). Therefore, it suffices to prove the invariance under the boost along $\xi^3$:
\begin{align}
\xi &= (\xi^0, \xi^1, \xi^2, \xi^3) \notag \\
&\rightarrow \xi' := \Lambda \xi = (\gamma(\xi^0 - \beta \xi^3), \xi^1, \xi^2, \gamma (\xi^3 - \beta \xi^0 )).
\end{align}
As mentioned in \cite{Nakanishi72b}, one can show the invariance under the boost by a contour deformation.

Under this transformation, $W_\zeta (\xi-i \eta)$ reads
\begin{align}
W_{\zeta} (\Lambda(\xi - i \eta)) &= \int \frac{d^{3} \vec{k}}{(2 \pi)^{3}} \frac{1}{2 E_{\vec{k}} } e^{i\vec{k} \cdot (\vec{\xi}' - i \vec{\eta}') - i E_{\vec{k}} (\xi^0{}' - i \eta^0{}') } \notag \\
&= \int \frac{d^{3} \vec{k}}{(2 \pi)^{3}} \frac{1}{2 E_{\vec{k}} } e^{i\vec{k}' \cdot (\vec{\xi} - i \vec{\eta}) - i E'_{\vec{k}} (\xi^0{} - i \eta^0{}) } ,
\end{align}
where we have defined $E_{\vec{k}} := \sqrt{\vec{k}^2 + \zeta}$ of the principal branch ($\operatorname{Re} E_{\vec{p}} > 0$), and
\begin{align}
E'_{\vec{k}} &:= \gamma( E_{\vec{k}} + \beta k_3 ), \notag \\ 
k_3' &:= \gamma (k_3 + \beta E_{\vec{k}} ),~~ \vec{k}{}' := (k_1,k_2,k_3{}').
\end{align}
Note that  a simple computation and $\operatorname{Re} E_{\vec{k}} > 0$ yield
\begin{align}
      E_{\vec{k}'} = \sqrt{\gamma^2 ( E_{\vec{k}} + \beta k_3 )^2 } = \gamma( E_{\vec{k}} + \beta k_3 ),
\end{align}
from which we have
\begin{align}
\frac{dk_3{}'}{dk_3} E_{\vec{k}} = E'_{\vec{k}} = E_{\vec{k}'}
\end{align}

By changing the variable from $\vec{k}$ to $\vec{k}'$, we obtain
\begin{align}
W_{\zeta} (\Lambda(\xi - i \eta)) &= \int_{\mathbb{R}^2 \times C_\beta} \frac{d^{3} \vec{k}'}{(2 \pi)^{3}} \frac{1}{2 E_{\vec{k}'} } e^{i\vec{k}' \cdot (\vec{\xi} - i \vec{\eta}) - i E_{\vec{k}'} (\xi^0 - i \eta^0) } ,
\end{align}
where the contour $C_\beta$ is defined by
\begin{align}
C_{\beta} := \{k_3{}' = \gamma (k_3 + \beta E_{\vec{k}} ) ;~ k_3 \in \textbf{R} \}.
\end{align}
See Fig. \ref{fig:Delta_Function_Contour.pdf}. 
Note that, for all $|\beta| <1$, $E_{\vec{k}'} = \gamma( E_{\vec{k}} + \beta k_3 )$ does not vanish on the contour $k_3{}' \in C_\beta$, namely $k_3 \in \textbf{R}$. Since the family of the contours $\{ C_{\beta'} \}_{0<\beta'<\beta}$ scans the region bounded by $C_{\beta = 0}$ and $C_\beta$, the integrand $\frac{1}{2 E_{\vec{k}'} } e^{i\vec{k}' \cdot (\vec{\xi} - i \vec{\eta}) - i E_{\vec{k}'} (\xi^0 - i \eta^0) } $ is holomorphic in the region bounded by $C_{\beta = 0}$ and $C_\beta$. Therefore, the holomorphy allows us to deform the contour $C_\beta$ into $C_{\beta = 0}$, i.e. the real axis, and finally
\begin{align}
W_{\zeta} (\Lambda(\xi - i \eta)) &= \int_{\mathbb{R}^2 \times C_\beta} \frac{d^{3} \vec{k}'}{(2 \pi)^{3}} \frac{1}{2 E_{\vec{k}'} } e^{i\vec{k}' \cdot (\vec{\xi} - i \vec{\eta}) - i E_{\vec{k}'} (\xi^0 - i \eta^0) } \notag \\
&= \int \frac{d^{3} \vec{k}'}{(2 \pi)^{3}} \frac{1}{2 E_{\vec{k}'} } e^{i\vec{k}' \cdot (\vec{\xi} - i \vec{\eta}) - i E_{\vec{k}'} (\xi^0 - i \eta^0) } \notag \\
&= W_{\zeta} (\xi - i \eta),
\end{align}
which establishes the Lorentz invariance.
\end{proof}

So far, we have verified the Lorentz invariance explicitly.
Because of importance of this assertion, we will prove it from another point of view.
The Lorentz invariance follows from a stronger symmetry, the proper complex Lorentz symmetry:

\begin{theorem}
\label{thm:complex_Lorentz}
Let $W_1(\xi - i \eta)$ be a holomorphic function in the tube $\mathbb{R}^4 - iV_+$ and invariant under the Euclidean rotation group $SO(4)$ (within the domain of definition of $W_1(\xi - i \eta)$)\footnote{Note that the action of $R \in SO(4)$ upon $(\xi - i \eta)$ is represented as $(\eta^0 + i \xi^0, \vec{\xi} - i \vec{\eta}) \mapsto R (\eta^0 + i \xi^0, \vec{\xi} - i \vec{\eta})$}. Then, $W_1(\xi - i \eta)$ is invariant under the proper complex Lorentz group $L_+ (\mathbb{C})$, including the restricted Lorentz group, namely, for any $\Lambda \in L_+(\mathbb{C})$,
\begin{align}
    z,~\Lambda z \in \mathbb{R}^4 - iV_+ ~~\Rightarrow~~ W_1 (\Lambda z) = W_1 (z). \label{eq:complex_Lorentz_invariance},
\end{align}
where $L_+ (\mathbb{C}) := \{ \Lambda \in \mathbb{C}^{4 \times 4}~;~ \Lambda^T G \Lambda = G \}$ with the metric $G = \operatorname{diag} (1,-1,-1,-1)$.
In particular, the holomorphic Wightman function of Theorem \ref{thm:holomorphy_general} satisfies (\ref{eq:complex_Lorentz_invariance}).
\end{theorem}

 \begin{figure}[t]
  \begin{center}
   \includegraphics[width= \linewidth]{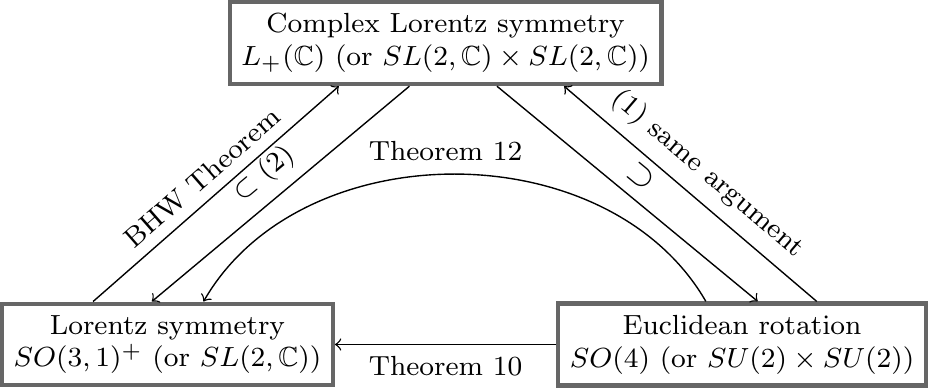}
  \end{center}
   \caption{A sketch of relations among the well-known Bargmann-Hall-Wightman (BHW) theorem, Theorem \ref{thm:Lorentz}, and Theorem \ref{thm:complex_Lorentz}. Theorem \ref{thm:complex_Lorentz} is a composition of (1) the same argument as the BHW theorem from Euclidean to complex Lorentz symmtry and (2) the restriction from complex Lorentz symmtry to Lorentz symmetry.}
    \label{fig:complex_lorentz.pdf}
\end{figure}

\begin{proof}
Since the Euclidean rotation gives a real environment of the complex Lorentz group, the assumption of the theorem and the identity theorem for holomorphic functions guarantee that, for every $z \in \mathbb{R}^4 - iV_+$, there exists a complex neighborhood of the identity element of the complex Lorentz group $L_+ (\mathbb{C})$ under which the Wightman function $W_1(z)$ is invariant.

Based on that, the same argument for proving the Bargmann-Hall-Wightman theorem (Theorem 2-11 and its Lemma of \cite{Streater:1989vi}) can be applied to this case, and therefore the former assertion (\ref{eq:complex_Lorentz_invariance}) holds.

An analytic continuation of a $SO(4)$ invariant function is invariant under $SO(4)$ within its domain of definition, since $\hat{M}_{\mu \nu} W_1 (z)$ vanishes in the domain due to the identity theorem, where $\hat{M}_{\mu \nu}$ is the $SO(4)$ symmetry generators. Thus, the latter assertion follows from the former one.
\end{proof}
Let us add some remarks.
\begin{enumerate}
\renewcommand{\labelenumi}{(\alph{enumi})}
    \item Unlike the other theorems, an generalization of this argument to $D \neq 4$ is nontrivial because of the usage of the same argument as the Bargmann-Hall-Wightman theorem.
    \item Using the Bargmann-Hall-Wightman theorem, we can prove the complex Lorentz invariance also from Theorem \ref{thm:Lorentz}.
    \item Relations among the well-known Bargmann-Hall-Wightman theorem, Theorem \ref{thm:Lorentz}, and Theorem \ref{thm:complex_Lorentz} are depicted in Fig.~\ref{fig:complex_lorentz.pdf}.
    \item As is well known, this theorem guarantees a single-valued analytic continuation of the Wightman function to the extended tube, $\mathscr{T}' := L_+(\mathbb{C}) (\mathbb{R}^4 - iV_+) = \{ \Lambda z \in \mathbb{C}^4 ~;~ \exists (z, \Lambda) \in (\mathbb{R}^4 - iV_+) \times L_+ (\mathbb{C}) \}$, which includes the Jost points $\mathbb{R}^4 \cap \mathscr{T}'$. Here, the \textit{Jost points} are just spacelike points: $\mathbb{R}^4 \cap \mathscr{T}' = \{ (\xi^0,\vec{\xi}) \in \mathbb{R}^4 ~;~ (\xi^0)^2 - \vec{\xi}^2 < 0\}$. Note that the proper complex Lorentz group includes $-1 \in L_+ (\mathbb{C})$, from which the equality $W_1(z) = W_1 (-z)$ follows.
    \item The reconstruction is based on the identification of (\ref{eq:W_S_connection_A}): $W_1 (-i\xi_4,\vec{\xi}) = S_1(\vec{\xi}, \xi_4)$. However, we have reconstructed the Wightman function using only the Schwinger function with positive imaginary time $\xi_4 > 0$. It should be possible to use the Schwinger function with negative imaginary-time $\xi_4 < 0$ for the reconstruction. The holomorphy in the extended tube together with the invariance under the proper complex Lorentz group, especially $-1 \in L_+ (\mathbb{C})$, guarantees the consistency that the reconstruction from $\xi_4 < 0$ would give the same holomorphic Wightman function as that from $\xi_4 > 0$.
\end{enumerate}

\subsubsection{Locality} \label{sec:locality}

Finally, let us comment on locality. Some argue that complex singularities are associated with non-locality.  One might claim that the non-locality of the Yang-Mills theory in a gauge-fixed picture is rather ``natural'' due to the Gribov-Singer obstruction, see \cite{Gribov78,Singer78,Maas13} and \cite{Zwanziger89,Baulieu-etal}. However, we argue that complex singularities themselves do not necessarily lead to non-locality.

For example, the problem of locality has been discussed in \cite{Stingl85,Stingl96,HKRSW90} (see also Sec.~V~A), in which they assert that complex poles describe short-lived excitations, and that the locality is broken in short range at the level of propagators but the corresponding $S$-matrix remains causal.
However, as we have mentioned above, this interpretation is different from our results.

To the best of our knowledge, the only axiomatic way to impose locality is the spacelike commutativity.
To argue that complex singularities themselves do not necessarily yield non-locality, it suffices to prove the spacelike commutativity at the level of two-point functions, because existence of complex singularities is a property of propagators.

\begin{theorem} \label{thm:spacelike_commutativity}
Let $S_1 (p) = D(p^2)$ be a two-point Schwinger function with complex singularities satisfying (i) -- (v).
By Theorem \ref{thm:holomorphy_general} and \ref{thm:complex_boundary_value}, $W_1 ( - i \xi_4 , \vec{\xi}) = S_1 (\vec{\xi}, \xi_4) ~~(\xi_4 > 0)$ has the analytic continuation $W_1(\xi-i \eta)$ to the tube $\mathbb{R}^4 - iV_+$ and there exists the boundary value as a distribution $W_1(\xi) = \lim_{\substack{\eta \rightarrow 0 \\ \eta \in V_+}} W_1(\xi-i \eta) \in \mathscr{D}'(\mathbb{R}^4)$.
Then, the boundary value $W_1(\xi)$ satisfies the spacelike commutativity:  $W_1(\xi) = W_1(-\xi)$ for spacelike $\xi$.
\end{theorem}

\begin{proof}
For a spacelike point $\xi$, there exists an element of the restricted Lorentz group $\Lambda$ such that $\Lambda \xi = - \xi$.
Therefore, the spacelike commutativity $W_1(\xi) = W_1(-\xi)$ immediately follows from Theorem \ref{thm:Lorentz}.
\end{proof}

Note that the spacelike commutativity at this level is also an immediate consequence of the holomorphy in the extended tube and the invariance under the (proper) complex Lorentz group (See Remark (d) of Theorem \ref{thm:complex_Lorentz}). 

One might argue that, e.g. from the Jost-Lehmann-Dyson (JLD) representation \cite{JLDrepr}, complex singularities could lead to violation of the local spacelike commutativity.
Nevertheless, the Wightman function with complex singularities breaks temperedness (Theorem \ref{thm:nontempered}). This non-temperedness enables a theory to evade the restriction of the theorems like the JLD representation that assumes existence of Fourier transform of Wightman functions.
Hence, there is no contradiction here.

In conclusion, \textit{even in the presence of complex singularities, the spacelike commutativity at the level of two-point functions remains intact.} Therefore, \textit{complex singularities themselves not necessarily lead to non-locality}.

\section{Interpretation in a state space with an indefinite metric} \label{sec:interpretation}
We have discussed analytic aspects of complex singularities.
In this section, we consider a possible kinematic structure yielding complex singularities, i.e., a realization of complex singularities in a quantum theory.
Since abandoning the positivity of the state-space metric is very common in Lorentz covariant gauge-fixed descriptions of gauge theories, we consider a quantum theory in a state space with an indefinite metric.

In Sec.~\ref{sec:complex_spectra}, we argue that the natural candidates providing complex singularities in an indefinite-metric state space are zero-norm pairs of eigenstates with complex eigenvalues.
In Sec~\ref{sec:Lee-Wick}, we present the Lee-Wick model as an example of QFT with complex poles.
Finally, in Sec.~\ref{sec:BRST-complex}, we discuss complex poles in the BRST formalism in a heuristic way.

\subsection{Complex singularities and complex spectra} \label{sec:complex_spectra}

An important observation is that complex energy spectrum can appear in an indefinite metric state space even if the Hamiltonian is (pseudo-)hermitian. For a review on indefinite-metric quantum field theories, see e.g. \cite{Nakanishi72}.

Beforehand, let us introduce some notions on an indefinite-metric state space.
Note that the completeness of eigenstates of a hermitian operator does not always hold even in a finite dimensional state space with an indefinite metric.
Instead of simple eigenstates, the set of ``\textit{generalized eigenstates}'' $\{ \ket{E^{0}},~\ket{E^{1}},~\cdots,  \ket{E^{n-1}} \}$
that are defined to be elements of sequences: $ (H - E) \ket{E^{0}} = E \ket{E^{1}},~ (H - E) \ket{E^{1}} = E \ket{E^{2}},~\cdots~,~ (H - E) \ket{E^{n-1}} = 0 $ spans the full state space in general, where $H$ is a hermitian operator and the value $E$ of such a sequence $\{ \ket{E^{0}},~\ket{E^{1}},~\cdots,  \ket{E^{n-1}} \}$ is called \textit{generalized eigenvalue}.
This follows from the standard Jordan decomposition. A generalized eigenstate $\ket{n}$ is said to be of \textit{order} $M$ if and only if both $(H - E_n)^{M} \ket{n} = 0$ and $(H - E_n)^{M-1} \ket{n} \neq 0$ hold.
For example, $\ket{E^{0}}$ of a sequence $\{ \ket{E^{0}},~\ket{E^{1}},~\cdots,  \ket{E^{n-1}} \}$ is a generalized eigenstate of order $n$.

For a while, we consider $0+1$ dimensional case in which a field $\phi(t)$ is regarded as an operator-valued function whose domain contains at least the vacuum $\ket{0}$, for simplicity. Alternatively, one could consider a situation in which field operators are smeared in spatial directions.

\begin{widetext}
We begin with the necessity of complex spectra for existence of complex singularities.
\begin{claim}
\label{claim:1}
Let us assume 
\begin{enumerate}
    \item completeness of (denumerable) generalized eigenstates $\ket{n}$ of the Hamiltonian $H$: $1 = \sum_{n,n'} \eta^{-1}_{n,n'}  \ket{n} \bra{n'}$, where $\eta_{n,n'} = \braket{n|n'}$ is the non-degenerate metric,
    \item translational covariance: $\phi(t) = e^{iHt} \phi(0) e^{-iHt}$,
    \item real-valuedness of generalized eigenvalues $E_n$ of the Hamiltonian $H$.
\end{enumerate}
Moreover, as technical assumptions, we assume
\begin{enumerate}
\setcounter{enumi}{3}
    \item existence of an upper bound $M$ on the orders of generalized eigenstates\footnote{
    Note that all states that are not generalized eigenstates of finite order can be seen as ``generalized eigenstates of infinite order''.
    The notion ``generalized eigenstates of infinite order'' is thus irrelevant to the spectral decomposition. Therefore, it would be appropriate to assume the upper bound.
    }, finiteness of a sum $\sum_{n'} \eta^{-1}_{n,n'}$ for any $\ket{n}$ in the complete system, and the absolute convergence of the sum,
\begin{align}
    \sum_{n,n'} \eta^{-1}_{n,n'} \sum_{k = 0}^{M(n)-1} e^{- i E_n t} \frac{(-it)^k}{k!}\bra{0} \phi(0) (H-E_n)^k \ket{n} \bra{n'} \phi(0) \ket{0},
\end{align}
 which actually equals $\braket{0|\phi(t) \phi(0)|0}$, where $E_n$ is the generalized eigenvalue of $\ket{n}$, $M(n)$ is the order, and $\ket{0}$ is the vacuum state satisfying $H\ket{0} = 0$.
\end{enumerate}
Then, the Wightman function $\braket{0|\phi(t) \phi(0)|0}$ can be regarded as a tempered distribution.
\end{claim}

\begin{proof}[Derivation]
Since $\ket{n}$ is a generalized eigenstate of order $M(n)$, $(H - E_n)^{M(n)} \ket{n} = 0$ and $(H - E_n)^{M(n)-1} \ket{n} \neq 0$ hold, which implies
\begin{align}
    e^{- i (H - E_n) t} \ket{n} = \sum_{k=0}^{M(n)-1} \frac{(-it)^k}{k!} (H-E_n)^k \ket{n}.
\end{align}

By the assumptions (i) and (ii), we have
\begin{align}
    \braket{0|\phi(t) \phi(0)|0} &= \sum_{n,n'} \eta^{-1}_{n,n'} e^{- i E_n t} \bra{0} \phi(0) e^{- i (H - E_n) t} \ket{n} \bra{n'} \phi(0) \ket{0} \notag \\
    &= \sum_{n,n'} \eta^{-1}_{n,n'} \sum_{k = 0}^{M(n)-1} e^{- i E_n t} \frac{(-it)^k}{k!}\bra{0} \phi(0) (H-E_n)^k \ket{n} \bra{n'} \phi(0) \ket{0}.
\end{align}
Note that the generalized eigenvalue $E_n$ is real by the assumption (iii).

For any test function $f(t) \in \mathscr{S} (\mathbb{R})$, we obtain
\begin{align}
    \left| \int dt ~ f(t) \braket{0|\phi(t) \phi(0)|0} \right| &= \left| \sum_{n,n'} \eta^{-1}_{n,n'} \sum_{k = 0}^{M(n)-1}  \frac{1}{k!} \left. \frac{\partial^k \tilde{f}}{\partial \omega^k} \right|_{\omega = E_n}  \bra{0} \phi(0) (H-E_n)^k \ket{n} \bra{n'} \phi(0) \ket{0} \right| \notag \\
    &\leq \left(  \sum_{n,n'} | \eta^{-1}_{n,n'} | \sum_{k = 0}^{M(n)-1} \left| \frac{1}{k!}   \bra{0} \phi(0) (H-E_n)^k \ket{n} \bra{n'} \phi(0) \ket{0} \right| \right) \notag \\
   &~~~~~~~ \times \left( \sup_{\omega, k < M} \left| \frac{\partial^k \tilde{f}}{\partial \omega^k} \right|  \right) ~\leq~ \operatorname{const.} \left( \sup_{\omega, k < M} \left| \frac{\partial^k \tilde{f}}{\partial \omega^k} \right|  \right)
\end{align}
where $\tilde{f}(\omega) = \int dt e^{-i\omega t} f(t)$ is the Fourier transform of $f(t)$ and we have used the assumptions (iv). This inequality proves $ \braket{0|\phi(t) \phi(0)|0} \in \mathscr{S}' (\mathbb{R})$.
\end{proof}

\end{widetext}

From this claim, the non-temperedness (Theorem \ref{thm:nontempered}) is incompatible with the reality of the spectrum.
Thus, complex spectra should be allowed for complex singularities to appear.
We call eigenvalues that are not real \textit{complex eigenvalues}.
Note that eigenstates of complex eigenvalues of a hermitian operator appear as pairs of zero-norm states.
As an introduction to the state-space structure with complex eigenvalues, we shall prove the following claim.

\begin{claim} \label{claim:2}
Let $H$ be a hermitian operator and have a complex eigenvalue: $H \ket{\alpha} = E_\alpha \ket{\alpha}, ~ E_\alpha \neq E_\alpha ^*$.
Suppose that its generalized eigenstates form a complete system.
Then, 
\begin{enumerate}
\renewcommand{\labelenumi}{(\arabic{enumi})}
    \item $\ket{\alpha}$ is a zero-norm state and
    \item there exists a partner state $\ket{\beta}$ such that $\braket{\beta|\alpha} \neq 0$, $\braket{\beta|\beta} = 0$, and $(H - E_\alpha^*)^k \ket{\beta} = 0 $ for some integer $k$.\footnote{One can prove the one-to-one correspondence between a sequence of generalized eigenstates of $E_\alpha$: $\{ \ket{\alpha}, (H- E_\alpha)  \ket{\alpha},  (H- E_\alpha)^2 \ket{\alpha} , \cdots \}$ and that of $E_\alpha^*$ in finite dimensional cases. For example, see section 7 of \cite{Nakanishi72}.}
\end{enumerate}
\end{claim}

\begin{proof}[Derivation]
(1) Since $E_\alpha \neq E_\alpha ^*$, the equation $E_\alpha \braket{\alpha|\alpha} = \braket{\alpha|H|\alpha} = E_\alpha ^* \braket{\alpha|\alpha}$ implies that $\ket{\alpha}$ is a zero-norm state: $\braket{\alpha|\alpha} = 0$.

(2) Because of the non-degeneracy of the metric, $\ket{\alpha}$ has a partner state, namely $\ket{\beta}$ such that $\braket{\beta|\alpha} \neq 0$. One can take a generalized eigenstate of $H$ as this state $\ket{\beta}$. Indeed, otherwise, the completeness would imply that $\ket{\alpha}$ is orthogonal to all states, which contradicts with the non-degeneracy.
Therefore, $\ket{\beta}$ satisfies: for some integer $k$,
\begin{align}
    \braket{\alpha|\beta} \neq 0,~~ (H - E_\beta)^k \ket{\beta} = 0, ~~ (H - E_\beta)^{k-1} \ket{\beta} \neq 0.
\end{align}
From the second and first equations, we have $(E_\alpha^* - E_\beta)^k \braket{\alpha|\beta}  = 0$ and therefore $E_\beta = E_\alpha^*$. Similarly to $\ket{\alpha}$, $\ket{\beta}$ is also a zero-norm state: $\braket{\beta|\beta} = 0$ since $E_\beta$ is not real, $E_\beta ^* \neq E_\beta$.
\end{proof}

The simplest possibility to provide complex singularities is a pair of the zero-norm states $\{ \ket{\alpha}, \ket{\beta} \}$.
Let us consider a consequence from such minimal complex spectra.
\begin{claim}
\label{claim:3}
Suppose, in addition to (i), (ii), (iv) of claim \ref{claim:1},
\begin{enumerate}
    \setcounter{enumi}{2}
    \renewcommand{\labelenumi}{(\roman{enumi}')}
    \item  Besides real eigenvalues, the hermitian Hamiltonian $H$ has one pair of eigenstates $\{ \ket{\alpha}, \ket{\beta} \}$ of complex conjugate eigenvalues $E_\alpha,~E_\beta = E_\alpha^*$ with a positive real part $\operatorname{Re} E_\alpha >0$.
    \setcounter{enumi}{4}
    \renewcommand{\labelenumi}{(\roman{enumi})}
    \item The field operator $\phi(t)$ is hermitian.
\end{enumerate}
Then the following statements hold:
\begin{enumerate}
\renewcommand{\labelenumi}{(\arabic{enumi})}
    \item If $\braket{\beta|\phi(0)|0} = 0$ or $\braket{\alpha|\phi(0)|0} = 0$, then the Wightman function is in $\mathscr{S}' (\mathbb{R})$. In particular, the Schwinger function has no complex singularity.
    \item If $\braket{\beta|\phi(0)|0} \neq 0$ and $\braket{\alpha|\phi(0)|0} \neq 0$, then the Schwinger function has a pair of simple complex conjugate poles besides the real singularities.
\end{enumerate}
\end{claim}

\begin{proof}[Derivation]
Firstly, let us examine the metric structure of the state space. The eigenstates of complex eigenvalues, $H \ket{\alpha} = E_\alpha \ket{\alpha},~~H \ket{\beta} = E_\alpha ^* \ket{\beta}$, are orthogonal to the generalized eigenstates with real eigenvalues $\ket{n}$. Indeed, for every $\ket{n}$ satisfying $(H - E_n)^{M(n)} \ket{n} = 0$ and $(H - E_n)^{M(n)-1} \ket{n} \neq 0$ with real $E_n$, $(E_\alpha^* - E_n)^{M(n)} \braket{\alpha | n} = 0$ and $(E_\alpha - E_n)^{M(n)} \braket{\beta | n}  = 0$ hold, from which $\braket{\alpha | n} = \braket{\beta | n} = 0$. The metric $\eta_{n,m} = \braket{n|m}$ is ``block-diagonalized'' to the sectors of real energies and of complex energies: we can decompose the completeness relation as
\begin{align}
    \sum_{n,n'} \eta^{-1}_{n,n'} = \sum_{n,n':~real} \eta^{-1}_{n,n'} + \sum_{n,n':~complex} \eta^{-1}_{n,n'}
\end{align}
The metric $\eta^{-1}_{n,n'}$ in the second term is a two-by-two matrix and can be written as: $\eta^{-1}_{\alpha,\alpha} = \eta^{-1}_{\beta,\beta} = 0$, $\eta^{-1}_{\alpha,\beta} = (\braket{\beta|\alpha})^{-1}$, and $\eta^{-1}_{\beta, \alpha} = (\braket{\alpha|\beta})^{-1}$.

Now, we have
\begin{align}
&\braket{0|\phi(t) \phi(0)|0} \notag \\
&=  \sum_{n,n':~real} \eta^{-1}_{n,n'} e^{- i E_n t} \braket{0|\phi(0)|n} \braket{n'|\phi(0)|0} \notag \\
&~~+   \sum_{n,n':~complex} \eta^{-1}_{n,n'} e^{- i E_n t} \braket{0|\phi(0)|n} \braket{n'|\phi(0)|0}.
\end{align}
The first term is characterized by claim \ref{claim:1}, which provides singularities only on the negative real axis in the Schwinger function.
On the other hand, the second term reads
\begin{align}
W_{complex}(t) &:= \sum_{n,n':~complex} \eta^{-1}_{n,n'} e^{- i E_n t} \braket{0|\phi(0)|n} \braket{n'|\phi(0)|0}  \notag \\
&= (\braket{\beta|\alpha})^{-1} e^{- i E_\alpha t} \braket{0|\phi(0)|\alpha} \braket{\beta|\phi(0)|0} \notag \\
&~~ + (\braket{\alpha|\beta})^{-1} e^{- i E_\alpha^* t} \braket{0|\phi(0)|\beta} \braket{\alpha|\phi(0)|0}.
\end{align}

Let us evaluate $W_{complex}(t)$ in the following cases.
\begin{enumerate}
\renewcommand{\labelenumi}{(\arabic{enumi}):}
    \item $\braket{\beta|\phi(0)|0} = 0$ or $\braket{\alpha|\phi(0)|0} = 0$. The hermiticity of $\phi$ yields
\begin{align}
    \braket{\alpha|\phi(0)|0} = \braket{0|\phi(0)|\alpha}^*, \notag \\
    \braket{\beta|\phi(0)|0} = \braket{0|\phi(0)|\beta}^*,
\end{align}
    from which $W_{complex}(t) = 0$ in this case. Thus, the Wightman function can be regarded as a tempered distribution.
    \item $\braket{\beta|\phi(0)|0} \neq 0$ and $\braket{\alpha|\phi(0)|0} \neq 0$. We define
\begin{align}
    Z := \frac{2 E_\alpha \braket{0|\phi(0)|\alpha} \braket{\beta|\phi(0)|0}}{\braket{\beta|\alpha}} ,
\end{align}
which does not vanish in this case. The Schwinger function of this part $S_{complex}(\tau)$ for $\tau \neq 0$ is given by
\begin{align}
    S_{complex}(\tau) &= W_{complex} (-i |\tau|) \notag \\
    &= \frac{Z}{2 E_\alpha } e^{- E_\alpha |\tau|} + \frac{Z^*}{2 E_\alpha^*}  e^{- E_\alpha^* |\tau|}.
\end{align}
This function can be represented as
\begin{align}
    S_{complex}(\tau) &= \int \frac{dk}{2 \pi} e^{ik \tau} \tilde{S}_{complex}(k), \notag \\
    \tilde{S}_{complex}(k) &= \frac{Z}{k^2 + E_\alpha^2} +  \frac{Z^*}{k^2 + (E_\alpha^*)^2},
\end{align}
which is indeed a pair of simple complex conjugate poles.
\end{enumerate}
Therefore, the pair of eigenstates $\{ \ket{\alpha}, \ket{\beta} \}$ leads to either (1) the Wightman function is in $\mathscr{S}' (\mathbb{R})$ or (2) the Schwinger function has a pair of simple complex conjugate poles.
\end{proof}

Therefore, complex singularities defined in the previous section can appear in a state space with an indefinite metric, when the Hamiltonian $H$ has complex spectra.
This claim suggests \textit{a correspondence between complex singularities and zero-norm pairs of eigenstates of complex eigenvalues}.
Finally, let us add remarks on this claim.
\begin{enumerate}
\renewcommand{\labelenumi}{(\alph{enumi})}
    \item The necessity of an indefinite metric for complex singularities is consistent with Theorem \ref{thm:violation_W-pos}, the violation of the Wightman positivity.
    \item Claim \ref{claim:3} also implies that, under the assumption of the hermiticity of the Hamiltonian and field operators, complex singularities should appear as complex conjugate pairs. This statement can be also understood by the (intuitive) representation of the Schwinger function $S(\tau)$: for $\tau > 0$, $S(\tau) = \braket{0|\phi(0) e^{- H \tau } \phi(0)|0}$. The hermiticity of the Hamiltonian and the field operator yields $S(\tau) \in \mathbb{R}$, from which $D(z)^* = D(z^*)$.
This complex-conjugate pairing is consistent with Remark (d) of Theorem \ref{thm:spec_repr_complex}.
\end{enumerate}

The discussion above is restricted to quantum mechanics, or $(0+1)$ dimension.
In the next subsection, we see an example of QFT with complex poles.

\subsection{Example: Lee-Wick model} \label{sec:Lee-Wick}
A simple possible QFT yielding complex poles is the Lee-Wick model of complex ghosts \cite{LW69}, which has been studied for long years.
Here we briefly review its kinematic structure following its covariant operator formulation given in Ref.~\cite{Nakanishi72b} and see that there indeed exists a hermitian field whose propagator has complex poles.

Let us start with the Lagrangian density of the Lee-Wick model of complex scalar field $\phi$ with complex mass $M^2 \in \mathbb{C}$,
\begin{align}
\mathscr{L} :=& \frac{1}{2} \bigl[ (\partial_\mu \phi)(\partial^\mu \phi) + (\partial_\mu \phi)^\dagger (\partial^\mu \phi)^\dagger \notag \\
&~~~ - M^2 \phi^2 - (M^*)^2 (\phi^\dagger)^2  \bigr].
\end{align}
We expand the field operator $\phi$ as 
\begin{align}
\phi(x) &= \phi^{(+)} (x) + \phi^{(-)} (x), \notag \\
\phi^{(+)} (x) &= \int \frac{d^3 p}{(2 \pi)^3} \frac{1}{\sqrt{2 E_{\vec{p}}}} \alpha(\vec{p}) e^{i \vec{p}\cdot \vec{x} - i E_{\vec{p}}t}, \notag \\
\phi^{(-)} (x) &= \int \frac{d^3 p}{(2 \pi)^3} \frac{1}{\sqrt{2 E_{\vec{p}}}} \beta^\dagger(\vec{p}) e^{-i \vec{p}\cdot \vec{x} + i E_{\vec{p}}t},
\end{align}
where $E_{\vec{p}} := \sqrt{M^2 + \vec{p}^2}$ and we chose $\operatorname{Re} E_{\vec{p}} \geq 0$ and $\operatorname{Re} \sqrt{E_{\vec{p}}} \geq 0$.
The canonical commutation relation implies $[\alpha(\vec{p}), \beta^\dagger(\vec{q})] =[\beta(\vec{p}), \alpha^\dagger(\vec{q})] = (2 \pi)^3 \delta(\vec{p} - \vec{q})$.
We define the vacuum $\ket{0}$ by $\alpha(\vec{p}) \ket{0} = \beta(\vec{p}) \ket{0} = 0$, or $\phi^{(+)} (x) \ket{0} = [\phi^{(-)} (x) ]^\dagger \ket{0} = 0$.
Note that the field operator $\phi(x)$ together with its parts $\phi^{(+)} (x)$ and $\phi^{(-)} (x)$ is a Lorentz scalar, and therefore the vacuum $\ket{0}$ is a Lorentz invariant state, see \cite{Nakanishi72b} for details.
Note that the Lorentz symmetry is manifest in this formulation until one (artificially) considers asymptotic states.
The Hamiltonian reads,
\begin{align}
H = \int \frac{d^3 p}{(2 \pi)^3} \left[ E_{\vec{p}} \beta^\dagger(\vec{p}) \alpha(\vec{p}) + E_{\vec{p}}^* \alpha^\dagger(\vec{p}) \beta(\vec{p})  \right],
\end{align}
ignoring some constant. Notice that the complex-energy states $\alpha^\dagger(\vec{p}) \ket{0}$ and $\beta^\dagger (\vec{p}) \ket{0}$ form a pair of zero-norm states $(\ket{\vec{p},\alpha} := \alpha^\dagger(\vec{p}) \ket{0}, \ket{\vec{p},\beta} := \beta^\dagger(\vec{p}) \ket{0})$ for every $\vec{p} \in \mathbb{R}^3$:
\begin{align}
    &\braket{\vec{p},\alpha|\vec{q},\alpha} = \braket{\vec{p},\beta|\vec{q},\beta} = 0, \notag \\
    &\braket{\vec{p},\alpha|\vec{q},\beta} = \braket{\vec{p},\beta|\vec{q},\alpha} = (2 \pi)^3 \delta(\vec{p} - \vec{q}).
\end{align}

The commutators of the fields are given by
\begin{align}
[\phi(x), \phi(y)] &= i \Delta(x-y,M^2), \notag \\
[\phi(x), \phi^\dagger(y)] &= 0
\end{align}
where
\begin{align}
\Delta(x,M^2) := \int \frac{d^3 p}{(2 \pi)^3} \frac{1}{ E_{\vec{p}}} \sin (\vec{p}\cdot \vec{x} -  E_{\vec{p}}t). \label{Lorentz_inv_del_fct}
\end{align}
Note that $\Delta(x,M^2)$ is a Lorentz-invariant function as shown in Lemma \ref{lem:Lorentz} as expected from the invariance of the field operator and the vacuum state. This theory is thus spacelike commutative at least in the level of elementary fields, since $\Delta(x-y,M^2)$ vanishes for spacelike $x-y$.

Next, let us show that the Euclidean propagator of a hermitian combination with a constant $Z \in \mathbb{C}$,
\begin{align}
\Phi := \sqrt{Z} \phi + \sqrt{Z^*} \phi^\dagger,
\end{align}
has indeed complex poles. In this sense, the complex fields $\phi$ and $\phi^\dagger$ are the counterparts in the covariant operator formalism of so-called i-particles \cite{Baulieu-etal}.

Using the following correlators,
\begin{align}
\braket{0|\phi(x)\phi(0)|0} &= \int \frac{d^3 p}{(2 \pi)^3} \frac{1}{ 2 E_{\vec{p}}}  e^{i \vec{p}\cdot \vec{x} - i E_{\vec{p}}t}, \notag \\
\braket{0|\phi(x)\phi^\dagger(0)|0} &= 0, \notag \\
\braket{0|\phi^\dagger(x)\phi^\dagger(0)|0} &= \int \frac{d^3 p}{(2 \pi)^3} \frac{1}{ 2 E_{\vec{p}}^*}  e^{i \vec{p}\cdot \vec{x} - i E_{\vec{p}}^*t}, \label{eq:Lee-Wick-Wightman}
\end{align}
we find
\begin{align}
D^>_\Phi (t,\vec{x}) &:= \braket{0|\Phi(x)\Phi(0)|0} \notag \\
&= \left[ Z \braket{0|\phi(x)\phi(0)|0}  +  Z^* \braket{0|\phi^\dagger(x)\phi^\dagger(0)|0}  \right] \notag \\
&= \int \frac{d^3 p}{(2 \pi)^3} \left[ \frac{Z}{ 2 E_{\vec{p}}}  e^{i \vec{p}\cdot \vec{x} - i E_{\vec{p}}t} + \frac{Z^*}{ 2 E_{\vec{p}}^*}  e^{i \vec{p}\cdot \vec{x} - i E_{\vec{p}}^* t} \right],
\end{align}
which is exactly the same as the Wightman function (\ref{eq:simple_complex_poles_Wightman}) reconstructed from the Schwinger function (\ref{eq:propagator_complex_poles}).
From the relation (\ref{eq:connection_Euc_QFT}), we obtain the Euclidean propagator $\Delta_\Phi (\tau,\vec{x})$ for $\tau \neq 0$,
\begin{align}
\Delta_\Phi& (\tau,\vec{x}) := \theta(-\tau) D^>_\Phi ( i \tau ,\vec{x}) + \theta(\tau) D^<_\Phi ( i \tau ,\vec{x}) \notag \\
&= \int \frac{d^3 p}{(2 \pi)^3} \left[ \frac{Z}{ 2 E_{\vec{p}}}  e^{i \vec{p}\cdot \vec{x} -  E_{\vec{p}}|\tau|} + \frac{Z^*}{ 2 E_{\vec{p}}^*}  e^{i \vec{p}\cdot \vec{x} -  E_{\vec{p}}^* |\tau|} \right] \notag \\
&=  \int \frac{d^3 p}{(2 \pi)^3} \int \frac{d p_4}{2 \pi}  e^{i \vec{p}\cdot \vec{x} +ip_4 \tau} \Biggl[ \frac{Z}{ p_4^2 + E_{\vec{p}}^2} + \frac{Z^*}{ p_4^2 + (E_{\vec{p}}^*)^2} \Biggr].
\end{align}
The Euclidean propagator in the momentum space is given by
\begin{align}
\Delta_\Phi (p_E) &= \frac{Z}{ p_E^2 + M^2} + \frac{Z^*}{ p_E^2 + (M^*)^2},
\end{align}
which indeed exhibits a pair of complex conjugate poles.

Therefore, a kinematic structure of the covariant operator formalism of the Lee-Wick model yields simple complex poles.
The simple complex poles correspond to the one-particle-like zero-norm states with complex masses.

Finally, let us comment on a construction of a composite operator whose propagator obeys the K\"all\'en-Lehmann representation \cite{Baulieu-etal}.

As mentioned above, the field $\phi(x)$ corresponds to the so-called i-particle. According to the toy model \cite{Baulieu-etal}, we define
\begin{align}
    \mathcal{O} (x) := \phi(x) \phi^\dagger (x).
\end{align}
This propagator can be expressed as
\begin{align}
    &D^>_{\mathcal{O}} (y-x) := \braket{0|\mathcal{O} (y) \mathcal{O} (x) |0} \notag \\
    &~~= \int \frac{d^3 p}{(2 \pi)^3} \frac{d^3 q}{(2 \pi)^3} \frac{1}{4 E_p E_q^*} e^{-i (E_p + E_q^*) (y^0 - x^0) + i (\vec{p} + \vec{q}) \cdot (\vec{y} - \vec{x})},
\end{align}
which seems not tempered since $(E_p + E_q^*)$ is complex in general. However, the following reasoning indicates that this composite-field propagator involves only real spectra\footnote{This phenomenon corresponds to non-uniqueness of Cauchy integral.
For example, if $D(k^2)$ has singularities only on the negative real axis, one can represent $D(k^2) = \int_C \frac{d \zeta}{2 \pi i} \frac{D(\zeta)}{\zeta - k^2}$, where $C$ is an arbitrary contour which separates the positive and negative real axis. In this representation, $D(k^2)$, which has no complex singularities, appears to have complex singularities on the contour $C$.}. 
The Euclidean propagator $\underline{\Delta_\mathcal{O}} (\tau, \vec{k})$ in the imaginary time $\tau$ and spatial momentum $\vec{k}$ is given by
\begin{align}
   \underline{\Delta_\mathcal{O}} (\tau, \vec{k}) &= \int \frac{d^3 p}{(2 \pi)^3} \frac{1}{4 E_p E_{k-p}^*} e^{- (E_p + E_{k-p}^*) |\tau| },
\end{align}
which reads in the momentum space,
\begin{align}
    \Delta_\mathcal{O}& (\vec{k}, k_4) =  \int \frac{d^3 p}{(2 \pi)^3}  \frac{E_p + E_{k-p}^*}{2 E_p E_{k-p}^*} \frac{1}{p_4^2 + (E_p + E_{k-p}^*)^2 } \notag \\
    &= \int \frac{d^3 p}{(2 \pi)^3} \frac{1}{2} \Biggl[ \frac{1}{E_p} \frac{1}{(k_4 - i E_p)^2 +  (E_{k-p}^*)^2 } \notag \\
    & ~~~~~~~~ + \frac{1}{E_{k-p}^*} \frac{1}{(k_4 + i E_{k-p}^*)^2 +  E_p^2} \Biggr] \notag \\
    &= \int \frac{d^4 p}{(2 \pi)^4} \frac{1}{p^2 + M^2} \frac{1}{(k-p)^2 + (M^*)^2}.
\end{align}
This is what is calculated in \cite{Baulieu-etal} and take a form of the K\"all\'en-Lehmann spectral representation with a positive spectral density. Back to the real-time propagator, this implies $\braket{0|\mathcal{O} (y) \mathcal{O} (x) |0}$ has only real spectra. Thus, the composite operator $\mathcal{O} (x) $ could be regarded to be ``physical''.

\subsection{Complex singularities in a BRST quartet} \label{sec:BRST-complex}

Here, we discuss implications from the interpretation of complex singularities in an indefinite-metric state space in light of confinement.
As discussed above, complex singularities correspond to zero-norm states.
Such states, which are not physical, should be confined according to some confinement mechanism.

It is worthwhile considering implications in the Kugo-Ojima BRST quartet mechanism \cite{Kugo:1979gm}.
Here, we assume existence of a hermitian nilpotent BRST operator $Q_B$: $Q_B^2 = 0,~Q_B^\dagger = Q_B$. Some issues on this existence are mentioned in Sec.~\ref{sec:BRST-comment}.
In this scenario, confined states should belong to BRST quartets, i.e., BRST exact (BRST-daughter) or BRST non-invariant (BRST-parent) states. Thus, complex energy states, which lead to complex singularities of the propagators, should belong to BRST quartets.

In this section, we provide only a sketch of the argument.
Suppose that the gluon propagator has complex singularities.
Then, ``one-gluon state'' has complex energy states, which is schematically expressed as
\begin{align}
    A_\mu (0) \ket{0} = \ket{E} + \ket{E^*} + \cdots, \label{eq:one-gluon-state-complex}
\end{align}
where $\ket{E}$ and $\ket{E^*}$ stand for a pair of complex energy states, $\braket{E^*|E} \neq 0$.
Since $\ket{E}$ and $\ket{E^*}$ should be excluded from the physical state space constructed from the BRST cohomology $\operatorname{Ker} Q_B / \operatorname{Im} Q_B$ to make the theory physical, we require that $\ket{E}$ and $\ket{E^*}$ are either BRST exact or BRST non-invariant states
\footnote{Notice that, if the complex energy states are confined correctly,
asymptotic states in the physical state space are expected to be well-defined. Therefore, if such a confinement mechanism works well, the non-temperedness of Wightman function and the ill-definedness of the asymptotic states would not provide any physical issue.} .

We can easily exclude a possibility that both $\ket{E}$ and $\ket{E^*}$ are BRST exact. Indeed, if they were BRST exact: $\ket{E} = Q_B \ket{\gamma}$ and $\ket{E^*} = Q_B \ket{\gamma^*}$, then the non-orthogonality $\braket{E^*|E} \neq 0$ contradicts with the nilpotency of the BRST charge $Q_B$, $Q_B^2 = 0$.
Therefore, at least either $\ket{E} \notin \operatorname{Ker} Q_B$ or $\ket{E^*} \notin \operatorname{Ker} Q_B$ holds.

We assume further that CPT (anti-unitary) operator $\Theta$ exists and satisfies
\begin{align}
    &\Theta^2 = 1, ~~ \Theta \ket{0} = \ket{0}, ~~ \Theta Q_B \Theta = Q_B, \notag \\
    &\Theta H \Theta = H , ~~  \Theta A_\mu (0) \Theta = - A_\mu (0).
\end{align}
$\Theta A_\mu (0) \Theta = - A_\mu (0)$ and $ \Theta \ket{0} = \ket{0}$ implies
\begin{align}
    \Theta  \ket{E} = - \ket{E^*} , ~~  \Theta  \ket{E^*} = - \ket{E}. \label{eq:complex_energy_CPT}
\end{align}

When either $\ket{E} \notin \operatorname{Ker} Q_B$ or $\ket{E^*} \notin \operatorname{Ker} Q_B$ holds, the possibilities are (i) $\ket{E} \notin \operatorname{Ker} Q_B$ and $\ket{E^*} \in \operatorname{Im} Q_B$, (ii) $\ket{E} \in \operatorname{Im} Q_B$ and $\ket{E^*} \notin \operatorname{Ker} Q_B$, and (iii) $\ket{E} \notin \operatorname{Ker} Q_B$ and $\ket{E^*} \notin \operatorname{Ker} Q_B$.
The first two possibilities (i) and (ii) can be excluded by (\ref{eq:complex_energy_CPT}) and $Q_B \Theta = \Theta Q_B$, namely, $\ket{E} \in \operatorname{Ker} Q_B \Leftrightarrow \ket{E^*} \in \operatorname{Ker} Q_B$. Thus, the only possibility is (iii) both complex energy states are BRST non-invariant.

Hence, existence of CPT operator and nonexistence of complex energy states in the physical state space implies that both $\ket{E}$ and $\ket{E^*}$ should contain BRST parent states. In the simplest possibility, complex energy states form a double-BRST-quartet.

As a consequence, since $Q_B \ket{E} = \ket{E,c} \neq 0$ and $Q_B \ket{E^*} = \ket{E^*,c} \neq 0$, we have,
\begin{align}
    (D_\mu C)(0) \ket{0} = Q_B A_\mu (0) \ket{0} =  \ket{E,c} + \ket{E^*,c} + \cdots.
\end{align}
Since the ghost propagator seems to have no complex singularity according to recent analyses, e.g. \cite{Siringo16a, Siringo16b, HK2018, BT2019, Falcao:2020vyr, SFK12, Fischer-Huber}, this implies that the gluon-ghost bound state should contain complex energy states whose energies are equal to those of the gluon. Therefore, a propagator of the gluon-ghost bound state should have complex singularities at the same position as the gluon propagator.

Let us summarize the discussion above.
Complex energy states should be ``eliminated'' from the physical state space by some confinement mechanism. In the Kugo-Ojima scenario, they should be in BRST quartets.
For complex singularities in the gluon propagator, the ``one-gluon state'' should have complex conjugate energy states (\ref{eq:one-gluon-state-complex}), $\ket{E}$ and $\ket{E^*}$. The other discussion in this section can be summarized as the following claim.
\begin{claim} \label{claim:4}
Suppose that $\ket{E}$ and $\ket{E^*}$ of the ``one-gluon state'' with $\braket{E|E^*} \neq 0$ are in BRST quartets.
Then, either $\ket{E}$ or $\ket{E^*}$ is not a BRST daughter state. Moreover, with the additional assumption of the existence of the CPT operator, both $\ket{E}$ and $\ket{E^*}$ contain BRST parent states.
\end{claim}

This claim predicts that a propagator of the gluon-ghost bound state should have complex singularities at the same positions as those of the gluon propagator.



\section{Summary} \label{sec:summary}

Let us summarize our findings.
In Sec.~\ref{sec:prelim_disc}, we have presented a sketch of the discussion emphasizing that complex singularities of propagators on the complex squared momentum plane differ depending on whether the propagator is Euclidean one or Minkowski one.
This is an important remark for determining a starting point toward considering the reconstruction.
We have to regard ``complex singularities'' as those of Euclidean propagator and consider the reconstruction carefully.

The main part of this paper consists of Sec.~\ref{sec:complex-singularities}: general properties of Wightman functions and Sec.~\ref{sec:interpretation}: implications on state spaces.

In Sec.~\ref{sec:complex-singularities}, we have defined complex singularities and reconstructed Wightman functions from Schwinger functions with complex singularities.
We have obtained the following general properties on this reconstruction as stated in the introduction:
\begin{enumerate}
\renewcommand{\labelenumi}{(\Alph{enumi})}
    \item Violation of the reflection positivity of the Schwinger functions (Theorem \ref{thm:violation_ref_pos}),
    \item Holomorphy in the tube (Theorem \ref{thm:holomorphy_general}) and existence of the boundary value as a distribution (Theorem \ref{thm:complex_boundary_value}),
    \item Violation of the temperedness (Theorem \ref{thm:nontempered}) and the positivity condition in $\mathscr{D}(\mathbb{R}^4)$ (Theorem \ref{thm:violation_W-pos}), 
    \item Validity of Lorentz symmetry (Theorem \ref{thm:Lorentz} and Theorem \ref{thm:complex_Lorentz}) and spacelike commutativity  (Theorem \ref{thm:spacelike_commutativity}) 
\end{enumerate}
The organization of our proofs of these theorems is depicted in Fig.~\ref{fig:summary-2-b}. See Appendix \ref{sec:Appendix_which_axioms} for a summary of violated axioms.

In Sec.~\ref{sec:interpretation}, we have considered a possible state-space structure in the presence of complex singularities.
Consequently, a quantum mechanical observation (Sec.~\ref{sec:complex_spectra}) suggests that 
\begin{enumerate}
\setcounter{enumi}{4}
\renewcommand{\labelenumi}{(\Alph{enumi})}
    \item  complex singularities correspond to zero-norm states with complex energy eigenvalues.
\end{enumerate}
Indeed, we have firstly argued the necessity of non-real spectra by proving Claim \ref{claim:1}.
Secondly, Claim \ref{claim:2} implies that the complex-energy states have zero-norm and form complex conjugate pairs.
Thirdly, Claim \ref{claim:3}, which asserts that a pair of zero-norm eigenstates of complex conjugate energies yield a pair of complex conjugate poles in $(0+1)$-dimensional theory, indicates that complex singularities correspond to pairs of zero-norm eigenstates of complex conjugate energies.

Moreover, we have discussed an example of a relativistic QFT having propagators with complex poles which is called the Lee-Wick model. This model also supports the correspondence between complex singularities and pairs of zero-norm states.
Incidentally, we have argued that the field operator of the Lee-Wick model can be understood as a counterpart in the covariant operator formalism of so-called i-particle \cite{Baulieu-etal}.

Finally, we have discussed implications of complex singularities in the BRST formalism.
Under assumptions that the Kugo-Ojima quartet mechanism works well and that the $CPT$ operator exists,
we have argued that both complex conjugate energy states of the ``one-gluon state'' $A_\mu (0) \ket{0}$ contain BRST parent states. This predicts that complex singularities of a propagator of the gluon-ghost composite operator should appear at the same locations as those of the gluon propagator.

\section{Discussion} \label{sec:discussion}

In this section, some remarks are made on related topics.

\subsection{On other interpretation of complex singularities}
Let us make comments on another interpretation of complex singularities.
We have reconstructed Wightman functions from Schwinger functions based on (\ref{eq:W_S_connection_A}) and (\ref{eq:W_S_connection_B}). As remarked in Sec.~II, this is different from a naive inverse Wick rotation on the complex momentum plane. An interpretation using the inverse Wick rotation is often discussed, e.g., in \cite{Stingl85,HKRSW90, Stingl96}. In these references, it is claimed that complex poles lead to (a) short-lived gluonic particles, (b) no free-limits, (c) violation of causality (in short-range), (d) violation of reflection positivity, (e) asymptotic incompleteness, and (f) violation of unitarity (in short-range).

In our reconstruction method, there are some differences on (a) short-lived particle, (c) violation of causality, and (f) unitarity:
(a) Instead of finite lifetime, the reconstructed Wightman function grows exponentially. 
(c) The causality as the spacelike commutativity is kept as mentioned in Sec.~\ref{sec:locality}.
(f) The hermiticity of Hamiltonian can be consistent with complex poles in an indefinite metric state space as discussed in Sec.~\ref{sec:interpretation}.

\subsection{BRST symmetry, confinement, and complex singularities}\label{sec:BRST-comment}

Finally, let us add some comments on BRST symmetry and confinement in relation to complex singularities.

First, we have assumed a nilpotent BRST charge in Sec.~\ref{sec:BRST-complex}.
Since the Kugo-Ojima quartet mechanism is a promising way to construct the physical state space in gauge-fixed pictures, it would be natural to hope the existence of a nilpotent BRST charge.
However, a part of the evidence for complex singularities in the Landau-gauge gluon propagator relies on numerical lattice calculations in the minimal Landau gauge, where the usual BRST symmetry is not guaranteed.
At the present situation, the ``best-case scenario'' is that the gluon propagator of the minimal Landau-gauge would be a good approximation of some gauge with a nilpotent BRST symmetry.
Developing the Lattice Landau gauge preserving the standard BRST symmetry in the continuum limit overcoming the Neuberger zero \cite{Hirschfeld, Lattice-BRST, Neuberger} would be an important future prospect.

Second, since complex singularities cause a problem on the asymptotic completeness as mentioned in Sec.~\ref{sec:temperedness} in the ``full'' state space, the Kugo-Ojima arguments could be modified. It would be interesting to explore this possibility.

Third, there are few theoretical developments of the axiomatic method without the spectral condition and positivity to our knowledge. Such studies could yield some constraints on complex singularities and are therefore interesting.

Fourth, Claim \ref{claim:4} predicts complex gluon-ghost bound states with the same energy as that of the gluon.
Conversely, appearance of complex singularities in a propagator of the gluon-ghost composite operator would be a necessary condition for the BRST formalism to ``work well'' if the gluon propagator has complex singularities.
Thus, seeking such complex gluon-ghost bound state would be interesting. Remarkably, the Bethe-Salpeter equation for the gluon-ghost bound state has been studied in light of BRST quartets in {\cite{Alkofer:2011pe}}.

Fifth, while one can expect that complex singularities of field correlators have something to do with a confinement mechanism, we ought to note that complex singularities could be trivial gauge-artifacts.
Although the complex singularities yield a violation of (reflection) positivity, this violation is itself neither necessary nor sufficient for the confinement of a particle corresponding to the field, e.g., the gluon confinement.
Indeed, this is not sufficient because this violation only indicates that the field involves some negative-norm states and does not deny the existence of asymptotic physical states.
This violation is not a necessary condition because BRST-parent states can be positive-norm, for example.
Similarly, although complex singularities correspond to confined states, their existence is neither necessary nor sufficient for the confinement of the corresponding particle.
Moreover, such ``confined states'' corresponding to complex singularities could only be members of BRST quartets that are irrelevant to the confinement mechanism like the timelike photon.
There are still many possibilities because understanding a confining theory as a quantum theory is far from being achieved.
Further studies are needed for clarification of relations between complex singularities and confinement mechanism.

\section*{Acknowledgements}
We thank Taichiro Kugo, Peter Lowdon, and Lorenz von Smekal for helpful and critical comments in early stage of this work.
Y.~H. is supported by JSPS Research Fellowship for Young Scientists Grant No.~20J20215, and K.-I.~K. is supported by Grant-in-Aid for Scientific Research, JSPS KAKENHI Grant (C) No.~19K03840.

\appendix

\section{Notations and axioms} \label{sec:notations_axioms}

In this section, we introduce notations required for mathematical discussions and review the standard Osterwalder-Schrader axiom for Euclidean field theories \cite{OS73}.

\subsection{Notations and conventions}
We use the notation $x = (\vec{x},x^4) = (x^1,x^2,x^3,x^4)$ for a four-vector and Euclidean inner product $xy = x^\mu y^\mu$ (and Lorentzian inner product only when explicitly mentioned).
When only one four-vector is relevant as in the main text, we also use the lower indices $x = (x_1,x_2,x_3,x_4)$.
We call the direction of $e_4 := (\vec{0},1)$ ``(imaginary-)time direction''. We also use the multi-index notation: for a multi-index $\alpha = (\alpha_{1,1},\alpha_{1,2}, \cdots, \alpha_{n,4}) \in \mathbb{Z}_{\geq 0}^{4n}$, $D_\alpha$ denotes
\begin{align}
    D_\alpha = \frac{\partial^{|\alpha|}}{(\partial x_1^1)^{\alpha_{1,1}} (\partial x_1^2)^{\alpha_{1,2}} \cdots (\partial x_n^4)^{\alpha_{n,4}}},
\end{align}
where $|\alpha| = \alpha_{1,1} + \cdots + \alpha_{n,4}$.

The Schwartz's space on $\mathbb{R}^n$ is denoted by $\mathscr{S}(\mathbb{R}^n)$.
Its dual space, the space of tempered distributions, is denoted by $\mathscr{S}'(\mathbb{R}^n)$.
We also define $\mathscr{D}(\mathbb{R}^n) := \{ f(\xi) ~;~  f(\xi)\mathrm{~is~a~} C^\infty\mathrm{~function~with~a~compact~support}  \}$ and its dual space $\mathscr{D}'(\mathbb{R}^n)$.
We can regard $\mathscr{S}'(\mathbb{R}^4) \subset \mathscr{D}'(\mathbb{R}^4)$.
An element of $\mathscr{D}'(\mathbb{R}^4)$ can be beyond polynomial growth unlike $\mathscr{S}'(\mathbb{R}^4)$.
An element of $\mathscr{D}'(\mathbb{R}^4)$ that cannot be regarded as a tempered distribution is called a non-tempered distribution.

\begin{widetext}
The important test function spaces are listed as follows. These spaces are equipped with the topologies in the same way as \cite[Sec. 2]{OS73}.
\begin{enumerate}
    \item Space of test functions on non-coincident points
\begin{align}
    ^0 \mathscr{S} (\mathbb{R}^{4n}) := \left\{ f \in \mathscr{S} (\mathbb{R}^{4n}) ~;~
    \begin{gathered}
    D^\alpha f(x_1,\cdots,x_n) = 0 ~\mathrm{for~any}~ \alpha \in \mathbb{Z}_{\geq 0}^{4n} \\
    ~\mathrm{if}~ x_i = x_j ~\mathrm{for~some}~ 1 \leq i < j \leq n
\end{gathered}
    \right\},
\end{align}
    \item Space of test functions with (imaginary-)time-ordered supports
\begin{align}
    \mathscr{S}_+ (\mathbb{R}^{4n}) &= \left\{ f \in \mathscr{S} (\mathbb{R}^{4n}) ~;~
    \begin{gathered}
    D^\alpha f(x_1,\cdots,x_n) = 0 ~\mathrm{for~any}~ \alpha \in \mathbb{Z}_{\geq 0}^{4n} \\
    ~\mathrm{unless}~ 0<x_1^4<x_2^4< \cdots < x_n^4
\end{gathered}
    \right\}, \\
    \mathscr{S}_< (\mathbb{R}^{4n}) &= \left\{ f \in \mathscr{S} (\mathbb{R}^{4n}) ~;~
    \begin{gathered}
    D^\alpha f(x_1,\cdots,x_n) = 0 ~\mathrm{for~any}~ \alpha \in \mathbb{Z}_{\geq 0}^{4n} \\
    ~\mathrm{unless}~ x_1^4<x_2^4< \cdots < x_n^4
\end{gathered}
    \right\},
\end{align}
    \item Space of test functions with supports of positive (imaginary-)time. 
    
    For functions of one-variable, $\mathscr{S}_+(\mathbb{R}) := \{ f(s) \in \mathscr{S}(\mathbb{R})~;~ \operatorname{supp} f \subset [0,\infty) \}$ and also $\mathscr{S}_-(\mathbb{R}) := \{ f(s) \in \mathscr{S}(\mathbb{R})~;~ \operatorname{supp} f \subset (-\infty,0] \}$. We define
\begin{align}
    \mathscr{S} (\mathbb{R}^{4}_+) &:= \mathscr{S} (\mathbb{R}^3) \hat{\otimes} \mathscr{S}_+(\mathbb{R}) , ~~~     \mathscr{S} (\mathbb{R}^{4n}_+) := \hat{\otimes}^n \mathscr{S} (\mathbb{R}^{4}_+),
\end{align}
where $\hat{\otimes}$ denotes the completed topological tensor product and $\hat{\otimes}^n$ the n-fold one.
    \item Space of test functions on ``non-negative (imaginary-)time''.
    
    $\mathscr{S}(\bar{\mathbb{R}}_+)$ denotes the space of test functions on the non-negative real half line: $\mathscr{S}(\bar{\mathbb{R}}_+) := \mathscr{S}(\mathbb{R}) / \mathscr{S}_-(\mathbb{R})$. Note that its dual space can be identified as $\mathscr{S}' (\bar{\mathbb{R}}_+) \simeq \{ F \in \mathscr{S}'(\mathbb{R}) ~;~ \operatorname{supp} F \subset [0,\infty) \}$. We define as above
\begin{align}
    \mathscr{S} (\bar{\mathbb{R}}^{4}_+) &:= \mathscr{S} (\mathbb{R}^3) \hat{\otimes} \mathscr{S}(\bar{\mathbb{R}}_+) , ~~~     \mathscr{S} (\bar{\mathbb{R}}^{4n}_+) := \hat{\otimes}^n \mathscr{S} (\bar{\mathbb{R}}^{4}_+),
\end{align}   
\end{enumerate}

We introduce the sets of terminating sequences $\underline{\mathscr{S}}, ~\underline{\mathscr{S}}_+,~\underline{\mathscr{S}}_<$, and $\underline{\mathscr{S}}(\bar{\mathbb{R}}^{4}_+)$ over the spaces $\mathscr{S} (\mathbb{R}^{4n}),~ \mathscr{S}_+ (\mathbb{R}^{4n}),~\mathscr{S}_< (\mathbb{R}^{4n})$, and $\mathscr{S}(\bar{\mathbb{R}}^{4n}_+)$, respectively. An element $\underline{f}$ of one of the spaces $\underline{\mathscr{S}}_*~ (= \underline{\mathscr{S}}, ~\underline{\mathscr{S}}_+,~\underline{\mathscr{S}}_<,~\underline{\mathscr{S}}(\bar{\mathbb{R}}^{4}_+)  )$ over $\mathscr{S}^n_* (=  \mathscr{S} (\mathbb{R}^{4n}),~ \mathscr{S}_+ (\mathbb{R}^{4n}),~\mathscr{S}_< (\mathbb{R}^{4n}) ,~ \mathscr{S} (\bar{\mathbb{R}}^{4n}_+) )$ is a terminating sequence $\underline{f} := (f_0, f_1, \cdots)$ with $f_0 \in \mathbb{C},~ f_n \in \mathscr{S}^n_* ~~(n = 1,2,\cdots)$, i.e.,
\begin{align}
    \underline{\mathscr{S}}_* := \bigoplus_{n=1}^\infty \mathscr{S}^n_*,   
\end{align}   
with $\mathscr{S}^0_* := \mathbb{C}$.

Next, we define some operations $\times$, $\cdot^\star$, and $\Theta$ on these spaces. 
\begin{enumerate}
\renewcommand{\labelenumi}{(\alph{enumi})}
    \item For $\underline{f} = (f_0,f_1,\cdots),~\underline{g}= (g_0,g_1,\cdots) \in \underline{\mathscr{S}}_*$, $\underline{f} \times \underline{g} $ is defined as 
\begin{align}
    \underline{f} \times \underline{g} &= ((\underline{f} \times \underline{g})_0, (\underline{f} \times \underline{g})_1,(\underline{f} \times \underline{g})_2,\cdots), \notag \\
    (\underline{f} \times \underline{g})_n (x_1,x_2,\cdots,x_n) &= \sum_{k = 0}^n (f_{n-k} \times g_k) (x_1,x_2,\cdots,x_n) \notag \\
    &=  \sum_{k = 0}^n f_{n-k} (x_1,x_2,\cdots,x_{n-k})  g_k (x_{n-k+1},x_{n-k+2},\cdots,x_n)  ,
\end{align}
    \item For $\underline{f} = (f_0,f_1,\cdots) \in \underline{\mathscr{S}}_*$, 
\begin{align}
    &\underline{f}^\star = (f_0^\star,f_1^\star,\cdots), ~~ f_n^\star(x_1,x_2,\cdots,x_n) := \bar{f}_n (x_n,x_{n-1},\cdots,x_1), \\
    &\Theta \underline{f} = ((\Theta \underline{f})_0, (\Theta \underline{f})_1,\cdots), ~~~ (\Theta \underline{f})_n (x_1,x_2,\cdots,x_n) := f_n (\vartheta x_1,\vartheta x_2,\cdots, \vartheta x_n),
\end{align}
where $\bar{\cdot}$ is complex conjugation in this appendix (to distinguish from $\cdot^\star$) and $\vartheta x = (\vec{x}, - x^4)$. In the main text, the complex conjugation is denoted by $\cdot^*$.
    \item For an element of the Euclidean group $(a,R) \in \mathbb{R}^4 \rtimes SO(4)$ and $f \in \mathscr{S}^n_*$,
\begin{align}
    f_{(a,R)} (x_1,\cdots,x_n) := f(Rx_1 + a, \cdots, Rx_n + a).
\end{align}
\end{enumerate}

For the spectral function, we define tempered distribution on a compactified set $[0, \infty]$ \cite[Sec. A.3.]{Bogolyubov:1990kw}.
We introduce the space of test functions on $[0, \infty]$ as
\begin{align}
    \mathscr{S}([0,\infty]) := \left\{ f(\lambda) = g(- (1+\lambda)^{-1}) ~;~g\mathrm{~is~a~}C^\infty \mathrm{~function~on~}[-1,0] \right\},
\end{align}
equipped with the topology characterized by the countable norm family\footnote{Note that this norm can be written in terms of $g(u)$ on $[-1,0]$ by identifying $u = - \frac{1}{1 + \lambda}$ as $\| f(\lambda) \|^{[0,\infty]}_n = \max_{\substack{k \in \{0,1,\cdots,n \} \\ u \in [-1,0]}} | \frac{\partial^k g}{\partial u^k} | $, which is clearly finite for $f \in  \mathscr{S}([0,\infty])$.} $\| f(\lambda) \|^{[0,\infty]}_n := \max_{k \in \{0,1,\cdots,n \}} \sup_{\lambda \geq 0} | \left( (1+\lambda)^2 \frac{\partial}{\partial \lambda}\right)^k f(\lambda) |$ for $n \in \mathbb{Z}_{\geq 0}$. Its dual space, namely the space of continuous linear functions of $\mathscr{S}([0,\infty])$, is denoted by $\mathscr{S}'([0,\infty])$. Elements of this space are called \textit{tempered distributions on $[0, \infty]$}.
With this definition, for $\rho(\sigma^2) \in \mathscr{S}'([0,\infty])$, the ``integral'' $\int_0^\infty d\sigma^2 \frac{\rho(\sigma^2)}{k^2 + \sigma^2}$ is formally well-defined.

\end{widetext}

\subsection{Osterwalder-Schrader Axioms}

Using the above notations, we state the standard Osterwalder-Schrader Axioms, for simplicity, for the scalar field. $\{ \mathcal{S}_n \}_{n=0}^\infty$ is a sequence of distributions $\mathcal{S}_n (x_1,\cdots,x_n)$, called Schwinger functions, satisfying
\begin{enumerate}
\renewcommand{\labelenumi}{[OS \arabic{enumi}]}
\setcounter{enumi}{-1}
    \item (\textit{Temperedness}):
\begin{align}
    \mathcal{S}_0 = 1, ~~~ \mathcal{S}_n \in  ~^0 \mathscr{S}' (\mathbb{R}^{4n}).
\end{align}
    \item (\textit{Euclidean Invariance}): for all $(a,R) \in \mathbb{R}^4 \rtimes SO(4)$ and $f \in  ~^0 \mathscr{S} (\mathbb{R}^{4n})$,
\begin{align}
    \mathcal{S}_n (f) = \mathcal{S}_n (f_{(a,R)}).
\end{align}
    \item (\textit{Reflection Positivity}): for all $\underline{f} = (f_0,f_1 \cdots) \in \underline{\mathscr{S}}_+$,
\begin{align}
    \sum_{n, m} \mathcal{S}_{n + m} (\Theta f_n^\star \times f_m) \geq 0.
\end{align}
    \item (\textit{Symmetry})
\begin{align}
    \mathcal{S}_n (x_1, \cdots, x_n) = \mathcal{S}_n (x_{\pi(1)} , \cdots, x_{\pi(n)}),
\end{align}
for any pertumutation $\pi(\cdot)$ of $n$ items
    \item (\textit{Cluster Property}): for all $\underline{f} = (f_0,f_1 \cdots), ~ \underline{g} = (g_0,g_1 \cdots) \in \underline{\mathscr{S}}_+$ and $a = (\vec{a},0)$,
\begin{align}
    \lim_{\lambda \rightarrow \infty} \sum_{n,m} \left[ \mathcal{S}_{n + m} (\Theta f_n^\star \times g_{m, (\lambda a , 1)}) - \mathcal{S}_{n} (\Theta f_n^\star) \mathcal{S}_{m} (g_m)  \right] = 0.
\end{align}

\renewcommand{\labelenumi}{[OS \arabic{enumi}']}
\setcounter{enumi}{-1}
    \item (\textit{Laplace Transform Condition})\footnote{Contrary to the original expectation, temperedness of the (higher-point) Wightman functions would not be guaranteed by [OS0] -- [OS4] \cite{Glaser:1974hy, OS75}.
    For the two-point sector, this condition is irrelevant since the temperedness of the Wightman function can be derived by the other conditions. However, it should be noted that complex singularities also violate this condition.
    Note also that this condition can be replaced by, e.g. , a slight stronger condition, ``linear-growth condition'', which controls growth of $\mathcal{S}_n$ in $n$ \cite{OS75}. Since the aim of imposing this condition is to control behavior of the higher-point functions, we will not go far into this condition in this paper.}: From the translational invariance [OS1], we can write $\mathcal{S}_n(x_1, \cdots, x_n) $ as $S_{n-1}(\xi_1 , \cdots, \xi_{n-1}) \in \mathscr{S}' (\mathbb{R}^{4(n-1)}_+)$, i.e., $\mathcal{S}_n(x_1, \cdots, x_n) =: S_{n-1} (x_2 - x_1, x_3 - x_2, \cdots, x_n - x_{n-1})$ for $x_1^4 < x_2^4 < \cdots < x_n^4$. This condition means that there exists, for every $n$, a Schwarz seminorm $\| \cdot \|_\mathscr{S}$ on $\mathscr{S}(\bar{\mathbb{R}}^{4(n-1)}_+)$ so that
\begin{align}
    |S_{n-1} (f)| \leq \| f^L \|_\mathscr{S}~~~\mathrm{for~all}~f \in \mathscr{S} (\mathbb{R}^{4(n-1)}_+),
\end{align}
where $f^L$ denotes the Laplace-Fourier transform defined by
\begin{align}
    &f^L (q_1, \cdots,q_{n-1}) \notag \\
    &:= \int d^{n-1} \xi~ f(\xi_1, \cdots, \xi_{n-1})  e^{\sum_{k = 1} ^n ( - q^4_k x^4_k + i \vec{q}_k \cdot \vec{x}_k  )} \Bigr|_{q_k^4 \geq 0} \notag \\
    &~~\in \mathscr{S}(\bar{\mathbb{R}}^{4(n-1)}_+)
\end{align}

\end{enumerate}

Let us comment on the standard axiom of Euclidean field theories.
\begin{enumerate}
\renewcommand{\labelenumi}{(\alph{enumi})}
    \item For $\underline{f}, \underline{g} \in \underline{\mathscr{S}}_+$, $\Theta \underline{f}^\star \times \underline{g} \in \underline{\mathscr{S}}_<$, which appears in [OS2] and [OS4].
    \item As a special case of [OS2], $\underline{f} = (0,f,0,0, \cdots) \in \underline{\mathscr{S}}_+$, we have the reflection positivity for the two-point function: For any $f(x) \in \mathscr{S}_+ (\mathbb{R}^{4}) = \{ f \in \mathscr{S} (\mathbb{R}^{4}) ~;~ \operatorname{supp} f \subset \mathbb{R}^3 \times [0,\infty) \}$,
\begin{align}
    \int d^4x d^4y ~ \bar{f}(\vartheta x) f(y) \mathcal{S}_2 (x,y) \geq 0, \label{eq:two_pt_ref_pos}
\end{align}
which is equivalent to, in terms of $S_1(y-x) := \mathcal{S}_2 (x,y)$ and Fourier transforms of $S_1$ and $f$ in the spatial directions,
\begin{align}
    \int d x^4 dy^4 \int \frac{d^3 \vec{p}}{(2 \pi)^3} ~ \bar{f}(\vec{p},x^4) f(\vec{p}, y^4) S_1 (\vec{p}, -x^4 - y^4) \geq 0. \label{eq:two_pt_ref_pos_mom}
\end{align}
    \item Usually, $\underline{S}_1 (\vec{p}, \tau)$ is an ordinary function. In this case, the reflection positivity is often checked by a necessary condition:
\begin{align}
\underline{S}_1 (\vec{p}, \tau) \geq 0 ~~~ \mathrm{for~all}~\tau >0,~\vec{p}. \label{eq:two_pt_ref_pos_concise}
\end{align}
If not, there exists $(\vec{p}_*,\tau_*)$ such that $\underline{S}_1 (\vec{p}_*, \tau_*) < 0$, and we can choose a test function $f(\vec{p},\tau)$ with its support sufficiently close to $(\vec{p}_*,\tau_* /2)$ so that the left-hand side of (\ref{eq:two_pt_ref_pos_mom}) is negative. Observing the sign of $ \underline{S}_1(\vec{p}, \tau)$ is relatively easy but is not enough to test the reflection positivity completely even in the two-point sector. For example, a propagator with complex poles and largely positive spectral function will not show the negativity of $\underline{S}_1 (\vec{p}, \tau)$, while the reflection positivity itself is violated as proven in Theorem \ref{thm:violation_ref_pos}.
\begin{widetext}
    \item The Schwinger function $S_{n-1}(\xi_1 , \cdots, \xi_{n-1}) \in \mathscr{S}' (\mathbb{R}^{4(n-1)}_+)$, i.e., $\mathcal{S}_n(x_1, \cdots, x_n) =: S_{n-1} (x_2 - x_1, x_3 - x_2, \cdots, x_n - x_{n-1})$ for $x_1^4 < x_2^4 < \cdots < x_n^4$ is regarded as the ``values'' of the Wightman function at pure imaginary times, or Euclidean points\footnote{More generally, the Schwinger functions at non-coincident points are regarded as the ``values'' at Euclidean points of the holomorphic Wightman function defined on the permuted extended tube. See \cite[Sec 4.5]{OS73}}
\begin{align}
 W_{n-1} ((-i\tau_1 ,\vec{\xi}_1),(-i\tau_2,\vec{\xi}_2),\cdots, (-i\tau_{n-1},\vec{\xi}_{n-1})) = S_{n-1} ((\vec{\xi}_1 , \tau_1),(\vec{\xi}_2, \tau_2),\cdots, (\vec{\xi}_{n-1}, \tau_{n-1})). \label{eq:W_S_connection_A}
\end{align}
One expects that the Wightman function is holomorphic in the (extended) holomorphic tube and that the holomorphic Wightman function provides the vacuum expectation value of the fields as its boundary value as the usual case \cite{Streater:1989vi}:
\begin{align}
 W_{n-1} (\xi_1,\cdots, \xi_{n-1} ) = \lim_{\substack{\eta_1,\cdots \eta_{n-1} \rightarrow 0 \\ \eta_1,\cdots \eta_{n-1} \in V_+}}W_{n-1} (\xi_1 - i \eta_1,\cdots, \xi_{n-1} - i \eta_{n-1} ),
\end{align}
where $V_+$ denotes the forward lightcone. Therefore, the Wightman function is reconstructed from the Schwinger function analytically continued to $\operatorname{Re} \tau_k > 0$:
\begin{align}
 W_{n-1}& ((t_1 ,\vec{\xi}_1),(t_2,\vec{\xi}_2),\cdots, (t_{n-1},\vec{\xi}_{n-1})) \notag \\
 &= \lim_{\tau_1, \cdots, \tau_{n-1} \rightarrow +0} S_{n-1} ((\vec{\xi}_1 , \tau_1 + it_1),(\vec{\xi}_2, \tau_2 + it_2),\cdots, (\vec{\xi}_{n-1}, \tau_{n-1} + it_{n-1})), \label{eq:W_S_connection_B}
\end{align}
since $((-i(\tau_1 + it_1), \vec{\xi}_1),(-i(\tau_2 + it_2), \vec{\xi}_2),\cdots, (-i(\tau_{n-1} + it_{n-1}), \vec{\xi}_{n-1})$ for $\tau_1, \cdots, \tau_{n-1} > 0$ is in the tube $\mathbb{R}^{4 (n-1)} - i V_+^{n-1}$. Note that the spacelike commutativity of the Wightman functions follows from the permutation symmetry [OS3] and expected holomorphy of the analytically continued Schwinger functions $S_{n-1}$ in the extended tube.
\end{widetext}
    \item A generalization of the Osterwalder-Schrader axiom without the reflection positivity was proposed in \cite{Jakobczyk:1987zw}. However, the new axiom (``Hilbert space structure condition'' with ``$\underline{\mathscr{S}}$-continuity'', where the latter is introduced for a convenient purpose) is strong enough to derive the Laplace transform condition and prohibit complex singularities. 
\end{enumerate}

\section{Proof of violation of the reflection positivity in the presence of complex singularities} \label{sec:Appendix_ref_pos}

For the proof of violation of the reflection positivity (Theorem \ref{thm:violation_ref_pos}), the goal of this section is to prove that the reflection positivity leads to temperedness of a reconstructed two-point Wightman function, or Theorem \ref{thm:ref_pos_yields_temperedness}. 
Consequently, violation of the reflection positivity in the presence of complex singularities follows from the non-temperedness (Theorem \ref{thm:nontempered}).

This proof is essentially a simplified version of the steps (a) and (b) of the Osterwalder-Schrader Theorem \cite{OS73}. 

\begin{lemma}
\label{lem:OS_analytic_cont}

Suppose that the two-point Schwinger function $\mathcal{S}_2$ satisfies 
\begin{enumerate}
    \item temperedness $\mathcal{S}_2 \in  ~^0 \mathscr{S}' (\mathbb{R}^{4 \cdot 2})$,
    \item translational invariance $\mathcal{S}_2 (x_1 + a, x_2 + a) = \mathcal{S}_2(x_1 , x_2)$ for all $a \in \mathbb{R}^4$, and
    \item the reflection positivity for the two-point sector (\ref{eq:two_pt_ref_pos}),
\end{enumerate}
which follow from [OS0] temperedness, [OS1] Euclidean invariance, and [OS2] reflection positivity, respectively.

Then, $S_1 (x_2 - x_1) := \mathcal{S}_n(x_1, x_2)$ (after smearing the spatial directions) can be analytically continued to the right-half plane ($\operatorname{Re} (x_2^4 - x_1^4) > 0$), and its analytic continuation can be regarded as a tempered distribution on the half-plane and the spatial directions. More precisely, for any $h(\vec{\xi}) \in \mathscr{S} (\mathbb{R}^{3})$,
there exists a holomorphic function $S_1( \tau + is | h)$ on the right-half plane ($\tau > 0$) such that 
\begin{enumerate}
\renewcommand{\labelenumi}{(\alph{enumi})}
    \item On the real axis, $S_1( \tau| h) = \int d^3 \vec{\xi}~ S_1(\vec{\xi},\tau) h(\vec{\xi})$,
    \item $S_1 (\tau + is |h)$ can be regarded as an element of $ [\mathscr{S} (\mathbb{R}_+) \hat{\otimes} \mathscr{S} (\mathbb{R})]'$, the dual space of $\mathscr{S} (\mathbb{R}_+) \hat{\otimes} \mathscr{S} (\mathbb{R})$,
    \item $S_1 (\tau + is |h)$ is continuous in $h(\vec{\xi}) \in \mathscr{S} (\mathbb{R}^{3})$.
\end{enumerate}
\end{lemma}

\begin{proof}

First of all, using (i) temperedness and (ii) translational invariance, there exists $S_{1}(\xi) \in \mathscr{S}' (\mathbb{R}^{4}_+)$ such that $\mathcal{S}_2 (x_1, x_2) = S_1 (x_2 - x_1)$ for $x_1^4 < x_2^4$.

We shall prove the claim with the following steps:
\begin{enumerate}
\renewcommand{\labelenumi}{Step \arabic{enumi}.}
\item constructing a Hilbert space with the form (\ref{eq:two_pt_ref_pos}),
\item defining spatial and (imaginary-)temporal translation operator and ``Hamiltonian,'' 
\item analytic continuation using the holomorphic semigroup generated by the ``Hamiltonian.''
\end{enumerate}

\textit{Step 1. defining a Hilbert space with the form (\ref{eq:two_pt_ref_pos}).}

Let us first begin with constructing a Hilbert space. For $f,~g \in \mathscr{S}_+ (\mathbb{R}^{4})$, we define a sesquilinear form on $\mathscr{S}_+ (\mathbb{R}^{4}) \times \mathscr{S}_+ (\mathbb{R}^{4})$ by, 
\begin{align}
    (f,g) := \mathcal{S}_2 (\Theta f^\star \times g) = \int d^4x d^4y~ \bar{f}(\vartheta x) g(y) S_1 (y-x), \label{eq:sesqui_form_on_S}
\end{align}
which is positive semi-definite: $\| f \|^2 := (f,f) \geq 0$ from (\ref{eq:two_pt_ref_pos}).
We introduce $\mathscr{N}$ as the space of all zero norm vectors, i.e.
\begin{align}
\mathscr{N} := \{ f \in \mathscr{S}_+ (\mathbb{R}^{4}) ~;~ \| f \|^2 = 0 \}
\end{align}
We then obtain a pre-Hilbert space $\mathscr{S}_+ (\mathbb{R}^{4}) / \mathscr{N}$ and denote its completion by $\mathscr{K}$.
Therefore, $\mathscr{K}$ is a Hilbert space and contains $\mathscr{S}_+ (\mathbb{R}^{4}) / \mathscr{N}$ as a dense subset $\mathscr{D}_0$, namely $\mathscr{S}_+ (\mathbb{R}^{4}) / \mathscr{N} \simeq \mathscr{D}_0 \subset \mathscr{K}$.

We denote the (continuous) natural map by $v : \mathscr{S}_+ (\mathbb{R}^{4}) \rightarrow \mathscr{K}$, whose image is $\mathscr{D}_0$, and the inner product on $\mathscr{K}$ by $(\cdot , \cdot)_\mathscr{K}$. It follows that, for $f,~g \in \mathscr{S}_+ (\mathbb{R}^{4})$,
\begin{align}
(v(f) , v(g))_\mathscr{K} = (f,g).
\end{align}

\textit{Step 2. constructing translation operators and ``Hamiltonian.''}

Next, we define translational operators on $\mathscr{K}$.

For spatial directions, we define $\hat{U}_s(\vec{a})$ on $\mathscr{S}_+ (\mathbb{R}^{4})$ by
\begin{align}
    (\hat{U}_s(\vec{a}) f) (x) := f(x - a),
\end{align}
for $a = (\vec{a},0)$ and $f \in \mathscr{S}_+ (\mathbb{R}^{4})$. Note that $(  \hat{U}_s(\vec{a}) f,  \hat{U}_s(\vec{a}) g ) = (f,g)$. We then define the unitary operator $U_s(\vec{a})$ on $\mathscr{K}$ by a continuous extension of $U_s(\vec{a})$ defined on $\mathscr{D}_0$,
\begin{align}
    U_s(\vec{a}) v(f) := v \left( \hat{U}_s(\vec{a}) f \right).
\end{align}

Similarly, for $\tau >0$, we define $\hat{T}^\tau$ on $\mathscr{S}_+ (\mathbb{R}^{4})$ by
\begin{align}
    (\hat{T}^\tau f) (x) := f(\vec{x}, x^4 - \tau).
\end{align}
Note that $\tau >0$ is necessary to guarantee $\operatorname{supp} (\hat{T}^\tau f) \subset \mathbb{R}^3 \times [0,\infty)$. Recalling (\ref{eq:sesqui_form_on_S}), we have
\begin{align}
    (\hat{T}^\tau f, g)  = ( f, \hat{T}^\tau g),   \label{eq:hat_T_tau_symm}
\end{align}
for $f,~g \in \mathscr{S}_+ (\mathbb{R}^{4})$ and $\tau \geq 0$.

Next, let us derive the following bound: for any $\tau >0$, $f \in \mathscr{S}_+ (\mathbb{R}^{4})$,
\begin{align}
    (f, \hat{T}^\tau f) \leq (f,f) = \| f \|^2. \label{eq:hat_T_tau_contraction}
\end{align}
Beforehand, we shall see that $(f, \hat{T}^\tau f)$ grows at most polynomially in $\tau$. Indeed, by the definition (\ref{eq:sesqui_form_on_S}),
\begin{align}
    (f, \hat{T}^\tau f) = \int d^4x d^4y~ \bar{f}(x) f(y) S_1 (\vec{y} - \vec{x}, \tau + x^4 + y^4),
\end{align}
which shows $(f, \hat{T}^\tau f)$ increases at most polynomially as $\tau \rightarrow \infty$
because of the temperedness of $S_{1}(\xi) \in \mathscr{S}' (\mathbb{R}^{4}_+)$. Then, by using the Cauchy-Schwarz inequality and (\ref{eq:hat_T_tau_symm}) recursively, we have
\begin{align}
     (f, \hat{T}^\tau f) &\leq (f,f)^{1/2} (\hat{T}^\tau f, \hat{T}^\tau f)^{1/2}
     = (f,f)^{1/2} (f, \hat{T}^{2 \tau} f)^{1/2} \notag \\
     &\leq (f,f)^{1/2 + 1/4} (\hat{T}^{2 \tau} f, \hat{T}^{2 \tau} f)^{1/4} \notag \\
     &= (f,f)^{1/2 + 1/4} (f, \hat{T}^{4 \tau} f)^{1/4} \leq \cdots \notag \\
     &\leq (f,f)^{\frac{1}{2} + \frac{1}{4} + \cdots +  \frac{1}{2^n}} \exp \left[ 2^{-n} \ln  (f, \hat{T}^{2^n \tau} f) \right], \label{eq:bound_f_hatTtau_f_derivation}
\end{align}
for all positive integer $n$, positive real $\tau > 0$, and $f \in \mathscr{S} (\mathbb{R}^{4}_+)$. As $n \rightarrow \infty$, $ 2^{-n} \ln  (f, \hat{T}^{2^n \tau} f) \rightarrow 0$ due to (at most) linear growth of $\ln  (f, \hat{T}^{2^n \tau} f)$ in $n$.
Therefore, the $n \rightarrow \infty$ limit of (\ref{eq:bound_f_hatTtau_f_derivation}) gives the bound (\ref{eq:hat_T_tau_contraction}).

In particular, for any $f \in \mathscr{N}$, $\hat{T}^\tau f$ is also zero-norm $\hat{T}^\tau f \in \mathscr{N}$, since
\begin{align}
    (\hat{T}^\tau f, \hat{T}^\tau f) = (f, \hat{T}^{2 \tau} f) \leq (f,f) = 0.
\end{align}
Thus the natural map of $\hat{T}^\tau$ on $\mathscr{S} (\mathbb{R}^{4}_+) / \mathscr{N}$ is well-defined.
We define $T_0 ^\tau$ to be the operator defined on $\mathscr{D}_0$:
\begin{align}
    T_0 ^\tau v(f) := v ( \hat{T}^\tau f ).
\end{align}

So far, $T_0 ^\tau$ is
defined on the dense domain $\mathscr{D}_0$,
symmetric from (\ref{eq:hat_T_tau_symm}),
and bounded from (\ref{eq:hat_T_tau_contraction}). Then, we can extend $T_0^\tau$ to be defined on the whole $\mathscr{K}$ by continuity and have a self-adjoint contraction\footnote{From (\ref{eq:hat_T_tau_contraction}), the operator norm of $T^\tau$ is less than or equal to 1: $\| T^\tau \|_{op.} \leq 1$. A bounded operator with this property is called contraction.} $T^\tau$ on $\mathscr{K}$.
Note that the semigroup $\{ T^\tau \}_{\tau \geq 0}$ is strongly continuous due to (1) the boundedness $\| T^\tau \|_{op.} \leq 1$ from ({\ref{eq:hat_T_tau_contraction}}) and (2) continuity for points in $\mathscr{D}_0$, $\lim_{\tau \downarrow 0} \|T^\tau v(f) - v(f) \| = 0$ from the definition ({\ref{eq:sesqui_form_on_S}}).

Let us define the infinitesimal generator of the semigroup $\{ T^\tau \}_{\tau \geq 0}$, ``Hamiltonian'', by $H$. Formally\footnote{In terms of the original space $\mathscr{S} (\mathbb{R}^{4}_+)$, $H$ can be regarded as $- \partial_4 = - \frac{\partial}{\partial x^4}$. Note that the reflection $\vartheta$ in (\ref{eq:sesqui_form_on_S}) makes $- \frac{\partial}{\partial x^4}$ hermitian. More precisely, $H$ can be defined on the dense domain $\mathscr{D}_0$, and
\begin{align}
    H v(f) = v (- \partial_4 f),
\end{align}
for $f \in \mathscr{S} (\mathbb{R}^{4}_+)$.},
\begin{align}
    H := \lim_{\tau \downarrow 0} \frac{1}{\tau} \left( 1 - T^\tau \right).
\end{align}

Since the family of self-adjoint operators $\{ T^\tau \}_{\tau \geq 0}$ satisfies (i) $\| T^\tau \|_{op.} \leq 1$ (ii) $\{ T^\tau \}_{\tau \geq 0}$ form a semigroup (iii) this semigroup is strongly continuous (iv) $T^0 = 1$, a variant of Stone's theorem yields that the infinitesimal generator $H$ is a self-adjoint operator, e.g., \cite[Theorem VIII.8 and page 315]{Reed-Simon-vol1}. Therefore, we can define a strongly continuous one-parameter group of unitary operators on $\mathscr{K}$ generated by $H$, $\{ T^{is} := e^{-i H s} \}_{s \in \mathbb{R}}$ by Stone's theorem. $T^{is}$ corresponds to the real-time translation operator.

Finally, we define a ``holomorphic semigroup''
\begin{align}
\{ T^{\tau + is}:= T^{\tau} T^{is} ~;~ \tau > 0,~s \in \mathbb{R} \}.
\end{align}

\textit{Step 3. analytic continuation using the holomorphic semigroup generated by the ``Hamiltonian.''}

First, let us consider 
\begin{align}
(v(f), T^{is} v(g))_{\mathscr{K}},
\end{align}
which is a continuous bilinear functional on $(\bar{f}(\vartheta x),g(y)) \in \mathscr{S}_-(\mathbb{R}^4) \times \mathscr{S}_+(\mathbb{R}^4)$, where $\mathscr{S}_-(\mathbb{R}^4) := \{ f(\vartheta x) ~;~ f(x) \in             \mathscr{S}_+(\mathbb{R}^4)$  \}.
From the Schwartz nuclear theorem, we can write this as a continuous linear functional of $\Theta f^\star \times g$,
\begin{align}
(v(f), T^{is} v(g))_{\mathscr{K}} = \int dxdy~ (\Theta f^\star \times g)(x,y) \mathcal{S}_2(x,y|s),
\end{align}
where $\mathcal{S}_2(x,y|s)$ is a distribution over the space $\mathscr{S}_-(\mathbb{R}^4) \hat{\otimes} \mathscr{S}_+(\mathbb{R}^4) \simeq \{ f(x,y) \in \mathscr{S} (\mathbb{R}^{4 \cdot 2})~;~ f = 0 \mathrm{~unless~} x^4 < 0 < y^4 \}$.
Due to the translational invariance arising from $[U_s(\vec{a}), T^{is}] = 0$ and $[T^{a^4}, T^{is}] = 0$, $\mathcal{S}_2(x+a,y+a|s) = \mathcal{S}_2(x,y|s)$ for $a \in \mathbb{R}^3 \times [0,\infty)$, from which $\mathcal{S}_2(x,y|s) = S_1(y-x|s)$ with $S_1(y-x|s) \in  \mathscr{S}' (\mathbb{R}^{4}_+)$.

Note that $S_1(\xi|s)$ satisfies
\begin{align}
    S_1(\xi|0) = S_1(\xi) \label{eq:Lem14-s=0}.
\end{align}

Moreover, the unitarity of $T^{is}$ provides the upper bound on $(v(f), T^{is} v(g))_{\mathscr{K}}$ in $s$. We can thus regard $S_1(y-x|s) \in  [\mathscr{S} (\mathbb{R}^{4}_+) \hat{\otimes} \mathscr{S}(\mathbb{R})]'$.

Using $S_1(\xi|s)$, we also have
\begin{align}
&(v(f), T^{\tau + is} v(g))_{\mathscr{K}} \notag \\
&= \int dxdy~ (\Theta f^\star \times g)(x,y) S_1(\vec{y} - \vec{x}, y^4 - x^4 + \tau |s).
\end{align}
From the construction of $T^{\tau + is}$, the left-hand side is holomorphic in $\tau + is$ for $\tau > 0$. Therefore, by using the uniqueness of the Schwartz nuclear theorem, $S_1(\vec{\xi},\tau |s) $ satisfies the Cauchy-Riemann equation in the sense of a distribution.

Now, we consider one smeared in the spatial directions
\begin{align}
    S_1(\tau,s|h) := \int d^3 \vec{\xi}~ S_1(\vec{\xi},\tau|s)  h(\vec{\xi}) \in [\mathscr{S} (\mathbb{R}_+) \hat{\otimes} \mathscr{S} (\mathbb{R})]', \label{eq:Lem14-s_1-tau-s-h}
\end{align}
for $h(\vec{\xi}) \in \mathscr{S} (\mathbb{R}^3)$. The Cauchy-Riemann equation of $ S_1(\tau,s|h)$ holds for $\tau > 0$ (still in the sense of a distribution). From \cite[p. 31]{Vladimirov}, $ S_1(\tau,s|h)$ is a holomorphic \textit{function} in $\tau + is$ for $\tau > 0$.

$ S_1(\tau,s|h)$ also satisfies
\begin{enumerate}
\renewcommand{\labelenumi}{(\alph{enumi})}
    \item On the real axis $s = 0$, $S_1( \tau| h) = \int d^3 \vec{\xi}~ S_1(\vec{\xi},\tau) h(\vec{\xi})$ from ({\ref{eq:Lem14-s=0}}),
    \item $S_1 (\tau + is |h)$ can be regarded as an element of $ [\mathscr{S} (\mathbb{R}_+) \hat{\otimes} \mathscr{S} (\mathbb{R})]'$ from the definition ({\ref{eq:Lem14-s_1-tau-s-h}}),
    \item $S_1 (\tau + is |h)$ is continuous in $h(\vec{\xi}) \in \mathscr{S} (\mathbb{R}^{3})$ from $S_1(y-x|s) \in  [\mathscr{S} (\mathbb{R}^{4}_+) \hat{\otimes} \mathscr{S}(\mathbb{R})]'$.
\end{enumerate}

Hence, this holomorphic function is what we desire. This completes the proof. 
\end{proof}

Furthermore, we need the following lemma to guarantee the existence and temperedness of the boundary value:

\begin{lemma}
\label{lem:boundary_value}

Let $S(\tau + is)$ be a holomorphic function defined on the right-half plane $\tau > 0$ that can be identified with an element of $ [\mathscr{S} (\mathbb{R}_+) \hat{\otimes} \mathscr{S} (\mathbb{R})]'$, the dual space of $\mathscr{S} (\mathbb{R}_+) \hat{\otimes} \mathscr{S} (\mathbb{R})$.
Then, the boundary value of $S(\tau + is)$ at $\tau \rightarrow 0$ is a tempered distribution: $\lim_{\tau \downarrow 0} ~S(\tau + is) \in \mathscr{S}' (\mathbb{R})$. Moreover, if such a holomorphic function $S(\tau + is|h)$ is a continuous linear functional of $h$ on another function space for each $\tau >0, s \in \mathbb{R}$, then the smeared boundary value is also continuous in $h$.
\footnote{This proof is based on Lemma 8.7 in \cite{OS73} and Theorem 2-10 in \cite{Streater:1989vi}.}
\end{lemma}

\begin{proof}
We shall prove that, for any test function $f(s) \in \mathscr{S}(\mathbb{R})$, the limit
\begin{align}
\lim_{\tau \downarrow 0} \int ds ~S(\tau + is) f(s) \label{eq:lem2_boundary_value}
\end{align}
exists and is continuous in $f(s) \in \mathscr{S}(\mathbb{R})$, i.e., $\lim_{\tau \downarrow 0} \int ds ~S(\tau + is) f(s) \rightarrow 0 $ as $f \rightarrow 0$ in $\mathscr{S}(\mathbb{R})$.

Let us proceed with the following steps:
\begin{enumerate}
\renewcommand{\labelenumi}{Step \arabic{enumi}.}
\item polynomial growth in $s$ and $\tau^{-1}$,
\item a bound for $S(\tau + is)$ smeared with a test function.
\end{enumerate}

\textit{Step 1. polynomial growth in $s$ and $\tau^{-1}$.}

We show that the holomorphic function $S(\tau + is)$ grows at most polynomially in $s$ and $\tau^{-1}$.

Let $\tau_0 + i s_0$ be an arbitrary point on the right-half plane.
The mean-value property yields
\begin{align}
    S(\tau_0 + is_0) = \int_0 ^{2 \pi} \frac{d \varphi}{2 \pi} S(\tau_0 + is_0 + r e^{i \varphi}),
\end{align}
for arbitrary $0<r<\tau_0$. We may average this expression in $r$ with some weight.
Let $h(r)$ be a smooth function with $\operatorname{supp} h \subset [\frac{1}{4}, \frac{1}{2}]$ satisfying $\int_0 ^\infty dr~r h(r) = 1$.
We define $h_0(r) := \tau_0^{-2} h( \tau_0^{-1} r)$, which satisfies $\int_0 ^\infty dr~r h_0(r) = 1$ and $\operatorname{supp} h_0 \subset [\frac{\tau_0}{4}, \frac{\tau_0}{2}] (\subset [0,\tau_0))$.

Therefore, we have
\begin{align}
    S(\tau_0 + is_0) &= \int_0 ^\infty dr~r h_0(r) \int_0 ^{2 \pi} \frac{d \varphi}{2 \pi} S(\tau_0 + is_0 + r e^{i \varphi}) \notag \\
    &= \int d\tau ds ~S(\tau + is) h_0 (\sqrt{(\tau -\tau_0)^2 + (s - s_0)^2 }).
\end{align}
Since $h_0 (\sqrt{(\tau -\tau_0)^2 + (s - s_0)^2 }) \in \mathscr{S} (\mathbb{R}_+) \hat{\otimes} \mathscr{S} (\mathbb{R})$ due to the compactness of $\operatorname{supp} h_0$, there exist non-negative integers $M,N \in \mathbb{Z}_{\geq 0}$ and a constant $C > 0$ such that
\begin{align}
    |S(\tau_0 + is_0)| &\leq C \| h_0 (\sqrt{(\tau -\tau_0)^2 + (s - s_0)^2 }) \|_{M,N},
\end{align}
where $\| \cdot \|_{M,N}$ is the norm\footnote{Recall that the topology of the spaces of test functions are introduced with the family of these (semi-)norms.} defined by
\begin{align}
    &\| f(\tau,s) \|_{M,N} := \notag \\
    &~~ \sum_{\substack{k_1,k_2 \in \mathbb{Z}_{\geq 0}\\ k_1 + k_2 \leq M}}  \sum_{\substack{\ell_1,\ell_2 \in \mathbb{Z}_{\geq 0}\\ \ell_1 + \ell_2 \leq N}}
    \sup_{\tau,s} \left| \tau^{k_1} s^{k_2} \frac{\partial^{\ell_1}}{\partial \tau^{\ell_1}} \frac{\partial^{\ell_2}}{\partial s^{\ell_2}} f(\tau,s)  \right|.
\end{align}
\begin{widetext}
The bound for $|S(\tau_0 + is_0)|$ can be evaluated as
\begin{align}
    |S(\tau_0 + is_0)| &\leq C \| h_0 (\sqrt{(\tau -\tau_0)^2 + (s - s_0)^2 }) \|_{M,N} \notag \\
    &= C \tau_0^{-2} \sum_{\substack{k_1,k_2 \in \mathbb{Z}_{\geq 0}\\ k_1 + k_2 \leq M}}  \sum_{\substack{\ell_1,\ell_2 \in \mathbb{Z}_{\geq 0}\\ \ell_1 + \ell_2 \leq N}}
    \sup_{\tau,s} \left| \tau^{k_1} s^{k_2} \frac{\partial^{\ell_1}}{\partial \tau^{\ell_1}} \frac{\partial^{\ell_2}}{\partial s^{\ell_2}} h (\sqrt{(\tau/\tau_0 - 1)^2 + (s/\tau_0 - s_0/\tau_0)^2 }) \right| \notag \\
    &= C \sum_{\substack{k_1,k_2 \in \mathbb{Z}_{\geq 0}\\ k_1 + k_2 \leq M}}  \sum_{\substack{\ell_1,\ell_2 \in \mathbb{Z}_{\geq 0}\\ \ell_1 + \ell_2 \leq N}} \tau_0^{k_1 + k_2 - \ell_1 - \ell_2 - 2} \sup_{\tau',s'}
    \left| (1 + \tau')^{k_1} (s_0/\tau_0 + s')^{k_2} \frac{\partial^{\ell_1}}{\partial \tau'^{\ell_1}} \frac{\partial^{\ell_2}}{\partial s'^{\ell_2}} h (\sqrt{\tau'^2 + s'^2}) \right| \notag \\
   &= C \sum_{\substack{k_1,k_2 \in \mathbb{Z}_{\geq 0}\\ k_1 + k_2 \leq M}}  \sum_{\substack{\ell_1,\ell_2 \in \mathbb{Z}_{\geq 0}\\ \ell_1 + \ell_2 \leq N}} \sum_{m = 0}^{k_2} \frac{k_2!}{m! (k_2-m)!}
   \tau_0^{k_1 + k_2 - \ell_1 - \ell_2 -m - 2} s_0^{m} \notag \\
    &~~~~~~~~~~~~~~~~~~~~~~~~ \times \sup_{\tau',s'}
    \left| (1 + \tau')^{k_1} (s')^{k_2-m} \frac{\partial^{\ell_1}}{\partial \tau'^{\ell_1}} \frac{\partial^{\ell_2}}{\partial s'^{\ell_2}} h (\sqrt{\tau'^2 + s'^2}) \right|.
\end{align}
\end{widetext}
Note that the last term of $\sup_{\tau',s'}$ does not depend on $\tau_0,s_0$.
Hence, we finally obtain that, for $0< \tau_0 < \tau_*$ with an arbitrary fixed $\tau_*$, there exists a polynomial $P(s_0)$ and an integer $n \in \mathbb{Z}$ such that
\begin{align}
    |S(\tau_0 + is_0)| &\leq \tau_0^{-n} P(s_0). \label{eq:lem_poly_growth}
\end{align}

\textit{Step 2. a bound for smeared $S(\tau + is)$.}

So far, we have shown that $|S(\tau + is)|$ is of at most polynomial growth in $\tau^{-1}$ as $\tau \downarrow 0$. 
To prove the existence and continuity of the limit (\ref{eq:lem2_boundary_value}), we derive a stronger bound for $S(\tau + is)$ smeared by a test function.

Let us consider $S(\tau + is)$ smeared by a test function $f(s) \in \mathscr{S}(\mathbb{R})$
\begin{align}
    S(\tau;f) := \int ds~ S(\tau + is) f(s). \label{eq:lem15-step2-def}
\end{align}
By the Cauchy-Riemann equation $\frac{\partial}{\partial \tau} S(\tau + is) = - i \frac{\partial}{\partial s} S(\tau + is)$, we have
\begin{align}
 \frac{\partial^{n+1}}{\partial \tau^{n+1}}   S(\tau;f) &= \int ds~ S(\tau + is) \left( i \frac{\partial}{\partial s} \right)^{n+1} f(s) \notag \\
 &= i^{n+1} S(\tau ; \partial_s^{n+1} f),
\end{align}
and therefore, for sufficiently small $\tau$,
\begin{align}
 \left| \frac{\partial^{n+1}}{\partial \tau^{n+1}}  S(\tau ; f)  \right| &\leq \left( \int ds~ P(s) \left( \frac{\partial}{\partial s}  \right)^{n+1} f(s) \right) \tau^{-n} \notag \\
 &= C_{n,f} \tau^{-n}, \label{eq:Appendix_B_S_tau_f_estimate}
\end{align}
where $C_{n,f} > 0$ is a positive constant and we have used the result of the previous step (\ref{eq:lem_poly_growth}).
Note that $C_{n,f} \rightarrow 0$ as $f \rightarrow 0$ in $\mathscr{S}(\mathbb{R})$.

\begin{widetext}
Moreover, note that $S(\tau ; f)$ is represented by the iterative integration
\begin{align}
S(\tau ; f) &= (-1)^{n+1}\int_\tau ^{\tau_*} d \tau_1 \int_{\tau_1} ^{\tau_*} d \tau_2 \cdots \int_{\tau_n} ^{\tau_*} d \tau_{n+1} \frac{\partial^{n+1} S}{\partial \tau^{n+1}}  (\tau_{n+1} ; f) + \sum_{k = 0} ^n \frac{1}{k!} (\tau - \tau_*)^k \frac{\partial^{k} S}{\partial \tau^{k}}  (\tau_* ; f).
\end{align}
Because of the estimate ({\ref{eq:Appendix_B_S_tau_f_estimate}}), the limit $\tau \rightarrow +0$ converges.
Thus, the boundary value the boundary value ({\ref{eq:lem2_boundary_value}}) exists.
For the continuity in $\mathscr{S}(\mathbb{R})$, we obtain the bound
\begin{align}
|S(\tau;f)| &\leq C_{n,f} \int_\tau ^{\tau_*} d \tau_1 \int_{\tau_1} ^{\tau_*} d \tau_2 \cdots \int_{\tau_n} ^{\tau_*} d \tau_{n+1} \tau^{-n} + \sum_{k = 0} ^n \frac{1}{k!} |\tau - \tau_*|^k  \left| \frac{\partial^{k} S}{\partial \tau^{k}}  (\tau_*;f) \right|,
\end{align}
which implies
\begin{align}
\lim_{\tau \downarrow 0} |S(\tau;f)| &\leq C_{n,f} \int_0 ^{\tau_*} d \tau_1 \int_{\tau_1} ^{\tau_*} d \tau_2 \cdots \int_{\tau_n} ^{\tau_*} d \tau_{n+1} \tau^{-n} + \sum_{k = 0} ^n \frac{1}{k!} | \tau_*|^k  \left| \frac{\partial^{k} S}{\partial \tau^{k}}  (\tau_*;f) \right|. \label{eq:lem_bound_for_smeared}
\end{align}
\end{widetext}
Therefore, the right-hand side of (\ref{eq:lem_bound_for_smeared}) tends to vanish as $f \rightarrow 0$ in $\mathscr{S}(\mathbb{R})$, which establishes the continuity of the boundary value in $\mathscr{S}(\mathbb{R})$.
Hence, the boundary value of a given holomorphic function is a tempered distribution.

For the latter assertion, suppose the holomorphic function $S(\tau + is|h)$ is a continuous linear functional on another space of test functions $h$. Similarly to ({\ref{eq:lem15-step2-def}}), let $S(\tau|h;f)$ denote the function smeared by a test function $f(s) \in \mathscr{S}(\mathbb{R})$.
From the assumed continuity, $h \rightarrow 0$ yields $S(\tau + is|h) \rightarrow 0$ for each $\tau >0, s \in \mathbb{R}$. Thus, $C_{n,f}$ and $\left| \frac{\partial^{k} S}{\partial \tau^{k}}  (\tau_*|h;f) \right|$ tend to vanish as $h \rightarrow 0$. Therefore, the bound ({\ref{eq:lem_bound_for_smeared}}) implies that the smeared boundary value $\lim_{\tau \downarrow 0} S(\tau|h;f)$ is continuous in $h$.

This completes the proof.
\end{proof}

\begin{theorem}
\label{thm:ref_pos_yields_temperedness}
Suppose that the two-point Schwinger function $\mathcal{S}_2$ satisfies 
\begin{enumerate}
    \item temperedness $\mathcal{S}_2 \in  ~^0 \mathscr{S}' (\mathbb{R}^{4 \cdot 2})$,
    \item translational invariance $\mathcal{S}_2 (x_1 + a, x_2 + a) = \mathcal{S}_2(x_1 , x_2)$ for all $a \in \mathbb{R}^4$, and
    \item the reflection positivity for the two-point sector (\ref{eq:two_pt_ref_pos}).
\end{enumerate}
Then, the reconstructed Wightman function is a tempered distribution.
\end{theorem}

\begin{proof}
It immediately follows from the lemmas \ref{lem:OS_analytic_cont} and \ref{lem:boundary_value} that the reconstructed Wightman function is a continuous bilinear functional on $(f,h) \in \mathscr{S}(\mathbb{R}) \times \mathscr{S}(\mathbb{R}^3)$.
We obtain the reconstructed Wightman function as a tempered distribution $W_1 (\xi^0, \vec{\xi}) = \lim_{\tau \downarrow 0} S_1 (\vec{\xi},\tau|\xi^0) \in \mathscr{S}(\mathbb{R}^4)$ by the Schwartz nuclear theorem.
\end{proof}

Note that (i) the temperedness and (ii) the translational invariance are assumed in the definition of complex singularities; only (iii) the reflection positivity can be invalid. From the non-temperedness of complex singularities, we finally obtain Theorem \ref{thm:violation_ref_pos}:
\setcounter{theorem}{6}
\begin{theorem}
If $S_1 (p) = D(p^2)$ is a two-point Schwinger function with complex singularities satisfying (i) -- (v), then the reflection positivity is violated.
\end{theorem}

\section{Which axioms are violated?} \label{sec:Appendix_which_axioms}

In this section, we summarize which axioms are violated or not violated due to the existence of complex singularities.

\begin{enumerate}
\renewcommand{\labelenumi}{\Roman{enumi}.}
    \item Osterwalder-Schrader axioms for Schwinger functions \cite{OS73, OS75}:

\begin{enumerate}
\renewcommand{\labelenumii}{[OS \arabic{enumii}]}
\setcounter{enumii}{-1}
    \item Temperedness (for the two-point function)  is \textit{assumed} in the definition Sec.~\ref{sec:complex-def}.
    \item Euclidean Invariance (for the two-point function)  is \textit{assumed} in the definition Sec.~\ref{sec:complex-def}.
    \item Reflection Positivity is \textit{violated} (Theorem \ref{thm:violation_ref_pos}).
    \item Symmetry (for the two-point function)  is \textit{assumed} in the definition Sec.~\ref{sec:complex-def}.
    \item Cluster Property (for the two-point function) \textit{depends on massless singularity} (irrelevant to complex singularities).
\renewcommand{\labelenumii}{[OS \arabic{enumii}']}
\setcounter{enumii}{-1}
    \item Laplace Transform Condition is itself violated, since this requires temperedness of the Wightman function. (However, this condition is required only for reconstructing higher-point functions \cite{OS75}) 
\end{enumerate}

    \item Wightman axioms for Wightman functions \cite[Theorem 3-7]{Streater:1989vi}
\begin{enumerate}
\renewcommand{\labelenumii}{[W \arabic{enumii}]}
\setcounter{enumii}{-1}
    \item Temperedness is \textit{violated}.
    \item Poincar\'e Symmetry (for the two-point function) is \textit{valid} (for test functions in $\mathscr{D}(\mathbb{R}^4)$)
    \item Spectral Condition is \textit{violated}, since the spectral condition requires the temperedness as a prerequisite.
    \item Spacelike commutativity  (for the two-point function) is \textit{valid} (for test functions in $\mathscr{D}(\mathbb{R}^4)$).
    \item Positivity is \textit{violated} even for test functions in $\mathscr{D}(\mathbb{R}^4)$
    \item Cluster property (for the two-point function) \textit{depends on massless singularity} (irrelevant to complex singularities).
    \end{enumerate}
\end{enumerate}
Therefore, the axioms whose violations are proved are: [OS2] Reflection Positivity, ([OS0'] Laplace Transform Condition,) [W0] Temperedness, [W2] Spectral Condition, and [W4] Positivity.

\end{document}